\documentclass[a4paper,11pt]{article}
\pdfoutput=1
\usepackage{jheppub}
\usepackage[a4paper,left=2.05cm,right=2.05cm,top=2.35cm,bottom=2.35cm]{geometry}
\usepackage{amsthm}
\usepackage{booktabs}
\usepackage{array}
\usepackage{mathtools}
\usepackage{longtable}
\usepackage{pdflscape}
\usepackage{url}
\newcommand{\cO}{\mathcal O}
\newcommand{\bP}{\mathbb P}
\newcommand{\bZ}{\mathbb Z}
\newcommand{\bC}{\mathbb C}
\newcommand{\bQ}{\mathbb Q}
\newcommand{\bF}{\mathbb F}
\newcommand{\HF}{\operatorname{HF}}
\newcommand{\HS}{\operatorname{HS}}

\DeclareMathOperator{\rank}{rank}
\DeclareMathOperator{\coker}{coker}

\newtheorem{theorem}{Theorem}

\newtheorem{lemma}{Lemma}
\newtheorem{conjecture}{Conjecture}

\title{\boldmath Hilbert Functions
and Line Bundle Cohomology on CICY Threefolds}

\author{Juntao Wang}
\affiliation{Beijing Institute of Mathematical Sciences and Applications,
Beijing 101408, China}
\emailAdd{jtwang.bimsa@gmail.com}

\abstract{Line bundle cohomology on complete intersection Calabi--Yau
threefolds is an important input in string phenomenology. Previous work,
based on direct extrapolation and on machine-learning analyses of finite data
sets, has shown that these cohomology dimensions often admit chamber-wise
polynomial descriptions on the Picard lattice. In this paper we explain this
structure through the Hilbert functions associated with the non-trivial maps
in the Koszul spectral sequence. Once Bott--Borel--Weil theory determines the
non-vanishing ambient cohomology groups, the rank defects of the relevant
Koszul maps are governed by Hilbert functions of cokernel or kernel
modules. This turns many empirical chamber formulae into explicit
rank-defect statements, with proofs that are either analytic or certified on
finite boxes. Using this approach, we recover analytically almost all of the
piecewise formulae appearing in the existing literature. For the remaining
cases, we prove finite-box certificate theorems for infinite families of line
bundles in the specified regions. We also identify new wall structures on
certain CICYs which refine the chamber decompositions suggested by
finite-range data.
Finally, we use the same framework to construct line bundle cohomology formula
libraries for two CICY threefolds which had not previously been analysed in
this way. These results suggest that Hilbert-function methods can provide
promising faster inputs for future string model-building scans based on line
bundle constructions.}


\begin{document}
\maketitle
\flushbottom

\section{Introduction}

Line bundle and bundle-valued cohomology are basic inputs in string model
building. For example, in heterotic compactifications, they determine spectra,
vector-like matter, singlets and many other physics quantities. They are
needed in large-scale scans of monad bundle models
\cite{AndersonHeLukas2007,AndersonHeLukas2008,AndersonGrayHeLukas2010,HeLeeLukas2010,HeKreuzerLeeLukas2011,BuchbinderConstantinLukas2014}
and line bundle sum models on complete intersection Calabi--Yau threefolds in
products of projective spaces (CICYs)
\cite{CandelasEtAl1988,GreenHubsch1987,AndersonGrayLukasPalti2011,AndersonGrayLukasPalti2012,AndersonConstantinGrayLukasPalti2014,GrootNibbelinkLoukasRuehle2015},
toric hypersurfaces \cite{HeLeeLukasSun2014}, elliptically fibered
Calabi--Yau threefolds \cite{BraunBrodieLukas2018}, and generalized CICYs
\cite{AndersonApruzziGaoGrayLee2015,LarforsPassaroSchneider2021}.
Line-bundle cohomology also appears in related model-building problems,
including Yukawa and spectral questions and many other questions
\cite{BraunHeOvrut2013,GrayWang2019,ConstantinLeungLukasNutricati2026,Liu2022Thesis,CveticGarciaEtxebarriaHalverson2011}.

The practical importance of these computations has led to dedicated software
and search strategies.  Existing CICY computations for monad model-building
and line-bundle sum model-building use Koszul-sequence methods developed in
the early algorithmic heterotic literature cited above, while cohomCalg
implements methods for
line-bundle cohomology on toric Calabi--Yau manifolds
\cite{BlumenhagenJurkeRahnRoschy2010Algorithm,BlumenhagenJurkeRahnRoschy2012Applications}.
More recent model-building searches use reinforcement learning, genetic
algorithms, and quantum annealing to navigate large spaces of line-bundle or
monad data
\cite{LarforsSchneider2020Explore,ConstantinHarveyLukas2022,AbelConstantinHarveyLukas2022,AbelConstantinHarveyLukasNutricati2024}.
These developments make the same bottleneck increasingly visible: one must
evaluate many cohomology dimensions, often for large families of line
bundles, and the expensive part is usually the rank of the induced Koszul maps.

For CICYs, the standard calculation resolves \(\mathcal O_X(k)\) by the
ambient Koszul complex, computes ambient cohomology by the
Bott--Borel--Weil theorem and the K\"unneth formula, and then determines the
ranks of the maps induced by the defining equations
\cite{Bott1957,Anderson2008Thesis}.  The ambient dimensions are explicit
binomial expressions, but the
induced maps are often large and their ranks vary with the line-bundle vector
\(k\).  Direct rank computation is reliable, but it can become expensive in
many cases, making systematic scans very time consuming.

Several recent studies have approached this bottleneck from different
directions, and substantial evidence has been gathered that these ranks are
highly structured.  Chamber-wise formulae were found for several CICYs in
\cite{ConstantinLukas2019}, and for favourable codimension-two CICYs in
\cite{LarforsSchneider2019}.  It was also shown that line-bundle cohomology
formulae can be reconstructed from finite data by machine-learning methods
\cite{KlaewerSchlechter2019,BrodieConstantinDeenLukas2020}.  Line-bundle
cohomology under flops and across extended Kahler cones was studied in
\cite{BrodieConstantinLukas2022Flops}.  In \cite{Constantin2024},
line-bundle cohomology generating functions are constructed from
Hilbert--Poincare series associated with coordinate rings and Cox rings; see also
\cite{BrodieConstantinGrayLukasRuehle2021RecentDevelopments} for recent
developments in line-bundle cohomology and applications to string
phenomenology.  The common outcome is that line-bundle cohomology is often
governed by piecewise polynomial formulae.

Empirically, the main phenomenon is therefore the appearance of piecewise
chamber structures with Hilbert-function-like behaviour in the line-bundle
variables.  These chambers are polynomial in nature, but they are separated
by walls where additional corrections appear and where naive maximal-rank
expectations fail.  Some of these walls are already forced by the
Bott--Borel--Weil decomposition of the ambient cohomology groups, while
others are not visible from the ambient cohomology data alone.

The point of this paper is to explain this phenomenon from the perspective
of Hilbert functions.  After Bott--Borel--Weil identifies the nonzero
ambient cohomology rows, Serre duality converts top cohomology into dual
polynomial pieces.  The relevant Koszul differential then becomes a
homogeneous map of graded or multigraded modules over a polynomial ring in
the remaining variables,
\begin{equation}
  \Phi:\bigoplus_j R(-\alpha_j)
  \longrightarrow
  \bigoplus_\ell R(-\beta_\ell).
  \label{eq:intro-graded-module-map}
\end{equation}
For a fixed line bundle, the matrix appearing in the CICY computation is the
degree \(b=b(k)\) component \(\Phi_b\).  Hence, if \(M=\coker\Phi\), then
\begin{equation}
  \dim\coker(\Phi_b)=\HF(M,b),
  \qquad
  \rank(\Phi_b)
  =
  \dim\!\left(\bigoplus_\ell R_{b-\beta_\ell}\right)-\HF(M,b).
  \label{eq:intro-rank-defect-hf}
\end{equation}
The same dictionary applies to endpoint maps, where kernel Hilbert functions
appear instead of cokernel Hilbert functions.  In this way, chamber formulae
on the Picard lattice can be interpreted as Hilbert functions of remaining
modules evaluated along degrees depending affinely on \(k\).

This viewpoint explains several structural features of known formulae.  A
free resolution of the graded module determines a Hilbert series whose
numerator encodes the syzygies responsible for rank defects.  Extracting
coefficients then produces the piecewise polynomial contributions and gives
an algebraic explanation for their wall behaviour.  A key advantage is that
the kernel or cokernel is treated as a graded module, so its Hilbert series
encodes all degrees at once and therefore controls infinitely many line
bundles.  This often reveals structures which are difficult to see from
finite-box formula extrapolation alone; several such refinements appear
below.

Analytic minimal resolutions of the residual graded modules are still
difficult to obtain in many cases.  However, the Hilbert numerator can also be
certified by Gr\"obner-basis computations or by rank certificates for
finitely many matrices, without losing the cohomological information needed
here.  This gives theorem-level finite-box statements on specified regions,
often leaving some coordinates free and hence covering infinitely many line
bundles.  These computations also indicate where further analytic structure
should emerge.

From a computational perspective, a certified Hilbert function replaces large
families of Koszul rank computations.  Once the graded module is
understood, cohomology in the corresponding chamber is obtained by
coefficient extraction or closed formulae.  This provides a practical route
toward efficient line-bundle cohomology libraries for CICYs and, potentially,
for related toric Calabi--Yau threefolds and their freely acting quotients.

The results of this paper are as follows.  First, we recover analytically
almost all of the known chamber formulae treated here by identifying the
corresponding Hilbert functions in the remaining variables.  Second, the
Hilbert numerators reveal additional structures not visible from
source--target dimensions or finite-range polynomial fitting, including
hidden subwalls inside fixed ambient-cohomology regions and delayed wall
contributions beyond naive maximal-rank expectations.  Third, we obtain new
line-bundle cohomology formulae for CICY examples not previously treated in
this framework, combining analytic Hilbert-series methods with replayable
finite-box certificates.

The paper is organised as follows.  Section~\ref{sec:koszul-to-hilbert}
introduces the Koszul spectral sequence for line bundles and explains
the Hilbert-function rank-defect dictionary.  Section~\ref{sec:hf-to-cohomology}
applies this dictionary to known hypersurface and codimension-two formulae.
Section~\ref{sec:new-formulas-rank-defects} presents new wall structures and
new formula packages.  The appendices contain the technical reductions,
exactness statements, and certificate formats used in the finite-box
arguments.

\section{From CICY Koszul maps to Hilbert functions}
\label{sec:koszul-to-hilbert}

Line bundle cohomology on a CICY is computed via the Koszul resolution on the ambient space. After the Bott–Borel–Weil theorem determines which ambient cohomology groups are nonzero, the remaining task is to compute the ranks of the maps induced by the defining equations. In each degree, these maps are graded components of homogeneous maps between graded modules over the Cox ring of the ambient space. Their rank defects in a given degree are therefore given by the value of the Hilbert function of the corresponding cokernel module in that degree. This observation links the Koszul spectral sequence to the chamber-by-chamber formulae for Hilbert functions studied below.

\subsection{The Koszul spectral sequence}

Let
\begin{equation}
  A=\prod_{i=1}^m \bP^{n_i}
\end{equation}
and let \(X\subset A\) be a smooth complete intersection cut out by a
regular sequence
\begin{equation}
  f_1,\ldots,f_c,\qquad \deg(f_\alpha)=d_\alpha\in \bZ^m .
\end{equation}
\(X\) is a CICY threefold when
\[
  \dim X=\sum_i n_i-c=3,\quad
  \sum_{\alpha=1}^c d_{\alpha i}=n_i+1
  \quad\text{for each factor }\bP^{n_i}.
\]
For \(k=(k_1,\ldots,k_m)\in\bZ^m\), write
\[
  \cO_X(k)=\cO_A(k)|_X .
\]
The ambient Koszul resolution of this line bundle is
\begin{equation}
0\to
\bigwedge^c E^\vee\otimes \cO_A(k)
\to \cdots \to
E^\vee\otimes \cO_A(k)
\to \cO_A(k)\to \cO_X(k)\to 0,
\label{eq:general-koszul}
\end{equation}
where
\begin{equation}
  E=\bigoplus_{\alpha=1}^c \cO_A(d_\alpha).
\end{equation}
Equivalently, the \(p\)-th Koszul term is
\begin{equation}
K_p(k)=
\bigoplus_{\substack{I\subset\{1,\ldots,c\}\\ |I|=p}}
\cO_A\Big(k-\sum_{\alpha\in I}d_\alpha\Big),
\qquad 0\leq p\leq c.
\label{eq:kp}
\end{equation}
The differential \(K_p(k)\to K_{p-1}(k)\) is the signed contraction with
\((f_1,\ldots,f_c)\).

Taking hypercohomology gives the standard Koszul spectral sequence
\begin{equation}
E_1^{-p,q}=H^q(A,K_p(k)),
\qquad 0\le p\le c,
\label{eq:e1-tableau}
\end{equation}
converging to
\begin{equation}
E_1^{-p,q}=H^q(A,K_p(k))
\quad\Longrightarrow\quad
H^{q-p}(X,\cO_X(k)).
\label{eq:general-spectral-sequence}
\end{equation}
The first differential is induced by the Koszul differential,
\begin{equation}
d_1:E_1^{-p,q}\longrightarrow E_1^{-p+1,q}.
\label{eq:d1-definition}
\end{equation}
More generally,
\begin{equation}
d_r:E_r^{-p,q}\longrightarrow E_r^{-p+r,q-r+1},
\label{eq:dr-definition}
\end{equation}
and
\begin{equation}
E_{r+1}^{-p,q}
=
\frac{
\ker\!\left(
d_r:E_r^{-p,q}\to E_r^{-p+r,q-r+1}
\right)
}{
\operatorname{im}\!\left(
d_r:E_r^{-p-r,q+r-1}\to E_r^{-p,q}
\right)
}.
\label{eq:er-plus-one-definition}
\end{equation}
Thus the higher differentials are not extra choices: they are the connecting
maps forced by the same Koszul differential after the preceding kernels and
images have been taken.

The stable page gives a filtration of \(H^i(X,\cO_X(k))\), with graded
pieces \(E_\infty^{-p,i+p}\).  In particular,
\begin{equation}
h^i(X,\cO_X(k))
=
\sum_p \dim E_\infty^{-p,i+p}.
\label{eq:cohomology-from-einfty}
\end{equation}
Consequently, computing line bundle cohomology separates into two parts:
first determining the ambient groups \(H^q(A,K_p(k))\), and then determining
the ranks of the differentials between them.

The first part is explicit.  On a projective factor \(\bP^n\),
Bott--Borel--Weil gives
\begin{equation}
H^j(\bP^n,\cO_{\bP^n}(a))\neq 0
\end{equation}
only in the following cases:
\begin{equation}
\begin{array}{c|c|c}
\text{range of }a & \text{nonzero cohomology} & \text{dimension}\\
\hline
a\geq 0 & H^0 & \binom{a+n}{n}\\
-n\leq a\leq -1 & \text{none} & 0\\
a\leq -n-1 & H^n & \binom{-a-1}{n}.
\end{array}
\label{eq:overview-projective-space-cohomology-table}
\end{equation}
By the K\"unneth formula, the cohomology of a line bundle on \(A\) is the tensor product
of the corresponding factorwise cohomology groups.  As \(k\) varies, the
inequalities in \eqref{eq:overview-projective-space-cohomology-table} divide the Picard lattice into corresponding ambient-cohomology regions.

Inside one such region, every nonzero ambient group is a tensor product of
polynomial spaces and dual polynomial spaces.  We use
\[
H^0(\bP^n,\cO(a))\simeq \bC[x_0,\ldots,x_n]_a
\qquad (a\ge0),
\]
and, by Serre duality,
\[
H^n(\bP^n,\cO(-a-n-1))^\vee
\simeq \bC[x_0,\ldots,x_n]_a
\qquad (a\ge0).
\]
After choosing monomial bases, the maps between nonzero ambient groups
become concrete polynomial multiplication maps induced by the defining
equations \(f_\alpha\).  The dimensions of kernels and cokernels of these
maps are the nontrivial input in the cohomology calculation.

\subsection{A concrete codimension-two example with a
\texorpdfstring{\(d_2\)}{d2} map}
\label{subsec:explicit-d2-example}

We next recall a simple example in which a higher differential can be written
explicitly as a multiplication map.  Consider the codimension-two CICY
\begin{equation}
X=
\left[
\begin{array}{c|cc}
\bP^1_x & 1 & 1\\
\bP^4_y & 2 & 3
\end{array}
\right]
\subset A=\bP^1_x\times\bP^4_y .
\label{eq:sec2-d2-example-cicy}
\end{equation}
Write the two defining equations as
\begin{equation}
F=x_0 f_0(y)+x_1 f_1(y),\qquad
G=x_0 g_0(y)+x_1 g_1(y),
\label{eq:sec2-FG}
\end{equation}
where \(f_0,f_1\) have degree two and \(g_0,g_1\) have degree three in the
\(y\)-variables.  We compute the cohomology of
\[
L=\cO_X(0,b),\qquad b\ge5.
\]
The Koszul terms are
\begin{equation}
K_0=\cO_A(0,b),\qquad
K_1=\cO_A(-1,b-2)\oplus\cO_A(-1,b-3),\qquad
K_2=\cO_A(-2,b-5).
\label{eq:sec2-d2-koszul-terms}
\end{equation}
Since \(H^\bullet(\bP^1,\cO(-1))=0\), the middle column has no ambient
cohomology.  The only nonzero entries of the \(E_1\)-page are
\begin{align}
E_1^{0,0}
&=H^0(A,K_0)
\simeq H^0(\bP^4,\cO(b)),\\
E_1^{-2,1}
&=H^1(A,K_2)
\simeq H^1(\bP^1,\cO(-2))\otimes H^0(\bP^4,\cO(b-5)).
\label{eq:sec2-poly-representatives}
\end{align}
Let \(\eta\) denote the standard generator of
\(H^1(\bP^1,\cO(-2))\).  On the standard affine cover of \(\bP^1\), it may
be represented by
\begin{equation}
\eta=\frac{1}{x_0x_1}.
\label{eq:sec2-eta-representative}
\end{equation}
Multiplication by either \(x_0\) or \(x_1\) sends this class to a
coboundary in \(\cO(-1)\).  This is why the \(d_1\)-image through the
acyclic middle column vanishes on cohomology, while still leaving a
nontrivial connecting map
\begin{equation}
d_2:E_2^{-2,1}\longrightarrow E_2^{0,0}.
\label{eq:sec2-example-d2}
\end{equation}

A direct \v{C}ech computation gives, up to an overall sign,
\begin{equation}
d_2(\eta\,h)
=
(f_0g_1-f_1g_0)\,h,
\label{eq:sec2-resultant-d2}
\end{equation}
for \(h\in H^0(\bP^4,\cO(b-5))\).  Thus the higher differential is
multiplication by the degree-five polynomial
\begin{equation}
\Delta=f_0g_1-f_1g_0
\in H^0(\bP^4,\cO(5)),
\label{eq:sec2-resultant}
\end{equation}
the determinant of the \(2\times2\) coefficient matrix of \(F\) and \(G\) in
the \(\bP^1\)-variables.  Equivalently, if
\[
R=\bC[y_0,\ldots,y_4],
\]
then the relevant map is the degree-\(b\) part of
\begin{equation}
R(-5)\stackrel{\cdot\Delta}{\longrightarrow}R,
\end{equation}
namely
\begin{equation}
m_\Delta(b):
H^0(\bP^4,\cO(b-5))\longrightarrow H^0(\bP^4,\cO(b)),
\qquad
h\longmapsto \Delta h.
\label{eq:sec2-resultant-map}
\end{equation}
Since \(R\) is a domain and \(\Delta\ne0\) for a general complete
intersection, this multiplication map is injective.  Hence
\begin{equation}
H^0(X,\cO_X(0,b))
\simeq
H^0(\bP^4,\cO(b))/
\Delta H^0(\bP^4,\cO(b-5)),
\qquad
H^i(X,\cO_X(0,b))=0\quad (i>0),
\label{eq:sec2-example-cohomology}
\end{equation}
and
\begin{equation}
h^0(X,\cO_X(0,b))
=
\binom{b+4}{4}-\binom{b-1}{4}.
\label{eq:sec2-example-h0}
\end{equation}

This example illustrates the general pattern.  The spectral-sequence
calculation produces a finite-dimensional map for each value of \(b\), but
these maps are not unrelated.  They are the graded pieces of one homogeneous
map of \(R\)-modules.  If
\[
\Delta=\sum_{|\alpha|=5} c_\alpha y^\alpha,
\]
then the matrix of \(m_\Delta(b)\) has columns indexed by monomials
\(y^\mu\) with \(|\mu|=b-5\), rows indexed by monomials \(y^\nu\) with
\(|\nu|=b\), and entries
\begin{equation}
\big(m_\Delta(b)\big)_{\nu,\mu}
=
\begin{cases}
c_{\nu-\mu},&\nu-\mu\in\bZ_{\ge0}^5,\ |\nu-\mu|=5,\\
0,&\text{otherwise}.
\end{cases}
\label{eq:sec2-resultant-matrix}
\end{equation}
Thus varying \(b\) amounts to taking different graded parts of a fixed
multiplication map.

\subsection{Graded modules and rank defects}

We now formulate the preceding observation in the language used in the rest
of the paper.  Let \(R\) be the Cox ring of the ambient space.  For a single projective space this
is just
\[
R=\bC[z_0,\ldots,z_n],
\qquad \deg z_\mu=1,
\]
with \(R_d\simeq H^0(\bP^n,\cO(d))\).  For a product of projective spaces it
is multigraded:
\[
R=\bC[z_{i,\mu}\mid 1\le i\le s,\ 0\le \mu\le n_i],
\qquad
\deg z_{i,\mu}=e_i\in\bZ^s,
\]
and
\[
R_{(d_1,\ldots,d_s)}
\cong
H^0\!\left(
\bP^{n_1}\times\cdots\times\bP^{n_s},
\cO(d_1,\ldots,d_s)
\right)
\]
whenever all \(d_i\ge0\).

Inside a fixed ambient-cohomology region, the relevant Koszul differential is a degree
piece of a homogeneous map
\begin{equation}
\Phi:
F=\bigoplus_a R(-\alpha_a)
\longrightarrow
G=\bigoplus_\ell R(-\beta_\ell).
\label{eq:sec2-module-map}
\end{equation}
For a given line bundle \(k\), the actual finite-dimensional matrix is the
degree \(\mathbf b=\mathbf b(k)\) piece
\begin{equation}
\Phi_{\mathbf b}:F_{\mathbf b}\longrightarrow G_{\mathbf b}.
\label{eq:sec2-degree-piece}
\end{equation}
Here \(\mathbf b(k)\) is an affine function of the line bundle parameters in
the chamber.  If
\[
M=\coker\Phi,
\]
then \(M\) is a graded \(R\)-module and
\[
M_{\mathbf b}
=
G_{\mathbf b}/\Phi_{\mathbf b}(F_{\mathbf b})
=
\coker(\Phi_{\mathbf b}).
\]
Its Hilbert function is
\begin{equation}
\HF(M,\mathbf b)=\dim_\bC M_{\mathbf b}.
\label{eq:sec2-hf-definition}
\end{equation}
Therefore
\begin{equation}
\rank\Phi_{\mathbf b}
=
\dim G_{\mathbf b}-\HF(M,\mathbf b).
\label{eq:sec2-rank-defect}
\end{equation}
This is the rank-defect principle.  The failure of the matrix
\(\Phi_{\mathbf b}\) to have full target rank is measured by the Hilbert
function of one fixed graded module.  Kernels are obtained from the same
identity:
\begin{equation}
\dim\ker\Phi_{\mathbf b}
=
\dim F_{\mathbf b}-\dim G_{\mathbf b}+\HF(M,\mathbf b).
\label{eq:sec2-kernel-from-cokernel}
\end{equation}
Thus every kernel or cokernel contribution in the Koszul spectral sequence
can be expressed in terms of Hilbert functions of graded modules in the remaining variables.

A free resolution packages these Hilbert functions uniformly.  Suppose
\begin{equation}
0\to F_s\to F_{s-1}\to\cdots\to F_1\to F_0\to M\to0,
\qquad
F_i=\bigoplus_j R(-\gamma_{ij}).
\label{eq:sec2-free-resolution}
\end{equation}
Then additivity of dimensions in each degree gives
\begin{equation}
\HS(M;\mathbf t)
=
\HS(R;\mathbf t)
\sum_i(-1)^i\sum_j \mathbf t^{\gamma_{ij}} .
\label{eq:sec2-hs-from-resolution}
\end{equation}
The polynomial
\begin{equation}
N_M(\mathbf t)=\sum_i(-1)^i\sum_j \mathbf t^{\gamma_{ij}}
\label{eq:sec2-hilbert-numerator}
\end{equation}
is the Hilbert numerator.  Once \(N_M\) is known, the required dimensions
are obtained by coefficient extraction from \(\HS(M;\mathbf t)\).

In computations it is often unnecessary, and sometimes impractical, to
display a full minimal free resolution.  It is enough to certify the Hilbert
function or Hilbert numerator.  One standard route is to compute an initial
module: after a term order has been fixed, Gr\"obner degeneration preserves
Hilbert functions.  If \(\operatorname{in}(M)\) denotes the monomial module
defined by the leading terms of a Gr\"obner basis for the presentation of
\(M\), then
\[
\HF(M,\mathbf b)=\HF(\operatorname{in}(M),\mathbf b).
\]
This is the computational commutative-algebra form of the rank-defect
principle used in the finite-box certificates below.

The example of the previous subsection corresponds to
\[
M=R/(\Delta).
\]
Indeed,
\begin{equation}
M_b=(R/(\Delta))_b=R_b/\Delta R_{b-5},
\label{eq:sec2-example-module-piece}
\end{equation}
and hence
\begin{equation}
h^0(X,\cO_X(0,b))
=
\HF(R/(\Delta),b).
\label{eq:sec2-h0-as-hf}
\end{equation}
The free resolution
\begin{equation}
0\longrightarrow R(-5)
\stackrel{\cdot\Delta}{\longrightarrow}
R\longrightarrow R/(\Delta)\longrightarrow0
\label{eq:sec2-delta-resolution}
\end{equation}
gives
\begin{equation}
\HS(R/(\Delta);t)
=
\frac{1-t^5}{(1-t)^5}.
\label{eq:sec2-delta-hilbert-series}
\end{equation}
Extracting the coefficient of \(t^b\) recovers
\begin{equation}
\HF(R/(\Delta),b)
=
\binom{b+4}{4}-\binom{b-1}{4},
\label{eq:sec2-coefficient-extraction}
\end{equation}
with the convention that \(R_{b-5}=0\) for \(b<5\).  This is the simplest
instance of a chamber formula arising from a Hilbert function.

\subsection{Elementary rank-defect models}

We finish with three small algebraic models.  They are not meant as new
results; their role is to fix the notation and to illustrate the types of
graded modules which occur later.

\paragraph{One equation.}
Let
\[
R=\bC[z_0,\ldots,z_n],
\]
and let \(H\in R_d\) be a nonzero homogeneous polynomial.  At degree \(b\),
multiplication by \(H\) gives
\begin{equation}
\mu_b:
R_{b-d}\longrightarrow R_b,
\qquad
u\longmapsto uH .
\label{eq:sec2-scalar-degree-map}
\end{equation}
The rank defect is
\begin{equation}
D_H(b):=\dim\coker\mu_b=\dim R_b-\rank\mu_b .
\label{eq:sec2-scalar-rank-defect-direct}
\end{equation}
Since \(R\) is a domain, \(\mu_b\) is injective whenever \(R_{b-d}\ne0\).
Thus
\begin{equation}
D_H(b)=
\binom{b+n}{n}
-\binom{b-d+n}{n},
\label{eq:sec2-one-equation-defect}
\end{equation}
where the second term is omitted when \(b<d\).  Equivalently,
\[
\coker\mu_b=(R/(H))_b,
\]
and
\begin{equation}
0\longrightarrow R(-d)
\stackrel{\cdot H}{\longrightarrow}
R\longrightarrow R/(H)\longrightarrow0
\label{eq:sec2-one-equation-resolution-new}
\end{equation}
gives
\begin{equation}
\HS(R/(H);t)=\frac{1-t^d}{(1-t)^{n+1}},
\qquad
\HF(R/(H),b)
=
\binom{b+n}{n}-\binom{b-d+n}{n}.
\label{eq:sec2-one-equation-hf-new}
\end{equation}

\paragraph{A scalar complete intersection.}
Let
\[
R=\bC[x_0,x_1,x_2],
\]
and let \(A,B\in R_2\) be two generic quadrics.  Consider
\begin{equation}
\Phi_b:R_{b-2}^{\oplus2}\longrightarrow R_b,
\qquad
(u,v)\longmapsto uA+vB .
\label{eq:sec2-model-scalar-map}
\end{equation}
The image is the degree-\(b\) part of the ideal \((A,B)\), so
\[
\coker\Phi_b=(R/(A,B))_b.
\]
Since \(A,B\) form a regular sequence, the Koszul resolution is
\begin{equation}
0\longrightarrow R(-4)
\longrightarrow R(-2)^2
\stackrel{(A\ B)}{\longrightarrow} R
\longrightarrow R/(A,B)\longrightarrow0 .
\label{eq:sec2-model-scalar-resolution}
\end{equation}
Hence
\begin{equation}
\HS(R/(A,B);t)
=
\frac{(1-t^2)^2}{(1-t)^3},
\end{equation}
and
\begin{equation}
\HF(R/(A,B),b)
=
\binom{b+2}{2}
-2\binom{b}{2}
+\binom{b-2}{2}.
\label{eq:sec2-model-scalar-hf}
\end{equation}
Thus
\[
\rank\Phi_b
=
\binom{b+2}{2}-\HF(R/(A,B),b).
\]
For \(b\ge4\), the Hilbert function is the constant \(4\), the length of
the complete intersection of two quadrics in \(\bP^2\).  The appearance of
the final term in \eqref{eq:sec2-model-scalar-hf} is already a simple
example of a chamber wall in degree.

\paragraph{A vector-valued cokernel.}
Finally, active maps need not have scalar targets.  Let
\[
R=\bC[x,y,z],
\]
and consider
\begin{equation}
\Phi:R(-1)^3\longrightarrow R^2,
\qquad
(u,v,w)\longmapsto (xu+yv,\;xv+yw).
\label{eq:sec2-model-vector-map}
\end{equation}
Equivalently,
\[
\Phi=
\begin{pmatrix}
x&y&0\\
0&x&y
\end{pmatrix}.
\]
At degree \(b\) this gives
\[
\Phi_b:R_{b-1}^{\oplus3}\longrightarrow R_b^{\oplus2}.
\]
The cokernel is a genuine module \(M=\coker\Phi\), not a quotient ring
\(R/I\).  The kernel of \(\Phi\) is generated by
\[
(y^2,-xy,x^2),
\]
and hence
\begin{equation}
0\longrightarrow R(-3)
\stackrel{(y^2,-xy,x^2)}{\longrightarrow}
R(-1)^3
\stackrel{\Phi}{\longrightarrow}
R^2
\longrightarrow M
\longrightarrow0 .
\label{eq:sec2-model-vector-resolution}
\end{equation}
Therefore
\begin{equation}
\HS(M;t)=
\frac{2-3t+t^3}{(1-t)^3},
\end{equation}
so
\begin{equation}
\HF(M,b)
=
2\binom{b+2}{2}
-3\binom{b+1}{2}
+\binom{b-1}{2}.
\label{eq:sec2-model-vector-hf}
\end{equation}
Consequently,
\[
\rank\Phi_b=\dim R_b^{\oplus2}-\HF(M,b).
\]
This is the prototype for the vector-valued active maps appearing later:
the rank defect is still a Hilbert function, but the relevant object is a
cokernel module rather than a scalar quotient ring.

\section{From Hilbert functions to line bundle cohomology}
\label{sec:hf-to-cohomology}
In this section, we apply the mechanism explained in the last section to examples studied in the literature. We will recover most of their formula analytically and put those where new structure pops up in the next section.  We will first discuss the hypersurface CICYs and then the codimension two cases.

\subsection{Hypersurface formulae}
\label{subsec:sec3-hypersurface-formulae}

The hypersurface examples are the cleanest testing ground for the method in this paper,
because the Koszul resolution has at most one nontrivial multiplication map.
Nevertheless, they already display non-trivial relationships between rank phenomena and Hilbert funcitons.
We first treat the \(\bP^1\times\bP^n\) family as the model calculation;
then we discuss the bicubic and two multigraded hypersurface examples.

\subsubsection{\texorpdfstring{\(\bP^1\times\bP^n\) hypersurfaces}{P1 x Pn hypersurfaces}}
\label{subsec:sec3-p1pn}

Let
\[
        X=\{f=0\}\subset \bP^1_x\times \bP^n_y
\]
be a smooth hypersurface of bidegree \((d,e)\), and write
\[
        f=\sum_{i=0}^{d}x_0^{d-i}x_1^i f_i(y),
        \qquad
        f_i\in S_e,\qquad
        S=\bC[y_0,\ldots,y_n].
\]
This hypersurface family is one of the examples treated by
\cite{Constantin2024}. We consider the line bundle sector
\(\cO_X(-r,b)\), with \(r\geq 2\), where the Koszul map is nontrivial.  The
ambient sequence is
\[
        0\longrightarrow
        \cO_{\bP^1\times\bP^n}(-r-d,b-e)
        \xrightarrow{\ f\ }
        \cO_{\bP^1\times\bP^n}(-r,b)
        \longrightarrow
        \cO_X(-r,b)
        \longrightarrow 0
\]
On this cohomology row, the map induced by multiplication by \(f\), after
Serre duality on the \(\bP^1\)-factor, becomes a graded map over \(S\).

Choose the monomial bases
\[
        H^1(\bP^1,\cO(-r-d))\simeq H^0(\bP^1,\cO(r+d-2))^\vee,
        \qquad
        H^1(\bP^1,\cO(-r))\simeq H^0(\bP^1,\cO(r-2))^\vee .
\]
Contraction by
\(\sum x_0^{d-i}x_1^i f_i\) gives the map
\[
        \Phi_r:S(-e)^{r+d-1}\longrightarrow S^{r-1}.
\]
With respect to the ordered monomial bases, this map is represented by the matrix
\begin{equation}
  \operatorname{Syl}_{r,d}(f_0,\ldots,f_d)
  =
  \begin{pmatrix}
    f_0&f_1&\cdots&f_d&0&\cdots&0\\
    0&f_0&f_1&\cdots&f_d&\ddots&\vdots\\
    \vdots&\ddots&\ddots&\ddots&&\ddots&0\\
    0&\cdots&0&f_0&f_1&\cdots&f_d
  \end{pmatrix},
  \label{eq:p1pn-banded-sylvester}
\end{equation}
with \(r-1\) rows and \(r+d-1\) columns.  By analogy with the Sylvester
coefficient matrices used in resultant theory \cite{CoxLittleOShea1998}, we
call \eqref{eq:p1pn-banded-sylvester} the banded Sylvester matrix associated
to \(f_0+\cdots+f_dz^d\).  We define
\[
        M_r=\coker \Phi_r .
\]

When \(1\leq d\leq n\) and \(f_0,\ldots,f_d\) form a regular sequence in
\(S\), the matrix \(\Phi_r\), represented by the banded Sylvester matrix
\eqref{eq:p1pn-banded-sylvester}, satisfies the hypotheses of the
Buchsbaum--Rim complex recalled in Appendix~\ref{app:buchsbaum-rim}.  The
resulting resolution is
\begin{equation}
0\longrightarrow B_{d+1}\longrightarrow\cdots\longrightarrow B_2
\longrightarrow S(-e)^{r+d-1}\xrightarrow{\Phi_r}S^{r-1}
\longrightarrow M_r\longrightarrow0 .
\label{eq:sec3-p1pn-br-resolution}
\end{equation}
The resulting Hilbert series is
\begin{equation}
\HS(M_r;t)=
\frac{(1-t^e)^{d+1}}{(1-t)^{n+1}}
\sum_{j=0}^{r-2}(r-1-j)\binom{d+j-1}{d-1}t^{ej}.
\label{eq:sec3-p1pn-Mr}
\end{equation}
Write
\[
  P_n(m)=
  \begin{cases}
    \binom{m+n}{n},& m\ge0,\\
    0,& m<0.
  \end{cases}
\]
Under this regular-sequence hypothesis, for the remaining line bundles
\[
  L=\cO_X(-r,b),\qquad r\ge2,\quad b\ge0,
\]
the hypersurface long exact sequence gives
\begin{equation}
  h^1(X,L)=\HF(M_r,b),
  \label{eq:sec3-p1pn-linebundle-h1}
\end{equation}
and the neighbouring \(H^0\) contribution is
\begin{equation}
  h^0(X,L)_{\mathrm{res}}
  =
  (r+d-1)P_n(b-e)-(r-1)P_n(b)+\HF(M_r,b).
  \label{eq:sec3-p1pn-linebundle-h0res}
\end{equation}
The Serre-dual range is obtained from
\[
  K_X\simeq\cO_X(d-2,e-n-1),
\]
and the remaining finite \(\bP^1\)-boundary line bundles are computed
directly from the same ambient long exact sequence.  The formula above agrees
with the generating-function formula for \(\bP^1\times\bP^n\) hypersurfaces
in \cite{Constantin2024}.

The Calabi--Yau threefold
\[
X=
\left[
\begin{array}{c|c}
\bP^1&2\\
\bP^3&4
\end{array}
\right],
\]
is a special case of the preceding calculation.  Write a line bundle as
\(\cO_X(k_1,k_2)\), and set
\[
  \operatorname{ind}(k_1,k_2)=\chi(X,\cO_X(k_1,k_2))
  =\frac{1}{3}\bigl(6k_1k_2^2+k_2^3+6k_1+11k_2\bigr).
\]
Write the defining polynomial as
\begin{equation}
f=x_0^2A+x_0x_1B+x_1^2C,
\quad A,B,C\in R_4,
\quad R=\bC[y_0,y_1,y_2,y_3].
\label{eq:sec3-24-binary-quadratic}
\end{equation}
The only nontrivial wall in this example appears in the negative
\(\bP^1\)-branch.  Put
\[
  k_1=-r,\qquad k_2=b,\qquad r\ge1,\quad b>0.
\]
Our method gives
\begin{equation}
  F_r(b):=
  [t^b]\frac{(1-t^4)^3}{(1-t)^4}
  \sum_{j=0}^{r-2}(r-1-j)(j+1)t^{4j},
  \qquad
  h^1(X,\cO_X(-r,b))=F_r(b),
  \label{eq:sec3-24-Fr}
\end{equation}
where the sum is empty for \(r=1\), hence \(F_1(b)=0\).  This agrees with
the direct long-exact-sequence computation for \(\cO_X(-1,b)\).
The same sequence gives
\[
  h^0(X,\cO_X(-r,b))
  =
  \operatorname{ind}(-r,b)+F_r(b).
\]
This identity is used only in the displayed range \(b>0\).  On the boundary
\(b=0\), there is no bottom-\(\bP^3\) source row, and the sequence instead
gives \(h^0=0\) and \(h^1=r-1\), matching the \(k_2=0\) line of
\cite{ConstantinLukas2019}.
Coefficient extraction gives
\begin{equation}
F_r(b)=
\begin{cases}
-\operatorname{ind}(-r,b), & 0<b<4r,\\
r+1-\operatorname{ind}(-r,4r), & b=4r,\\
\dfrac{32}{3}r(r^2-1), & b>4r.
\end{cases}
\label{eq:sec3-24-Fr-wall}
\end{equation}
After substituting \(r=-k_1\) and \(b=k_2\), this is exactly the wall term
in \cite{ConstantinLukas2019}.  The remaining chambers either involve no
actual map or are obtained from this one by Serre duality, so we do not
repeat them here.

\subsubsection{\texorpdfstring{The bicubic \(\bP^2\times\bP^2[3,3]\)}{The bicubic P2 x P2[3,3]}}
\label{subsec:sec3-bicubic}
Let
\[
  X\subset \bP^2_x\times\bP^2_y,\qquad [X]=(3,3),
\]
and write
\begin{equation}
  f(x,y)=\sum_{|\alpha|=3}x^\alpha F_\alpha(y),
  \qquad
  F_\alpha\in S_3,\qquad
  S=\bC[y_0,y_1,y_2].
  \label{eq:bicubic-expansion}
\end{equation}
Consider the chamber
\[
  L=\cO_X(-r,b),\qquad r\geq3,\quad b\geq0 .
\]
The hypersurface sequence is
\[
0\to
\cO_A(-r-3,b-3)
\stackrel{f}{\longrightarrow}
\cO_A(-r,b)
\to L\to0 .
\]
Both \(\bP^2_x\)-degrees, \(-r-3\) and \(-r\), are in top cohomology.
On the \(\bP^2_y\)-factor, only bottom cohomology contributes in the target;
the shifted source lies in the bottom row for \(b\ge3\).  With the convention
\(S_{b-3}=0\) for \(b<3\), the
bottom-\(\bP^2_y\) part of the long exact sequence is
\begin{equation}
0\to H^1(X,L)\to
H^2(\bP^2_x,\cO(-r-3))\otimes S_{b-3}
\stackrel{\mu_{r,b}}{\longrightarrow}
H^2(\bP^2_x,\cO(-r))\otimes S_b
\to H^2(X,L)\to0 .
\label{eq:bicubic-short-row}
\end{equation}
Thus
\[
  h^1(X,L)=\dim\ker\mu_{r,b},
  \qquad
  h^2(X,L)=\dim\coker\mu_{r,b}.
\]

After Serre duality on the first \(\bP^2\), the map \(\mu_{r,b}\) is the
degree-\(b\) part of
\begin{equation}
  \Phi_r:S(-3)^{\binom{r+2}{2}}
  \longrightarrow
  S^{\binom{r-1}{2}} .
  \label{eq:bicubic-map}
\end{equation}
Rows are indexed by monomials \(x^\beta\) with \(|\beta|=r-3\), columns by
monomials \(x^\gamma\) with \(|\gamma|=r\), and
\begin{equation}
  (\Phi_r)_{\beta,\gamma}
  =
  \begin{cases}
  F_{\gamma-\beta},& \gamma_i\geq\beta_i\ \text{for all }i,\\
  0,&\text{otherwise}.
  \end{cases}
  \label{eq:bicubic-matrix-entry}
\end{equation}
For example, when \(r=4\), the rows are \(x_0,x_1,x_2\), the columns are
the degree-four monomials in \(x\), and
\[
\begin{pmatrix}
F_{300}&F_{210}&F_{201}&F_{120}&F_{111}&F_{102}&F_{030}&F_{021}&F_{012}&F_{003}&0&0&0&0&0\\
0&F_{300}&0&F_{210}&F_{201}&0&F_{120}&F_{111}&F_{102}&0&F_{030}&F_{021}&F_{012}&F_{003}&0\\
0&0&F_{300}&0&F_{210}&F_{201}&0&F_{120}&F_{111}&F_{102}&0&F_{030}&F_{021}&F_{012}&F_{003}
\end{pmatrix}.
\]
This finite graded convolution matrix controls the chamber formula: its
degree-\(b\) kernel gives \(h^1\), and its degree-\(b\) cokernel gives
\(h^2\).

Set \(M_r=\coker\Phi_r\).  We compute the Hilbert function of \(M_r\) by
showing that one critical degree map is an isomorphism.  For each fixed
\(r\), let \(\Delta_r(f)\) be the determinant of
\[
  (\Phi_r)_r:
  S_{r-3}^{\binom{r+2}{2}}
  \longrightarrow
  S_r^{\binom{r-1}{2}} .
\]
The source and target have the same dimension.  To prove that
\(\Delta_r\) is not the zero polynomial in the bicubic coefficients, it is
enough to evaluate it at one convenient coefficient point.  We choose the
special, singular bicubic
\[
  f_0=(x_0y_0+x_1y_1+x_2y_2)^3
  =
  \sum_{|\alpha|=3}\binom{3}{\alpha}x^\alpha y^\alpha ,
\]
so \(F_\alpha(y)=\binom{3}{\alpha}y^\alpha\).  The singularity of this
special member is irrelevant here: it is used only as a nonvanishing
certificate for the determinant polynomial.

On the polynomial ring in the \(x\)- and \(y\)-variables define
\[
  e=\sum_{i=0}^2 x_i\frac{\partial}{\partial y_i},
  \qquad
  \mathsf f=\sum_{i=0}^2 y_i\frac{\partial}{\partial x_i},
  \qquad
  h=\sum_i x_i\frac{\partial}{\partial x_i}
    -\sum_i y_i\frac{\partial}{\partial y_i}.
\]
These operators satisfy the \(\mathfrak{sl}_2\) relations
\[
  [h,e]=2e,\qquad [h,\mathsf f]=-2\mathsf f,\qquad [e,\mathsf f]=h .
\]
Using divided powers \(X^\gamma=x^\gamma/\gamma!\) in the \(x\)-variables,
\[
  \mathsf f^3(X^\gamma y^\delta)
  =
  \sum_{\substack{\alpha\leq\gamma\\ |\alpha|=3}}
  \binom{3}{\alpha}X^{\gamma-\alpha}y^{\delta+\alpha}.
\]
This is exactly the degree-\(r\) matrix \((\Phi_r)_r\) for \(f_0\), up to
diagonal rescaling of bases.  Indeed \(|\gamma|=r\), \(|\delta|=r-3\),
\(\beta=\gamma-\alpha\), and the target monomial has degrees
\(|\beta|=r-3\), \(|\delta+\alpha|=r\).

Inside the total-degree \(2r-3\) subspace, the source
\(\operatorname{Sym}^r_x\otimes\operatorname{Sym}^{r-3}_y\) is the
\(h\)-weight \(3\) subspace, and the target
\(\operatorname{Sym}^{r-3}_x\otimes\operatorname{Sym}^{r}_y\) is the
\(h\)-weight \(-3\) subspace.  Over \(\bC\), in any finite-dimensional
\(\mathfrak{sl}_2\)-representation, \(\mathsf f^3\) is an isomorphism between
the weight-\(3\) and weight-\(-3\) spaces in each irreducible summand: every
irreducible summand which contains weight \(3\) also contains weight \(-3\),
and the three lowering steps between them have nonzero coefficients.
Therefore
\[
  \Delta_r(f_0)\neq0 .
\]
It follows that \(\Delta_r\) is not identically zero.  For each fixed
\(r\), the condition \(\Delta_r(f)\neq0\) defines a nonempty Zariski open
subset of the bicubic coefficient space.  Intersecting with the smooth
locus, for this fixed \(r\), a general smooth bicubic has \((\Phi_r)_r\)
an isomorphism.  For a statement uniform in all \(r\ge3\), the correct
statement is that it holds for a very general smooth bicubic.

Assume now that \((\Phi_r)_r\) is an isomorphism.  If \(b<r\) and
\(v\in\ker(\Phi_r)_b\), choose a nonzero \(u\in S_{r-b}\).  Since
\(\Phi_r\) is an \(S\)-module map,
\[
  (\Phi_r)_r(uv)=u(\Phi_r)_b(v)=0 .
\]
The degree-\(r\) map is injective and the source is free, so \(uv=0\);
hence \(v=0\).  Thus \((\Phi_r)_b\) is injective for \(b<r\).  Since
\(M_r\) is generated by the degree-zero generators coming from the target and
\((M_r)_r=0\), all higher pieces vanish:
\[
  (M_r)_b=S_{b-r}(M_r)_r=0,\qquad b\ge r .
\]
Thus \((\Phi_r)_b\) is surjective for \(b\ge r\).

Consequently \(M_r\) is a finite-length graded \(S\)-module, generated in
degree \(0\), and its last possible nonzero degree is \(r-1\).  Its Hilbert
function is read directly from the injective/surjective dichotomy above.
Throughout the coefficient formulae below, terms of the form
\(\binom{m+2}{2}\) coming from \(S_m\) are interpreted as \(0\) for \(m<0\);
equivalently, \(\binom{b-1}{2}=0\) for \(b<3\), and similarly for the shifted
binomials below.
For \(b<r\), the map \((\Phi_r)_b\) is injective, so
\[
  \HF(M_r,b)
  =
  \binom{r-1}{2}\binom{b+2}{2}
  -
  \binom{r+2}{2}\binom{b-1}{2}.
\]
For \(b\geq r\), the map is surjective, so \(\HF(M_r,b)=0\).  Equivalently,
\[
  \sum_{b\geq0}\HF(M_r,b)t^b
  =
  \sum_{b=0}^{r-1}
  \left[
  \binom{r-1}{2}\binom{b+2}{2}
  -
  \binom{r+2}{2}\binom{b-1}{2}
  \right]t^b .
\]
Summing this finite Hilbert function gives
\begin{equation}
  \HS(M_r;t)=
  \frac{
  \binom{r-1}{2}
  -\binom{r+2}{2}t^3
  +\frac32(r+2)(r-1)t^{r+1}
  -\frac32(r+1)(r-2)t^{r+2}}
  {(1-t)^3}.
  \label{eq:bicubic-hs}
\end{equation}
Extracting coefficients from \eqref{eq:bicubic-hs}, using
\[
  \frac{1}{(1-t)^3}=\sum_{m\ge0}\binom{m+2}{2}t^m,
\]
gives
\begin{align}
  D_r(b):=\HF(M_r,b)
  &=
  \binom{r-1}{2}\binom{b+2}{2}
  -\binom{r+2}{2}\binom{b-1}{2}
  \nonumber\\
  &\quad
  +\frac32(r+2)(r-1)\binom{b-r+1}{2}
  -\frac32(r+1)(r-2)\binom{b-r}{2},
  \label{eq:bicubic-defect}
\end{align}
where the same zero convention applies to each shifted binomial.  The long
exact sequence
\eqref{eq:bicubic-short-row} gives
\begin{align}
  h^2(X,\cO_X(-r,b))&=D_r(b),\label{eq:bicubic-h2}\\
  h^1(X,\cO_X(-r,b))&=
  \binom{r+2}{2}\binom{b-1}{2}
  -\binom{r-1}{2}\binom{b+2}{2}
  +D_r(b).
  \label{eq:bicubic-h1}
\end{align}
Equivalently, in the interior range,
\[
  3\le b<r:\quad
  h^1=0,\qquad
  h^2=\frac32(r-b)(rb-2),
\]
while for \(b\ge r\),
\[
  h^2=0,\qquad
  h^1=\frac32(b-r)(rb-2).
\]
The values \(b=1,2\) are already included by the convention
\(S_{b-3}=0\).  At \(b=0\), the same formula gives
\(D_r(0)=\binom{r-1}{2}\) in \(h^2\), while the omitted top
\(\bP^2_y\)-row contributes
\[
  h^3(X,\cO_X(-r,0))=\binom{r+2}{2}.
\]
Together with the effective chamber, the finite \(r=1,2\) boundary cases,
computed directly from the same hypersurface sequence, the \(x/y\)-symmetric
chamber, and Serre duality, this recovers the bicubic formula of
\cite{ConstantinLukas2019} for a very general bicubic.

\subsubsection{Two multigraded hypersurface examples}
\label{subsec:sec3-multigraded-hypersurface-examples}

For the next two hypersurfaces we analyze only the chambers which have clean
analytic descriptions.  The chambers with new structure will be discussed in
the next section.

\paragraph*{The Picard-number-three hypersurface
\texorpdfstring{\(\bP^1\times\bP^1\times\bP^2[2,2,3]\)}{P1 x P1 x P2[2,2,3]}.}
\label{subsec:sec3-p1p1p2-223}

Write the defining equation of the hypersurface
\[
  X\subset \bP^1_x\times\bP^1_y\times\bP^2_z,\qquad [X]=(2,2,3),
\] as a quadratic in the
first \(\bP^1\):
\[
  f=x_0^2A+x_0x_1B+x_1^2C,\qquad
  A,B,C\in H^0(\bP^1_y\times\bP^2_z,\cO(2,3)).
\]
For
\[
  L=\cO_X(-a,b,c),\qquad a\ge2,
\]
the row contributing to \(H^0(X,L)\), after Serre duality on
\(\bP^1_x\), is the bidegree \((b,c)\) part of the map
\[
  \Phi_a:R(-(2,3))^{a+1}\longrightarrow R^{a-1},
  \qquad
  R=\bC[y_0,y_1,z_0,z_1,z_2],
\]
whose rows are shifted copies of \((A,B,C)\).  This finite matrix is the
degree-\((b,c)\) part of a single map on the factor
\[
  B=\bP^1_y\times\bP^2_z
\]
namely
\[
  \cO_B(-2,-3)^{a+1}\longrightarrow \cO_B^{a-1},
  \qquad
  (s_0,\ldots,s_a)\longmapsto
  (As_i+Bs_{i+1}+Cs_{i+2})_{i=0}^{a-2}.
\]
Let \(K_a\) be the kernel of this map:
\[
0\to
K_a\to
\cO_B(-2,-3)^{a+1}
\to
\cO_B^{a-1}.
\]
After twisting by \(\cO_B(b,c)\) and taking \(H^0(B,-)\), this gives exactly
the kernel of the finite degree-\((b,c)\) Toeplitz matrix above.  Assume
that the common zero locus of a general triple \((A,B,C)\) has codimension
\(3\) in \(B\).  Under this generality assumption, \(K_a\) is resolved by
the Buchsbaum--Rim complex of this band matrix.  Indeed, after evaluating at
a point of \(B\), the rows are shifted copies of the scalar triple
\((A(p),B(p),C(p))\).  They are linearly independent unless
\(A(p)=B(p)=C(p)=0\), and codimension \(3\) is the Buchsbaum--Rim exactness
condition for an \((a-1)\times(a+1)\) matrix; see
Appendix~\ref{app:buchsbaum-rim}.  Hence
\[
0\to
\cO_B(-2a-2,-3a-3)^{a-1}
\to
\cO_B(-2a,-3a)^{a+1}
\to K_a\to0 .
\]
Thus \(h^0(X,L)\) in the negative-\(\bP^1_x\) chamber is
\[
  h^0(X,\cO_X(-a,b,c))=h^0(B,K_a(b,c)).
\]
Twist the displayed resolution by \(\cO_B(b,c)\).  With
\[
  u=b-2a,\qquad v=c-3a,
\]
we get
\[
0\to
\cO_B(u-2,v-3)^{a-1}
\to
\cO_B(u,v)^{a+1}
\to
K_a(b,c)\to0 .
\]
For this branch, set
\[
  E_a(u,v)=
  \ker\!\left[
  H^1\!\left(B,\cO_B(u-2,v-3)^{a-1}\right)
  \longrightarrow
  H^1\!\left(B,\cO_B(u,v)^{a+1}\right)
  \right],
\]
where the map is induced by the first arrow of the displayed sequence.  The
consequences are:
\[
\begin{array}{c|c|c}
\text{range} & h^0(X,\cO_X(-a,b,c)) & \text{source of the value}\\
\hline
c<3a
& 0
& H^0(B,\cO_B(u,v))=0
\\[2mm]
c=3a
& (a+1)\max\{b-2a+1,0\}
& H^0(B,\cO_B(u,0))^{a+1}
\\[2mm]
c>3a,\ b\ge 2a-1
& \chi(\cO_X(-a,b,c))+9a(a^2-1)
& \text{Euler characteristic}
\\[2mm]
c>3a,\ b<2a-1
& \dim E_a(b-2a,c-3a)
& \text{remaining }H^1\text{-kernel}
\end{array}
\]
Thus the first three lines are read directly from the displayed sequence,
while the last line is the kernel branch.  We return to that branch
in Subsection~\ref{subsec:new-7880-negative-p1-box}, where structured
initial-module certificates give a finite-box formula for fixed
\((a,b)\)-ranges and arbitrary \(c\).
In the notation \(k_1=-a\), \(k_2=b\), \(k_3=c\), the wall line
\(c=3a\) and the stable line \(c>3a,\ b\ge2a-1\) are exactly the
negative-\(k_1\) wall and stable pieces of the \(k_3>0\) formula in
\cite{ConstantinLukas2019}.  The final \(E_a\)-line is the complementary remaining
kernel range; here we keep the reduction rather than replacing it by a
closed chamber polynomial.

\paragraph*{The tetraquadric outside the two-negative chamber.}
\label{subsec:sec3-tetra-outside-two-negative}

The tetraquadric
\[
  X\subset(\bP^1)^4,\qquad [X]=(2,2,2,2),
\]
has defining equation \(f\in H^0((\bP^1)^4,\cO(2,2,2,2))\).  For
\(L_k=\cO_X(k_1,k_2,k_3,k_4)\), the hypersurface sequence is
\[
0\to
\cO_A(k_1-2,k_2-2,k_3-2,k_4-2)
\xrightarrow{\cdot f}
\cO_A(k_1,k_2,k_3,k_4)\to L_k\to0.
\]
Here we discuss the one-negative case with \(k_1\le -2\) and
\(k_2,k_3,k_4\ge2\).  The stable two-negative chamber is treated in
Section~\ref{subsec:new-7862}.  The other stable one-negative chambers are
obtained by permuting the four \(\bP^1\)-factors, while the map-free
boundary cases are handled separately.  Take \(k_1=-r\), \(r\ge2\), and write
\[
  f=x_{10}^2A+x_{10}x_{11}B+x_{11}^2C,\qquad
  A,B,C\in R_{(2,2,2)}
\]
for the coordinate ring \(R\) of the remaining three \(\bP^1\)'s.  The
active map is
\[
  \Phi_r:R(-(2,2,2))^{r+1}\longrightarrow R^{r-1}.
\]
Explicitly,
\[
  \Phi_r(u_0,\ldots,u_r)
  =
  (Au_0+Bu_1+Cu_2,\ldots,
  Au_{r-2}+Bu_{r-1}+Cu_r).
\]
Let \(\delta=(2,2,2)\), and let
\[
  M_r=\operatorname{coker}\big(R(-\delta)^{r+1}\to R^{r-1}\big)
\]
be the graded cokernel, and let \(\mathcal M_r\) be the corresponding
cokernel sheaf on \(\mathcal B=(\bP^1)^3\).  In degree
\((k_2,k_3,k_4)\), the source
and target of the active
long-exact-sequence map have dimensions
\[
  (r+1)\prod_{j=2}^4(k_j-1),
  \qquad
  (r-1)\prod_{j=2}^4(k_j+1),
\]
so the degreewise cokernel of \(\Phi_r\) is exactly the correction term added
to the Euler difference.  Write
\[
  \chi_r(k_2,k_3,k_4)
  =
  (r+1)\prod_{j=2}^4(k_j-1)
  -(r-1)\prod_{j=2}^4(k_j+1).
\]
It remains to compute the degreewise cokernel.  For a general triple
\(A,B,C\), the common zero set
\[
  \{A=B=C=0\}\subset \mathcal B
\]
has codimension \(3\).  Away from this set the band matrix has full row rank,
so the Buchsbaum--Rim criterion applies.  Hence \(\mathcal M_r\) has the
exact Buchsbaum--Rim complex
\[
0\to
\cO_{\mathcal B}(-(r+1)\delta)^{r-1}
\to
\cO_{\mathcal B}(-r\delta)^{r+1}
\to
\cO_{\mathcal B}(-\delta)^{r+1}
\to
\cO_{\mathcal B}^{r-1}
\to
\mathcal M_r
\to0 .
\]
Taking the alternating sum of the Hilbert series of the vector-bundle terms
in this complex gives
\[
  \mathcal H_r(t_2,t_3,t_4)
  =
  \frac{
  (r-1)
  -(r+1)t^\delta
  +(r+1)t^{r\delta}
  -(r-1)t^{(r+1)\delta}}
  {(1-t_2)^2(1-t_3)^2(1-t_4)^2},
  \qquad
  t^\delta=t_2^2t_3^2t_4^2 .
\]
For a fixed degree \((k_2,k_3,k_4)\), we shift the four terms above by
\((k_2,k_3,k_4)\) and evaluate each \(\bP^1\)-factor by Bott's formula.
This gives the following function.
After permuting the last three factors, assume \(2\le k_2\le k_3\le k_4\).
The degreewise cokernel is
\[
F_r(k_2,k_3,k_4)=
\begin{cases}
(r-1)(8r^2+6r+k_4-1),
& k_2=k_3=2r,\ k_4\ge2r,\\
8r(r^2-1),
& k_2\ge2r-1,\ k_3>2r,\\
\max\{0,-\chi_r(k_2,k_3,k_4)\},
& \text{otherwise}.
\end{cases}
\]
Thus
\[
\begin{aligned}
  h^1(X,\cO_X(-r,k_2,k_3,k_4))
  &=F_r(k_2,k_3,k_4),\\
  h^0(X,\cO_X(-r,k_2,k_3,k_4))
  &=(r+1)\prod_{j=2}^4(k_j-1)
    -(r-1)\prod_{j=2}^4(k_j+1)
    +F_r(k_2,k_3,k_4).
\end{aligned}
\]
Moreover \(h^2=h^3=0\) in this stable one-negative range.
This is exactly the stable one-negative tetraquadric formula of
\cite{ConstantinLukas2019}, after permuting the four \(\bP^1\)-factors if
necessary: their \(h^1\) is the same \(F_r\), and their \(h^0\) is the same
ambient Euler difference plus \(F_r\).  The positive-part wall formula in
Ref.~\cite{ConstantinLukas2019} belongs to the two-negative chamber, not to
this one-negative range.

The boundary \(k_1=0\) is also completely explicit.  For \(b,c,d\ge0\),
\[
 h^0(X,\cO_X(0,b,c,d))
 =
 (b+1)(c+1)(d+1)
 +
 u(b-2)u(c-2)u(d-2),
 \qquad
 u(m)=\max\{m+1,0\}.
\]
This is the \(h^0\) part of the boundary formulae.  The boundary
higher-cohomology terms come from the same ambient cohomology check; for example
\[
  h^1(X,\cO_X(0,0,c,d))=(c-1)(d-1)\qquad(c,d>0).
\]
Permuting the four factors gives the other one-negative and boundary
pieces.  Thus, up to permutations and Serre duality, the stable two-negative
chamber is the only remaining tetraquadric chamber with a nontrivial
active map.

\subsection{Codimension-two examples}
\label{subsec:sec3-codim-two-examples}
In this section, we study the examples in \cite{LarforsSchneider2019} and we can see our method can recover all of their formulae analytically.
\subsubsection{\texorpdfstring{The quadratic \(\bP^1\times\bP^4\) family}{The quadratic P1 x P4 family}}
\label{subsec:sec3-7806}

The first several codimension-two cases include the following three configurations
\[
  X_{\alpha,\beta}=
  \left[
  \begin{array}{c|cc}
  \bP^1&0&2\\
  \bP^4&\alpha&\beta
  \end{array}
  \right],
  \qquad
  (\alpha,\beta)=(3,2),(2,3),(4,1),
\]
which are CICYs \(7806,7882,7888\).  Write
\[
  g\in S_\alpha,\qquad
  f=x_0^2A+x_0x_1B+x_1^2C,\qquad
  A,B,C\in S_\beta,
  \qquad S=\bC[y_0,\ldots,y_4].
\]
We first consider the one-negative chamber, \(L=\cO_X(-r,m)\), \(r\ge1\), \(m>0\).  Here the active
map is the degree-\(m\) part of
\[
  \Phi_r:
  S(-\alpha)^{r-1}\oplus S(-\beta)^{r+1}
  \longrightarrow S^{r-1},
\]
with components
\[
  (\Phi_r(u,v))_i=g\,u_i+A\,v_i+B\,v_{i+1}+C\,v_{i+2},
  \qquad 0\le i\le r-2.
\]
Let \(M_r^{(\alpha,\beta)}=\coker\Phi_r\).  The cohomology in this
chamber is
\[
  h^1(X_{\alpha,\beta},\cO(-r,m))
  =\HF(M_r^{(\alpha,\beta)},m),
  \qquad
  h^0=\chi+\HF(M_r^{(\alpha,\beta)},m),
\]
and \(h^q=0\) for \(q\ge2\).  The proof is the same row-homology argument
as for hypersurfaces; the leftmost Koszul homology has the wrong total
degree and does not contribute to sheaf cohomology.
When \(r=1\), the target \(S^{r-1}\) is zero; we use the convention
\(M_1^{(\alpha,\beta)}=0\), so \(h^1=0\) in this \(m>0\) chamber.

The Hilbert series is analytic on the open locus where
\[
  g,A,B,C
\]
form a regular sequence.  Removing the \(g\)-block first gives the
binary-quadratic Buchsbaum--Rim module
\[
  T_r:S(-\beta)^{r+1}\to S^{r-1}.
\]
Since \(A,B,C\) are a regular sequence, its cokernel has Hilbert series
given by the \((d,e,n)=(2,\beta,4)\) case of the hypersurface model.  The
extra equation \(g\) is then a nonzerodivisor on that
module, so
\[
  \HS(M_r^{(\alpha,\beta)};t)
  =
  \frac{(1-t^\alpha)(1-t^\beta)^3}{(1-t)^5}
  \sum_{j=0}^{r-2}(j+1)(r-1-j)t^{\beta j}.
\]
One can use this Hilbert function for the remaining variables to explain all three chamber
formulae:
\[
\begin{array}{c|c|c|c}
\text{CICY} & (\alpha,\beta) & \text{wall} & \text{stable }h^1\\
\hline
7806 & (3,2) & m=2r & 4r(r^2-1)\\
7882 & (2,3) & m=3r & 9r(r^2-1)\\
7888 & (4,1) & m=r & \dfrac{2}{3}r(r^2-1).
\end{array}
\]
The stable entry is the length contribution
\[
  \alpha\beta^3
  \sum_{j=0}^{r-2}(j+1)(r-1-j)
  =
  \frac{\alpha\beta^3}{6}r(r^2-1),
\]
which gives the three values in the table.
For \(0<m<\beta r\), the Hilbert function for the remaining variables cancels the Euler
characteristic:
\[
  h^0=0,\qquad h^1=-\chi .
\]
At the wall \(m=\beta r\), the value of \(h^1\) is the stable value in the
table reduced by \(r-1\).  For \(m\ge \beta r+1\), the correction Hilbert
function has stabilized to the length contribution displayed in the table.
The horizontal boundary \(m=0\) is not part of the \(m>0\) chamber above and direct calculation gives
\[
  h^0=0,\qquad h^1=r-1,\qquad h^2=0,\qquad h^3=r+1 .
\]
The half-plane \(m<0\) follows by Serre duality.

\subsubsection{\texorpdfstring{The other two \(\bP^1\times\bP^4\) cases}{The other two P1 x P4 cases}}
\label{subsec:sec3-larfors-five-p1p4}

The two remaining \(\bP^1\times\bP^4\) cases of
\cite{LarforsSchneider2019} are
\[
\begin{array}{c@{\qquad\qquad}c}
7858:\ 
\left[
\begin{array}{c|cc}
\bP^1&1&1\\
\bP^4&3&2
\end{array}\right]
&
7885:\ 
\left[
\begin{array}{c|cc}
\bP^1&1&1\\
\bP^4&4&1
\end{array}\right].
\end{array}
\]
We write them uniformly as \((\alpha,\beta)=(3,2)\) and
\((\alpha,\beta)=(4,1)\), respectively.  Let
\[
  R=\bC[y_0,\ldots,y_4],
\]
and write the two defining equations as:
\[
  p=x_0A_0+x_1A_1,\qquad
  q=x_0B_0+x_1B_1,
  \qquad
  A_i\in R_\alpha,\quad B_i\in R_\beta .
\]
We work on the open locus where \(X\) is smooth and
\(A_0,A_1,B_0,B_1\) form a homogeneous regular sequence in \(R\).

Consider the one-negative line bundle
\[
  L=\cO_X(-r,m),\qquad r\ge2,\qquad m>0 .
\]
The case \(r=1\) is a low-rank boundary case: the target \(R^{r-1}\) in
the band module below vanishes, so it is not included in this uniform
one-negative calculation.
The Koszul sequence on \(A=\bP^1\times\bP^4\) is
\[
0\to
\cO_A(-r-2,m-\alpha-\beta)
\to
\cO_A(-r-1,m-\alpha)\oplus\cO_A(-r-1,m-\beta)
\to
\cO_A(-r,m)
\to
L\to0 .
\]
Taking the active \(\bP^1\) top-cohomology row turns multiplication by
\(p\) and \(q\) into the degree-\(m\) part of
\[
  \Phi_{r;\alpha,\beta}:
  R(-\alpha)^r\oplus R(-\beta)^r
  \longrightarrow R^{r-1}.
\]
For example, for \(r=4\),
\[
  \Phi_{4;\alpha,\beta}=
  \begin{pmatrix}
    A_0&A_1&0&0&B_0&B_1&0&0\\
    0&A_0&A_1&0&0&B_0&B_1&0\\
    0&0&A_0&A_1&0&0&B_0&B_1
  \end{pmatrix}.
\]
Thus the correction term is the Hilbert function of the cokernel
\[
  N_r^{(\alpha,\beta)}=\coker\Phi_{r;\alpha,\beta}.
\]

The resolution input is the following filtered form, proved in
Appendix~\ref{app:two-linear-band-module}.  The associated graded module is
a direct sum of shifted copies of the complete-intersection quotient:
\[
  \operatorname{gr}N_r^{(\alpha,\beta)}
  \simeq
  \bigoplus_{\substack{i,j\ge0\\ i+j\le r-2}}
  \left(R/(A_0,A_1,B_0,B_1)\right)(-\alpha i-\beta j)
  ^{\oplus(r-1-i-j)} .
\]
Since \(A_0,A_1,B_0,B_1\) is a regular sequence,
\[
  \HS(R/(A_0,A_1,B_0,B_1);t)
  =
  \frac{(1-t^\alpha)^2(1-t^\beta)^2}{(1-t)^5}.
\]
Therefore
\begin{equation}
  \HS(N_r^{(\alpha,\beta)};t)
  =
  \frac{(1-t^\alpha)^2(1-t^\beta)^2}{(1-t)^5}
  \sum_{\substack{i,j\ge0\\ i+j\le r-2}}
  (r-1-i-j)t^{\alpha i+\beta j}.
\label{eq:ls-linear-HS}
\end{equation}
Equivalently, if
\[
  q_{\alpha,\beta}(s)
  =
  [t^s]\frac{(1-t^\alpha)^2(1-t^\beta)^2}{(1-t)^5},
  \qquad q_{\alpha,\beta}(s)=0\quad(s<0),
\]
then
\[
  \HF(N_r^{(\alpha,\beta)},m)
  =
  \sum_{\substack{i,j\ge0\\ i+j\le r-2}}
  (r-1-i-j)\,q_{\alpha,\beta}(m-\alpha i-\beta j).
\]
The cohomology in this chamber is
\[
  h^1(X,\cO_X(-r,m))=\HF(N_r^{(\alpha,\beta)},m),
  \qquad
  h^0=\chi(X,\cO_X(-r,m))+\HF(N_r^{(\alpha,\beta)},m),
\]
and \(h^q=0\) for \(q\ge2\).

Let us compare this directly with the formulae of
\cite{LarforsSchneider2019}.  For \(7858\), put
\[
  \chi_{3,2}(-r,m)
  =
  -3rm^2-2r+\frac{5}{6}m^3+\frac{25}{6}m .
\]
The corresponding one-negative formula in \cite{LarforsSchneider2019} for
\(h^1\) is
\begin{equation}
h^1(X,\cO_X(-r,m))=
\begin{cases}
0, & m<0,\\
r-1, & m=0,\\
-\chi_{3,2}(-r,m), & 0<m<2r,\\
-\chi_{3,2}(-r,m)+1, & m=2r,\\
-\chi_{3,2}(-r,m)
+\dfrac32(m-2r)\bigl(3+(m-2r)^2\bigr),
& 2r<m<3r,\\
6r(r^2-1), & m\ge3r .
\end{cases}
\label{eq:sec3-ls-7858-piecewise}
\end{equation}
Our expression gives the same function because
\[
  \HF(N_r^{(3,2)},m)
  =
  \sum_{\substack{i,j\ge0\\ i+j\le r-2}}
  (r-1-i-j)\,q_{3,2}(m-3i-2j).
\]
For \(0<m<2r\), this coefficient sum equals
\(-\chi_{3,2}(-r,m)\).  At \(m=2r\) one shifted degree in the remaining variables hits the
first extra complete-intersection value, giving the \(+1\).  For
\(2r<m<3r\), the remaining non-stable shifted degrees give the displayed
cubic correction.  For \(m\ge3r\), every shifted degree in the sum has
stabilized to
\[
  \deg(A_0,A_1,B_0,B_1)=3^2\,2^2=36,
\]
and the triangular multiplicity
\(\sum_{i+j\le r-2}(r-1-i-j)=\binom{r+1}{3}\) gives
\[
  36\binom{r+1}{3}=6r(r^2-1).
\]

For \(7885\), put
\[
  \chi_{4,1}(-r,m)
  =
  -2rm^2-2r+\frac{5}{6}m^3+\frac{25}{6}m ,
\]
and
\[
  C(n)=[t^n]\frac{(1-t^4)^2}{(1-t)^5(1-t^3)},
  \qquad C(n)=0\quad(n<0).
\]
The corresponding formula in \cite{LarforsSchneider2019} is
\begin{equation}
h^1(X,\cO_X(-r,m))=
\begin{cases}
0, & m<0,\\
r-1, & m=0,\\
-\chi_{4,1}(-r,m), & 0<m<r,\\
-\chi_{4,1}(-r,m)+C(m-r), & r\le m<4r-2,\\
\frac{8}{3} r(r^2-1), & m\ge4r-2 .
\end{cases}
\label{eq:sec3-ls-7885-piecewise}
\end{equation}
Our expression gives the same function because
\[
  \HF(N_r^{(4,1)},m)
  =
  \sum_{\substack{i,j\ge0\\ i+j\le r-2}}
  (r-1-i-j)\,q_{4,1}(m-4i-j).
\]
For \(0<m<r\), this is just the index contribution
\(-\chi_{4,1}(-r,m)\).  In the intermediate region
\(r\le m<4r-2\), the non-stable shifted complete-intersection values are
packaged by the coefficient \(C(m-r)\).  For \(m\ge4r-2\), every shifted
degree has stabilized to
\[
  \deg(A_0,A_1,B_0,B_1)=4^2\,1^2=16,
\]
so the correction contribution is
\[
  16\binom{r+1}{3}=\frac83 r(r^2-1).
\]

Thus the displays \eqref{eq:sec3-ls-7858-piecewise} and
\eqref{eq:sec3-ls-7885-piecewise} are not separate ansatz formulae: they
are the coefficient expansions of the Hilbert series
\eqref{eq:ls-linear-HS}.  For \(r\ge2\) and \(m>0\), the full cohomology is
obtained from the same correction term by
\[
  h^\bullet(X,\cO_X(-r,m))
  =
  \bigl(\chi_{\alpha,\beta}(-r,m)+\HF(N_r^{(\alpha,\beta)},m),
  \HF(N_r^{(\alpha,\beta)},m),0,0\bigr).
\]
The \(m=0\) line in the displayed formulae is a separate boundary case; the
full cohomology vector is not obtained by substituting \(m=0\) into the
\(m>0\) formula above.  In particular,
\[
  h^3(X,\cO_X(-r,0))=h^0(X,\cO_X(r,0))=r+1 .
\]
The \(m<0\) side is obtained from the corresponding positive chamber by
Serre duality.

There is also an isolated \(d_2\)-type boundary in the two-pencil cases,
separate from the one-negative chamber above.  On the ray \(k_1=0\), the
middle Koszul terms lie in the \(\bP^1\) acyclic range, and the second-page map is
multiplication by
\[
  \Delta=A_0B_1-A_1B_0,\qquad \deg \Delta=\alpha+\beta=5.
\]
In the positive boundary row this is the map
\[
  R_{m-5}\xrightarrow{\cdot\Delta}R_m,
\]
with the convention \(R_n=0\) for \(n<0\).  Hence the boundary contribution
is the Hilbert function of \(R/(\Delta)\); the opposite boundary follows by
Serre duality.  This is again an ordinary cokernel of a concrete map, not a
separate chamber ansatz.

\subsubsection{Three \texorpdfstring{\(\bP^2\times\bP^3\) cases }{P2 x P3 endpoint rows}}
\label{subsec:sec3-larfors-p2p3}

The three remaining codimension-two examples in
\cite{LarforsSchneider2019} share a single mechanism.  Their
main wall structure, for \(k=1,2,3\), is governed by the same left-endpoint
kernel.  The examples \(7833\) and
\(7883\) also contain finer structures inside the
chamber which are detected
by the right-endpoint cokernel Hilbert numerator.

We write the three configurations uniformly as
\[
  X_k=
  \left[
  \begin{array}{c|cc}
  \bP^2_x&2&1\\
  \bP^3_y&k&4-k
  \end{array}
  \right],
  \qquad k=1,2,3,
  \]
corresponding to CICYs \(7833,7844,7883\).  Let
\[
  T=\bC[x_0,x_1,x_2],\qquad
  S=\bC[y_0,\ldots,y_3],
\]
and write the two equations as
\[
  F\in T_2\otimes S_k,\qquad
  G\in T_1\otimes S_{4-k}.
\]
Thus \(G\) is linear in the \(\bP^2_x\)-variables,
\[
  G=x_0g_0+x_1g_1+x_2g_2,\qquad g_i\in S_{4-k}.
\]

Consider \(L=\cO_{X_k}(-r,b)\), with \(r\ge3\) and \(b>0\).  The relevant chamber has one top
\(\bP^2_x\)-row; the one-top \(\bP^3_y\)-rows are Serre dual to it.  After
Serre duality on \(\bP^2_x\), we obtain the \(S\)-free complex
\[
  0
  \to
  T_{r-3}\otimes S(-4)
  \to
  T_{r-2}\otimes S(-k)
  \oplus
  T_{r-1}\otimes S(-(4-k))
  \to
  T_r\otimes S
  \to
  M_r^{(k)}
  \to0,
\]
where
\[
  T_j=H^0(\bP^2_x,\cO(j)),\qquad
  S_m=H^0(\bP^3_y,\cO(m)),
\]
and \(T_j=S_m=0\) for negative degrees.  Set
\[
  C_2=T_{r-3}\otimes S(-4),\qquad
  C_1=T_{r-2}\otimes S(-k)\oplus T_{r-1}\otimes S(-(4-k)),
  \qquad
  C_0=T_r\otimes S .
\]
Thus \(M_r^{(k)}=\coker(C_1\to C_0)\).  We apply the graded dual
\[
  D(-)=\operatorname{Hom}_S(-,S(-4))
\]
to \(C_2\to C_1\to C_0\), and then take degree \(b\).  This
reverses the arrows:
\[
\begin{aligned}
0\to&
T_r^\vee\otimes S_{b-4}\\
\to&
T_{r-2}^\vee\otimes S_{b+k-4}
\oplus
T_{r-1}^\vee\otimes S_{b-k}\\
\to&
T_{r-3}^\vee\otimes S_b
\to0 .
\end{aligned}
\]
The \(M_r^{(k)}\)-term is therefore retained as the left endpoint: applying
\(D\) to
\[
  C_1\to C_0\to M_r^{(k)}\to0
\]
identifies \(D(M_r^{(k)})_b\) with the left kernel of the dual row.  We
write
\[
  K_{r,k}(b)=
  \dim\ker\!\left(D(C_0)_b\to D(C_1)_b\right)
\]
for the left-endpoint dimension, and
\[
  R_{r,k}(b)=
  \dim\coker\!\left(D(C_1)_b\to D(C_2)_b\right)
\]
for the right-endpoint dimension.  The middle term is then locked by the row
Euler characteristic:
\[
  h^0_{\rm row}=K_{r,k}(b),\qquad
  h^1_{\rm row}=K_{r,k}(b)+R_{r,k}(b)-\chi(\cO_{X_k}(-r,b)),
  \qquad
  h^2_{\rm row}=R_{r,k}(b).
\]

The clean model is \(7844\), corresponding to \(k=2\).  In this case the
left endpoint has the free resolution
\begin{equation}
0\to
S(-6r-4)^{\binom{r-1}{2}}
\to
S(-6r-2)^{r^2}
\to
S(-6r)^{\binom{r+2}{2}}
\to
K_r^{(2)}
\to0 .
\label{eq:ls-7844-left-kernel-resolution}
\end{equation}
Consequently \(K_{r,2}(b)=0\) for \(0<b<6r\), the wall value is
\[
  K_{r,2}(6r)=\binom{r+2}{2},
\]
and the chamber \(b>6r\) is read from the shifted numerator
\[
  \binom{r+2}{2}-r^2t^2+\binom{r-1}{2}t^4 .
\]
Thus the main wall structure in \cite{LarforsSchneider2019},
\(m_1=-6m_0\), is the first nonzero degree of the left endpoint.
The same determinant computation in
Appendix~\ref{app:ls-7844-left-kernel}, shifts the first possible left
endpoint degree to \((8-k)r\) for \(X_k\).  Thus the corresponding walls are
\(b=7r,6r,5r\) for CICYs \(7833,7844,7883\), respectively.

For \(7844\), the finite diagonal strip comes from the right endpoint.  Let
\[
  P(n)=
  \begin{cases}
  \binom{n+3}{3},& n\ge0,\\
  0,& n<0,
  \end{cases}
  \qquad
  G_r(b)=
  \binom{r-1}{2}P(b)-r^2P(b-2)+\binom{r+2}{2}P(b-4).
\]
The endpoint calculation gives the Hilbert-series numerator of the
right-endpoint cokernel \(E_r\) in the form
\begin{equation}
  \HS(E_r;t)
  =
  \sum_{b=0}^{B_r}\max\{G_r(b)-K_{r,2}(b),0\}t^b,
  \qquad
  N(E_r;t)=(1-t)^4\HS(E_r;t).
\label{eq:ls-7844-right-end-numerator}
\end{equation}
For \(3\le r\le15\) the relevant right-end ranks are calculated exactly in
the finite range where the positive part can occur.  Thus the finite strip
in \(7844\) is the right endpoint \(R_{r,2}(b)\), while the wall and the
stable chamber above \(b=6r\) are controlled by the left endpoint and the row
Euler characteristic.

The neighbouring endpoint families \(7833\) and \(7883\) use the same row,
but they do have more structure along lines shown in the following table:
\[
\begin{array}{c|c|c}
\text{CICY} & \text{black-line relation} & \text{verified range}\\
\hline
7833 & r=2b-1 & 2\le b\le25\\
7833 & r=2b+1 & 4\le b\le25\\
7883 & (r,b)=(3q,2q) & 1\le q\le23 .
\end{array}
\]
On the two \(7833\) lines the shifted left endpoint has not yet appeared:
its first possible degree is \(7r\), while \(b<7r\) on both listed
diagonals.  The right-endpoint calculations in the package give
\(R_{r,1}(b)=0\) on the same finite ranges.  Hence both endpoint defects
vanish, and the row Euler characteristic gives
\[
  h^\bullet\bigl(X_{7833},\cO(-2b+1,b)\bigr)
  =
  \left(0,\frac{(b+2)(4b^2+7b-9)}6,0,0\right)
\]
for \(2\le b\le25\), and
\[
  h^\bullet\bigl(X_{7833},\cO(-2b-1,b)\bigr)
  =
  \left(0,\frac{(b-2)(4b^2-7b-9)}6,0,0\right)
\]
for \(4\le b\le25\).  For \(7883\), the shifted left endpoint starts in
degree \(5r=15q\), so it vanishes along \((r,b)=(3q,2q)\).  The
right-endpoint calculations give value \(1\) at \(q=1\) and vanishing for
\(2\le q\le23\).  Thus
\[
  h^\bullet\bigl(X_{7883},\cO(-3,2)\bigr)=(0,3,1,0),
\]
and
\[
  h^\bullet\bigl(X_{7883},\cO(-3q,2q)\bigr)
  =
  \left(0,\frac{q(q^2+5)}3,0,0\right)
  \qquad(2\le q\le23).
\]

The proof of the left-endpoint resolution and determinant shift is given in
the appendix.  The same appendix also explains the finite right-endpoint
checks used above: after quotienting \(D(C_1)_b\) by the image of
\(D(C_0)_b\), the induced map to \(D(C_2)_b\) is verified to have maximal
rank in the listed finite ranges.  The accompanying archive keeps the
specialization data and pivot/minor data needed to replay these checks.
These black-line pieces already appear in the formulae of
\cite{LarforsSchneider2019}; the point here is to identify them with the same
left--right endpoint mechanism.

\section{New structures and new formulas from Hilbert functions}
\label{sec:new-formulas-rank-defects}
In this section we use Hilbert functions in two ways.  First, we
show that they can detect new structures in line bundle cohomology which are
invisible to polynomial extrapolation or machine-learning fits on bounded
data.  Second, we show that the same Hilbert-function viewpoint can produce
analytic cohomology formulae on regions of the Picard lattice, and can reduce
the remaining cases to explicit finite-rank or finite-box verification. In either case,
we can get cohomology formula for infinitely many line bundles. 

\subsection{New structures from Hilbert functions}
\label{subsec:new-structures-hilbert-numerators}

We give three examples in which the ambient cohomology and the source--target
dimensions suggest a clean maximal-rank answer, but the Hilbert
numerator contains extra terms.  These terms produce additional walls inside
a fixed ambient-cohomology region.  The first two examples come from CICY
\(7880\): a one-sided top-\(\bP^2\) wall and a finite kernel
calculation on the negative-\(\bP^1\) branch.  The third example is a
bigraded boundary wall in the two-negative region of the tetraquadric
\(7862\).

\subsubsection{\texorpdfstring{CICY \(7880\): hidden top-\(\bP^2\) subwalls}{CICY 7880: hidden top-P2 subwalls}}
\label{subsec:new-7880}

We now return to the hypersurface
\[
  X_{7880}\subset \bP^1_x\times\bP^1_y\times\bP^2_z,
  \qquad [X_{7880}]=(2,2,3).
\]
Subsection~\ref{subsec:sec3-p1p1p2-223} has covered some chambers of this CICY.  Here
we focus on the top-\(\bP^2_z\) case, and we will see that there is a new
internal wall which is not visible from ambient cohomology signs
or from source--target dimensions.

Write the defining equation as
\[
  f=\sum_{\alpha,\beta=0}^{2}
  x_0^{2-\alpha}x_1^\alpha
  y_0^{2-\beta}y_1^\beta
  F_{\alpha\beta}(z_0,z_1,z_2),
  \qquad F_{\alpha\beta}\in H^0(\bP^2,\cO_{\bP^2}(3)).
\]
For
\[
  L=\cO_X(p,q,-c),\qquad p,q>0,\quad c\ge3,
\]
Serre duality on the \(\bP^2_z\)-factor turns the top-\(\bP^2_z\) row into
the degree-\(c\) part of
\[
 \Theta_{p,q}:
 S(-3)^{(p+1)(q+1)}
 \longrightarrow
 S^{(p-1)(q-1)},\qquad
 S=\bC[z_0,z_1,z_2],
\]
where an element of the source is written as an array
\[
  u=(u_{IJ})_{0\leq I\leq p,\;0\leq J\leq q},
  \qquad u_{IJ}\in S_{c-3}
\]
in degree \(c\), and \(u_{IJ}\) is understood to be zero if an index lies
outside this range.  The target is indexed by
\[
  0\leq i\leq p-2,\qquad 0\leq j\leq q-2,
\]
and the map is
\[
  (\Theta u)_{ij}
  =
  \sum_{\alpha,\beta=0}^{2}F_{\alpha\beta}\,
  u_{i+\alpha,j+\beta}.
\]
The plus sign in the subscripts is a consequence of the Serre-dual
transpose convention.  Before dualizing, multiplication by \(f\) sends a
source monomial indexed by \((i,j)\) to target monomials indexed by
\((i+\alpha,j+\beta)\).  After passing to the dual map, the target component
\((i,j)\) pairs with the source components \(u_{i+\alpha,j+\beta}\).
Let
\[
  M_{p,q}=\coker\Theta_{p,q}.
\]
Then the top-\(\bP^2_z\) contribution is
\[
  h^1(X,\cO_X(p,q,-c))=\HF(M_{p,q},c),
\]
\[
  h^2(X,\cO_X(p,q,-c))
  =
  (p+1)(q+1)\binom{c-1}{2}
  -(p-1)(q-1)\binom{c+2}{2}
  +\HF(M_{p,q},c),
\]
with \(h^0=h^3=0\) in this chamber.

The source and target dimensions suggest a maximal-rank prediction, but the
graded module contains additional graded pieces that are not visible from
these dimensions alone.  In this example we prove two generic one-sided wall
theorems:
\[
\begin{array}{c|c|c|c}
q & c & \HF(M_{p,q},c) & \text{first nonzero value}\\
\hline
2 & 6 & \max\{p-49,0\} & p=50\\
3 & 9 & 2\max\{p-101,0\} & p=102 .
\end{array}
\]
These two cases suggest an infinite one-sided series of hidden walls.  They
are internal to a single ambient-cohomology region and come from
free summands in the \(B=\bC[s,t]\)-linear module controlling the
\(p\)-direction.  In the following \(B\)-Hilbert functions, the grading is
the \(p-2\) grading:
\[
  B(-48)
  \quad\text{for }q=2,
  \qquad
  B(-100)^{\oplus2}
  \quad\text{for }q=3.
\]
Thus the numerator detects the delayed wall as a graded object, before one
performs any individual rank computation at the first nonzero degree.  This
provides a useful example of how Hilbert functions can track walls in line
bundle cohomology calculations.  We now treat the two cases in turn.

\begin{theorem}[A hidden \(q=2\) top-\(\bP^2\) subwall]
\label{thm:223-q2-hidden-wall}
For coefficients \(F_{\alpha\beta}\) in a nonempty Zariski-open subset of
the full coefficient space, and hence for a general smooth \(X_{7880}\), one
has
\[
  \HF(M_{p,2},6)=\max\{p-49,0\}.
\]
Equivalently,
\[
  h^1(X_{7880},\cO_X(p,2,-6))=\max\{p-49,0\},
\]
and
\[
  h^2(X_{7880},\cO_X(p,2,-6))
  =
  2p+58+\max\{p-49,0\}.
\]
Thus for \(p\ge50\) the active map is not of maximal rank, even though
the source of \(\Theta_{p,2}(6)\) has dimension larger than its target.
\end{theorem}

The proof is given in Appendix~\ref{app:7880-q2-hidden-wall}.

For comparison, a direct matrix calculation at one integer specialization,
reduced modulo \(32003\) and checked again modulo \(32009\), gives the
following ranks, where
\(h^1=\dim\operatorname{tgt}-\operatorname{rank}\Theta_{p,2}(6)\) and
\(h^2=\dim\operatorname{src}-\operatorname{rank}\Theta_{p,2}(6)\):
\[
\begin{array}{c|c|c|c|c}
p&\dim\operatorname{src}&\dim\operatorname{tgt}
&\operatorname{rank}\Theta_{p,2}(6)&(h^0,h^1,h^2,h^3)\\
\hline
48&1470&1316&1316&(0,0,154,0)\\
49&1500&1344&1344&(0,0,156,0)\\
50&1530&1372&1371&(0,1,159,0)\\
51&1560&1400&1398&(0,2,162,0)\\
52&1590&1428&1425&(0,3,165,0)\\
60&1830&1652&1641&(0,11,189,0).
\end{array}
\]

\begin{theorem}[The \(q=3\) hidden wall]
\label{thm:223-q3-hidden-line}
For coefficients \(F_{\alpha\beta}\) in a nonempty Zariski-open subset of
the full coefficient space, and hence for a general smooth \(X_{7880}\), one
has
\[
  \HF(M_{p,3},9)=2\max\{p-101,0\}.
\]
Equivalently,
\[
  h^1(X_{7880},\cO_X(p,3,-9))=2\max\{p-101,0\},
\]
and
\[
  h^2(X_{7880},\cO_X(p,3,-9))
  =
  2p+222+2\max\{p-101,0\}.
\]
\end{theorem}

\begin{proof}
Write
\[
  \HS(M_{p,3};t)=\frac{N_{p,3}(t)}{(1-t)^3}.
\]
The clean maximal-rank numerator in degree \(q=3\) is
\[
  N^{\rm clean}_{p,3}(t)
  =
  (2p-2)-4(p+1)t^3+(2p+222)t^9
  +(6p-390)t^{10}+(174-6p)t^{11}.
\]
The characteristic-zero certificate for the associated two-row
\(B=\bC[s,t]\)-linear presentation gives
\begin{equation}
  N_{p,3}(t)
  =
  N^{\rm clean}_{p,3}(t)
  +
  2\max\{p-101,0\}\,t^9(1-t)^3 .
  \label{eq:7880-q3-numerator}
\end{equation}
Equivalently, the hidden \(B\)-module contribution in the \(p-2\) grading is
\(B(-100)^{\oplus2}\).  Here \(K_{99}\) denotes the degree-\(99\) kernel
piece in this \(B\)-linear presentation; the required exactness condition is
\(K_{99}=0\).  We verify this open rank condition by a nonzero integer minor
in the corresponding degree-\(99\) presentation matrix, detected modulo
\(32003\).  The \(B(-100)^{\oplus2}\) module structure supplies the forced
upper bound, while the minor excludes further kernel contribution in the
relevant degree.  Hence the identity holds on a nonempty open set of cubic
coefficients.

Now take the coefficient of \(t^9\).  The clean numerator contributes zero:
\[
  [t^9]\frac{N^{\rm clean}_{p,3}(t)}{(1-t)^3}=0,
\]
while the correction contributes
\[
  [t^9]\frac{2\max\{p-101,0\}t^9(1-t)^3}{(1-t)^3}
  =
  2\max\{p-101,0\}.
\]
Therefore
\[
  \HF(M_{p,3},9)=2\max\{p-101,0\}.
\]
The displayed formula for \(h^2\) follows from the top-\(\bP^2\) row
identity above, since
\[
  (p+1)4\binom{8}{2}-(p-1)2\binom{11}{2}=2p+222.
\]
\end{proof}

In the \(B\)-degree variable \(\tau\), the hidden \(q=3\) contribution is
\[
  \HS_{\mathrm{hidden},3}(\tau)=\frac{2\tau^{100}}{(1-\tau)^2}.
  \label{eq:223-q3-hidden-series}
\]

Finite-field scans suggest the following general Hilbert-numerator pattern.

\begin{conjecture}[Hidden wall structure in top-\(\bP^2\) numerator series for \(7880\)]
\label{conj:7880-hidden-numerator-series}
For a general CICY \(X_{7880}\), the graded module \(M_{p,q}\) has, in
degree \(c=3q\), a hidden \(B=\bC[s,t]\)-free contribution
\[
  B(-(52q-56))^{\oplus(q-1)}
\]
in the \(p-2\) grading.
Equivalently,
\[
  \HS_{\mathrm{hidden},q}(\tau)
  =
  \frac{(q-1)\tau^{52q-56}}{(1-\tau)^2},
\]
and hence
\[
  \HF(M_{p,q},3q)
  =
  (q-1)\max\{p-(52q-55),0\}.
\]
The corresponding hidden wall begins at
\[
  p=52q-54 .
\]
\end{conjecture}

The conjecture should be read as a Hilbert-numerator statement, not as a
rank-fitting rule.  For \(q=2\) it is exactly
Theorem~\ref{thm:223-q2-hidden-wall}; for \(q=3\) it is
Theorem~\ref{thm:223-q3-hidden-line}.  For \(q=5\), finite-field rank
computations at the predicted
wall give
\[
  q=5,\qquad c=15,\qquad
  \HF(M_{p,5},15)=4\max\{p-205,0\}.
\]
They agree for the tested values
\[
  p=200,202,204,205,206,207,208,210,215,220
\]
over the two primes \(32003\) and \(32009\).  This is evidence for the same
Hilbert-numerator pattern, not a characteristic-zero proof.
The proof problem is now sharply formulated: identify, inside the
\(B\)-linear dual graded module at \(z\)-degree \(3q\), the free summand
\(B(-(52q-56))^{\oplus(q-1)}\).  Once this summand is proved, the wall
position follows immediately from the Hilbert function of \(B\).

\paragraph{What is certified here.}
Apart from the two one-sided hidden walls above, the completed finite box has
the expected maximal-rank Hilbert function: for every pair
\[
  2\le p,q\le14
\]
the two critical-rank checks described in
Appendix~\ref{app:critical-rank-rays} certify the maximal-rank,
positive-part answer for all \(c\ge0\).  Thus the box gives \(13^2=169\)
infinite \(c\)-rays, while Theorems~\ref{thm:223-q2-hidden-wall} and
\ref{thm:223-q3-hidden-line} account for the hidden walls outside this small
square.

The machine-readable evidence for the \(7880\) computations in this
subsection is included in the companion repository
\url{https://github.com/jtwangbimsa-dragon/cicylinecoh_alpha}.

\subsubsection{\texorpdfstring{CICY \(7880\): the negative-\(\bP^1\) kernel box}{CICY 7880: the negative-P1 kernel box}}
\label{subsec:new-7880-negative-p1-box}

We now return to the branch left open in
Subsection~\ref{subsec:sec3-p1p1p2-223}.  There the line bundle is
\[
  L=\cO_X(-a,b,c),\qquad a\ge2,
\]
and the only unresolved range was
\[
  c>3a,\qquad b<2a-1.
\]
The structured Hilbert-numerator certificates and the direct finite-field
cohomology checks for this \(7880\) one-negative box are included in the same
companion repository as above.
For fixed \(a\) and \(b\), the degree-\((b,c)\) Toeplitz map can be viewed as
one graded map over
\[
  S=\bC[z_0,z_1,z_2]:
\]
\begin{equation}
  \Phi_{a,b}:S(-3)^{(a+1)(b-1)}
  \longrightarrow S^{(a-1)(b+1)} .
  \label{eq:7880-negative-p1-fixed-b-map}
\end{equation}
Here the factor \(b-1\) is \(\dim H^0(\bP^1_y,\cO(b-2))\), the factor
\(b+1\) is \(\dim H^0(\bP^1_y,\cO(b))\), and the shift \(S(-3)\) accounts for the
degree \(3\) in the \(z\)-variables.  Thus the degree-\(c\) kernel of
\eqref{eq:7880-negative-p1-fixed-b-map} is exactly the kernel term
from the table in Subsection~\ref{subsec:sec3-p1p1p2-223}.

Let
\[
  E_{a,b}=\ker\Phi_{a,b},\qquad
  \HS(E_{a,b};t)=\frac{N_{a,b}(t)}{(1-t)^3}.
\]
Then, in this branch,
\begin{equation}
  h^0(X,\cO_X(-a,b,c))
  =
  [t^c]\frac{N_{a,b}(t)}{(1-t)^3},
  \qquad
  h^1=h^0-\chi_a(b,c),
  \label{eq:7880-negative-p1-kernel-coho}
\end{equation}
and \(h^2=h^3=0\), where
\[
  \chi_a(b,c)=(-3ab-ac+bc)c+2(-a+b)+3c .
\]

In the finite box
\[
  2\le a,b\le10,\qquad b<2a-1.
\]
we computed the Hilbert numerator \(N_{a,b}(t)\) of the kernel module
\(E_{a,b}\).  This is a finite statement in the variables \(a,b\), but for each
fixed pair it determines all degrees \(c\).  Put
\[
  H_{a,b}(c):=[t^c]\frac{N_{a,b}(t)}{(1-t)^3}.
\]
Let \(r=b-a\) and let
\(\epsilon_a=1\) if \(a\) is even and \(\epsilon_a=0\) if \(a\) is odd.
The numerator is zero unless
\[
  1\le r\le4,\qquad r+2\le a\le10-r.
\]
In the nonzero cases it is
\[
\begin{aligned}
N_{a,a+1}(t)
  &=2t^{3\binom{a+1}{2}},
  &&3\le a\le9,\\
N_{a,a+2}(t)
  &=(4-\epsilon_a)t^{3\lfloor (a+1)^2/4\rfloor}
    +\epsilon_a t^{3\lfloor (a+1)^2/4\rfloor+3},
  &&4\le a\le8,\\
N_{a,a+3}(t)
  &=(6-2\epsilon_a)t^{3\lfloor (a+1)^2/5\rfloor}
    +2\epsilon_a t^{3\lfloor (a+1)^2/5\rfloor+3},
  &&5\le a\le7,\\
N_{6,10}(t)
  &=9t^{24}-t^{27}.
\end{aligned}
\]
Thus, for
\[
  2\le a,b\le10,\qquad b<2a-1,\qquad c>3a,
\]
the one-negative branch is given by the simplified formula
\[
  h^\bullet(X,\cO_X(-a,b,c))
  =
  \bigl(H_{a,b}(c),\,H_{a,b}(c)-\chi_a(b,c),\,0,\,0\bigr).
\]
The numerator computation covers all \(61\) pairs in this finite box: \(45\)
have zero numerator and the \(16\) listed above have nonzero numerator.

\subsubsection{\texorpdfstring{CICY \(7862\): a certified two-negative finite box}{CICY 7862: a certified two-negative finite box}}
\label{subsec:new-7862}

The CICY \(7862\) has the following configuration
\[
  X_{7862}\subset(\bP^1)^4,\qquad [X]=(2,2,2,2).
\]
All chambers except the genuinely two-negative ones are handled by the
hypersurface mechanisms discussed before.  The remaining chamber has the form
\[
  L=\cO_X(-r,-s,a,b),\qquad r,s,a,b\geq2.
\]
The hypersurface long exact sequence reduces the computation to the active
bigraded map
\[
  \Phi_{r,s}(a,b):
  R_{a-2,b-2}^{(r+1)(s+1)}
  \longrightarrow
  R_{a,b}^{(r-1)(s-1)},
  \qquad
  R=\bC[u_0,u_1,v_0,v_1].
\]
Let
\[
  M_{r,s}
  =
  \coker\left(
  R(-2,-2)^{(r+1)(s+1)}
  \longrightarrow R^{(r-1)(s-1)}
  \right),
\]
and set \(F_{r,s}(a,b)=\HF(M_{r,s},(a,b))\).  Then
\begin{align}
  h^2(X,\cO_X(-r,-s,a,b))&=F_{r,s}(a,b),\nonumber\\
  h^1(X,\cO_X(-r,-s,a,b))
  &=(r+1)(s+1)(a-1)(b-1)
    -(r-1)(s-1)(a+1)(b+1)\nonumber\\
  &\quad+F_{r,s}(a,b),
  \label{eq:tetra-new-general-two-negative}
\end{align}
and \(h^0=h^3=0\).  Thus the entire question is the Hilbert function of
the vector-valued cokernel module \(M_{r,s}\).

For the tetraquadric this Hilbert function is genuinely bigraded.  For a
fixed pair \((r,s)\), computing an initial module of
\(\operatorname{im}\Phi_{r,s}\) gives the monomial standard basis, hence the
whole bigraded Hilbert series of \(M_{r,s}\).  Equivalently, it gives one
bigraded Hilbert numerator
\[
  N_{r,s}(u,v)=(1-u)^2(1-v)^2\HS(M_{r,s};u,v),
\]
whose coefficients determine all \(F_{r,s}(a,b)\), and hence all line bundles
\(\cO_X(-r,-s,a,b)\) with \(a,b\ge2\), through
\eqref{eq:tetra-new-general-two-negative}.  Thus a finite box in the
negative directions still covers infinitely many line bundles.  The useful
question is therefore not whether a few sampled maps have the expected rank,
but whether
\(N_{r,s}(u,v)\) contains extra boundary terms beyond the clean
positive-part expression.  The clean maximal-rank guess would give
\[
  F_{r,s}^{\rm clean}(a,b)
  =
  \max\{(r-1)(s-1)(a+1)(b+1)
  -(r+1)(s+1)(a-1)(b-1),0\}.
\]

Exact \(\bQ\)-initial-module computations give a clean finite range where the
positive-part expression is correct.  The rank-certificate package then
extends this to the finite box
\[
  2\le r,s\le9,\qquad a,b\ge2.
\]
In this box the only certified departures from the clean expression occur on
the diagonal pairs
\[
  (r,s)=(4,4),(6,6),(8,8).
\]

\begin{theorem}[Certified finite box for the tetraquadric two-negative chamber]
\label{thm:tetra-7862-finite-box-rs-2-9}
Let \(X\) be a general smooth tetraquadric.  For
\[
  L=\cO_X(-r,-s,a,b),\qquad 2\le r,s\le9,\qquad a,b\ge2,
\]
the cohomology is given by \eqref{eq:tetra-new-general-two-negative}, with
\[
  F_{r,s}(a,b)=F_{r,s}^{\rm clean}(a,b)+\Delta_{r,s}(a,b),
\]
where \(\Delta_{r,s}=0\) unless \((r,s)=(2m,2m)\) for
\(m=2,3,4\).  For these three diagonal cases,
\[
  \Delta_{2m,2m}(a,b)
  =
  (m-1){\bf 1}_{a=m}\max\{b-(46m-49),0\}
  +(m-1){\bf 1}_{b=m}\max\{a-(46m-49),0\}.
\]
\end{theorem}

\begin{proof}[Proof sketch]
The finite-box certificate reduces the claim to \(2310\) finite rank checks,
all of which are closed in characteristic zero by the accompanying
rank-certificate package.  The maximal-rank rows are certified by minors or
fixed-\(a\) structural covers, while the eight non-maximal high-wall rows are
certified by exact rational fixed-\(a\) kernel replay; substituting the
resulting Hilbert function into \eqref{eq:tetra-new-general-two-negative}
gives the formula.
\end{proof}

The certificate package, including the rank data and replay artifacts used in
this proof, can be found 
\href{https://drive.google.com/drive/folders/1Y_R9r4OMLkBOgba3Jkvboay7OfgqEnl0?usp=sharing}{here}.

For \((r,s)=(4,4)\), writing
\[
  \HS^{\rm clean}_{4,4}(u,v)
  =
  \sum_{a,b\ge2}F^{\rm clean}_{4,4}(a,b)u^av^b ,
\]
the computation gives
\[
  \HS(M_{4,4};u,v)
  =
  \HS^{\rm clean}_{4,4}(u,v)
  +
  \frac{u^2v^{44}}{(1-v)^2}
  +
  \frac{u^{44}v^2}{(1-u)^2}.
\]
Equivalently, for a general smooth tetraquadric,
\begin{equation}
  F_{4,4}(a,b)
  =
  \max\{2[17(a+b)-8(ab+1)],0\}
  +{\bf 1}_{a=2}\max(b-43,0)
  +{\bf 1}_{b=2}\max(a-43,0).
  \label{eq:tetra-new-44-wall}
\end{equation}
Thus the two new subwalls are the boundary rays \(a=2,\ b\ge44\) and
\(b=2,\ a\ge44\).  For instance, at \((a,b)=(2,44)\) the clean formula
predicts \((h^1,h^2)=(0,140)\), while \eqref{eq:tetra-new-44-wall} gives
\((h^1,h^2)=(1,141)\).

The same certified finite-box theorem gives the diagonal pattern for
\(m=2,3,4\):
\[
  \HS(M_{2m,2m};u,v)
  =
  \HS^{\rm clean}_{2m,2m}(u,v)
  +(m-1)\left(
  \frac{u^m v^{46m-48}}{(1-v)^2}
  +
  \frac{u^{46m-48}v^m}{(1-u)^2}
  \right).
\]
For \(m>4\) the displayed expression remains a natural extrapolation from the
finite-field scans, but it is not part of the certified finite box above.

\subsection{\texorpdfstring{New formulae: CICY \(7707\) as a case study}{New formulae: CICY 7707 as a case study}}
\label{subsec:new-7707}
This subsection gives a complete worked example of the method on CICY
\(7707\).  The goal is to show how Hilbert-function calculations can be used to
derive new line-bundle cohomology formulae for a CICY that was not
previously covered by a closed formula package.  In practice the output has
two complementary parts.  Some chambers admit uniform analytic Hilbert-series
arguments, while other chambers are controlled by finite-box theorems proved
by exact computational commutative algebra.  The latter can still cover
a large and useful class of line bundles.  
Let
\[
X_{7707}=
\left[\begin{array}{c|cc}
\bP^1_x&1&1\\
\bP^1_y&0&2\\
\bP^3_z&3&1
\end{array}\right],
\]
and put
\[
  D_1=(1,0,3),\qquad D_2=(1,2,1).
\]
Write
\begin{equation}
  F=x_0A_0(z)+x_1A_1(z),\qquad
  G=x_0B_0(y,z)+x_1B_1(y,z),
  \label{eq:7707-FG}
\end{equation}
where
\[
  A_i\in H^0(\bP^3_z,\cO(3)),\qquad
  B_i\in H^0(\bP^1_y\times\bP^3_z,\cO(2,1)).
\]
Eliminating the \(\bP^1_x\)-coordinate gives
\begin{equation}
  \Delta=A_0B_1-A_1B_0
  \in H^0(\bP^1_y\times\bP^3_z,\cO(2,4)).
  \label{eq:7707-delta}
\end{equation}
The four Koszul shifts are
\[
  0,\qquad D_1=(1,0,3),\qquad D_2=(1,2,1),\qquad
  D_1+D_2=(2,2,4).
\]
The Koszul spectral sequence has
\[
  E_1^{-p,q}(k)
  =
  \bigoplus_{|I|=p}
  H^q\!\left(A,\cO_A(k-D_I)\right),
  \qquad
  A=\bP^1_x\times\bP^1_y\times\bP^3_z ,
\]
with \(D_I=\sum_{i\in I}D_i\), first differential induced by \(F,G\), and
\[
  d_r:E_r^{-p,q}\longrightarrow E_r^{-p+r,q-r+1}.
\]
Since the Koszul complex has length two, only \(d_1\) and \(d_2\) can occur.
For CICY \(7707\) we use the partition
\[
k_1,k_2\in \{\le -2,\ -1,\ 0,\ 1,\ \ge2\},
\qquad
k_3\in
\{\le -4,\ -3,\ -2,\ -1,\ 0,\ 1,\ 2,\ 3,\ \ge4\},
\]
which gives \(5\cdot5\cdot9=225\) coarse regions.  Most of these
regions contain no active rank problem.  We therefore list
only the ranges with active maps below.  
Throughout the range statements, set
\begin{equation}
  \chi_{7707}(r,s,t)
  =
  3rst+3rt^2+2st^2+2r+2s+\frac{t^3+11t}{3}.
  \label{eq:7707-euler-characteristic}
\end{equation}

\subsubsection*{\texorpdfstring{\(k=(-\ell,b,c)\) and \(k=(\ell,-b,-c)\)}{k=(-ell,b,c) and k=(ell,-b,-c)}}
In the range
\[
  k=(-\ell,b,c),\qquad \ell,b\ge2,\qquad c\ge4,
\]
and its Serre-dual range \(k=(\ell,-b,-c)\), the active rank defect is the
Hilbert function of the module \(M_{\ell,3}\) in
Lemma~\ref{lem:7707-p1x-row}.  Put
\[
  \rho^x_\ell(b,c)=\HF(M_{\ell,3};b,c),
\]
equivalently, \(\rho^x_\ell(b,c)\) is the coefficient of \(u^bv^c\) in the
Hilbert series displayed in \eqref{eq:7707-p1x-HS}.
The corresponding cohomology vectors are
\begin{align}
  h^\bullet\bigl(X,\cO_X(-\ell,b,c)\bigr)
  &=
  \bigl(\chi_{7707}(-\ell,b,c)+\rho^x_\ell(b,c),
  \rho^x_\ell(b,c),0,0\bigr),
  \label{eq:7707-p1x-cohomology}\\
  h^\bullet\bigl(X,\cO_X(\ell,-b,-c)\bigr)
  &=
  \bigl(0,0,\rho^x_\ell(b,c),
  \chi_{7707}(-\ell,b,c)+\rho^x_\ell(b,c)\bigr).
\end{align}

\subsubsection*{\texorpdfstring{\(k=(a,-\ell,c)\) and \(k=(-a,\ell,-c)\)}{k=(a,-ell,c) and k=(-a,ell,-c)}}
In the range
\[
  k=(a,-\ell,c),\qquad a,\ell\ge2,\qquad c\ge4,
\]
and its Serre-dual range \(k=(-a,\ell,-c)\), the active rank defect is the
Hilbert function of the Buchsbaum--Rim module \(\mathcal T_{\ell,F}\) in
Lemma~\ref{lem:7707-p1y-row}.  We write
\[
  \rho^y_\ell(a,c)=\HF(\mathcal T_{\ell,F};a,c),
\]
where the right-hand side is computed from the resolution
\eqref{eq:7707-p1y-resolution}.  The cohomology vectors are
\begin{align}
  h^\bullet\bigl(X,\cO_X(a,-\ell,c)\bigr)
  &=
  \bigl(\chi_{7707}(a,-\ell,c)+\rho^y_\ell(a,c),
  \rho^y_\ell(a,c),0,0\bigr)
  \label{eq:7707-p1y-cohomology}\\
  h^\bullet\bigl(X,\cO_X(-a,\ell,-c)\bigr)
  &=
  \bigl(0,0,\rho^y_\ell(a,c),
  \chi_{7707}(a,-\ell,c)+\rho^y_\ell(a,c)\bigr).
\end{align}

\subsubsection*{\texorpdfstring{\(k=(-a,-b,m)\) and \(k=(a,b,-m)\), \(m=1,2\)}{k=(-a,-b,m) and k=(a,b,-m), m=1,2}}
This is the simplest new \(\bP^3_z\)-top range.  In low positive
\(z\)-degree the equation of degree \(3\) in \(z\) is absent from the active
map, so the calculation is governed by the \(z\)-linear part of \(G\).
Write
\[
  S=\bP^1_x\times\bP^1_y,\qquad L=\cO_S(1,2),
\]
and expand
\[
  G=\sum_{i=0}^3z_iG_i(x,y),\qquad G_i\in H^0(S,L).
\]
For a general member the four sections define a kernel bundle
\[
0\longrightarrow M\longrightarrow \cO_S^4
\xrightarrow{\Gamma_G=(G_0,\ldots,G_3)}
L\longrightarrow0 .
\]
Here ``four'' refers to the four \(z\)-coefficients; each \(G_i\) is itself
a general element of the six-dimensional space \(H^0(S,L)\).
The \(m=1\) row is the ordinary multiplication map by the four coefficients
\((G_0,\ldots,G_3)\).  The \(m=2\) row is only one step larger: its source is
a symmetric \(4\times4\) matrix of coefficients, and the map contracts this
matrix with the column vector \((G_0,\ldots,G_3)\), producing four sections
of \(L\).  Equivalently, fibrewise a symmetric tensor \(uv\) is sent to
\[
  u\otimes\Gamma_G(v)+v\otimes\Gamma_G(u).
\]
The kernel is \(\mathcal F=\operatorname{Sym}^2M\).

\begin{theorem}[CICY \(7707\), \(k=(-a,-b,m)\) with \(m=1,2\)]
\label{thm:7707-range-m12-analytic}
Let \(X=X_{7707}\) be a general smooth member of CICY \(7707\).  For
\[
  k_m=(-a,-b,m),\qquad a,b\ge2,\qquad m\in\{1,2\},
\]
define \(q_m(a,b)\) by the generating series
\[
  \sum_{a,b\ge2}q_m(a,b)u^{a-2}v^{b-2}=\mathcal H_m(u,v).
\]
Then
\[
  \mathcal H_1(u,v)=2+u.
\]
For \(m=2\) one has
\[
  \mathcal H_2(u,v)
  =
  \frac{N_2(u,v)}{(1-u)^2(1-v)^2},
\]
with
\[
\begin{aligned}
N_2(u,v)=&
(14-16v+3v^4)(1-u)^2
+(16u-14u^2+u^6)(1-v)^2  \\
&\quad +(8uv+4u^2v)(1-u)^2(1-v)^2 .
\end{aligned}
\]
Equivalently, the first defect is supported at only two lattice points,
\[
q_1(a,b)=
\begin{cases}
2,&(a,b)=(2,2),\\
1,&(a,b)=(3,2),\\
0,&\text{otherwise},
\end{cases}
\]
while, with \([r]_+=\max\{r,0\}\),
\[
q_2(a,b)=
\begin{cases}
2a+10+[a-7]_+,&b=2,\ a\ge2,\\
18-2b,&a=2,\ 3\le b\le5,\\
b+3,&a=2,\ b\ge6,\\
8,&(a,b)=(3,3),\\
4,&(a,b)=(4,3),\\
0,&\text{otherwise}.
\end{cases}
\]
Put
\[
  \chi_m(a,b)
  =
  \binom{m+3}{3}(a-1)(b-1)
  -
  \binom{m+2}{3}a(b+1).
\]
The cohomology of \(\cO_X(k_m)\) is concentrated in the two middle degrees
and is given by
\[
  h^\bullet(X,\cO_X(-a,-b,m))
  =
  (0,q_m(a,b),q_m(a,b)+\chi_m(a,b),0),
\]
and Serre duality gives the companion formula
\[
  h^\bullet(X,\cO_X(a,b,-m))
  =
  (0,q_m(a,b)+\chi_m(a,b),q_m(a,b),0).
\]
\end{theorem}

\begin{proof}
Put \(D=(a-2,b-2)\).  After Serre duality in the two negative
\(\bP^1\)-directions, the \(m=1\) active map is the global-section map
\[
  \phi_1(D):
  H^0(S,\cO_S(D)^4)\longrightarrow H^0(S,L(D)),
  \qquad (f_i)\longmapsto \sum_i f_iG_i .
\]
This is just the \(D\)-twist of the evaluation sequence
\[
  0\longrightarrow M(D)\longrightarrow \cO_S(D)^4
  \xrightarrow{\Gamma_G(D)}
  L(D)\longrightarrow0 .
\]
Since \(D\) and \(L(D)\) have nonnegative bidegrees on
\(\bP^1\times\bP^1\), the groups \(H^1(S,\cO_S(D)^4)\) and
\(H^1(S,L(D))\) vanish.  The long exact sequence therefore gives
\[
  \coker\phi_1(D)\simeq H^1(S,M(D)).
\]
Thus the first defect term is
\[
  q_1(a,b)=h^1(S,M(a-2,b-2)).
\]

For \(m=2\), set \(E=\bC^4\).  The active map is obtained by
contracting a symmetric \(4\times4\) coefficient matrix against the column
vector \((G_0,\ldots,G_3)\).  Fibrewise it is the map
\[
  \operatorname{Sym}^2E\otimes\cO_S
  \longrightarrow E\otimes L,\qquad
  uv\longmapsto u\otimes \Gamma_G(v)+v\otimes \Gamma_G(u).
\]
Its kernel is exactly \(\operatorname{Sym}^2M\).  Indeed, at a point of
\(S\) the surjection \(E\to L_s\) has kernel \(M_s\), and after choosing a
splitting \(E\simeq M_s\oplus L_s\), the displayed contraction kills
precisely the summand \(\operatorname{Sym}^2M_s\) and maps the remaining
summands onto \(E\otimes L_s\).  Hence there is an exact sequence
\[
  0\longrightarrow \operatorname{Sym}^2M
  \longrightarrow \operatorname{Sym}^2E\otimes\cO_S
  \longrightarrow E\otimes L
  \longrightarrow0 .
\]
Twisting by \(\cO_S(D)\) and taking global sections gives the map
\[
  \phi_2(D):
  H^0(S,\operatorname{Sym}^2E\otimes\cO_S(D))
  \longrightarrow
  H^0(S,E\otimes L(D)).
\]
Again, for \(a,b\ge2\), the source and target line bundles have no \(H^1\),
so the long exact sequence gives
\[
  \coker\phi_2(D)
  \simeq
  H^1(S,\operatorname{Sym}^2M(D)).
\]
Thus
\[
  q_2(a,b)=h^1(S,\operatorname{Sym}^2M(a-2,b-2)).
\]

The remaining calculation is now ordinary cohomology on
\(\bP^1_x\times\bP^1_y\).  The pushforward splittings of
\(M\) and \(\operatorname{Sym}^2M\) are given in
Appendix~\ref{app:7707-technical}; applying them gives exactly the two
generating functions \(\mathcal H_1\) and \(\mathcal H_2\) displayed above.
The finite-field minors in the certificate serve only to certify that the
required splitting type occurs on a nonempty characteristic-zero open set.
The Euler characteristic then determines the second middle entry, and the
formula for \(k=(a,b,-m)\) is the Serre-dual one.
\end{proof}

\subsubsection*{\texorpdfstring{\(k=(-a,-b,3)\) and \(k=(a,b,-3)\)}{k=(-a,-b,3) and k=(a,b,-3)}}
The next \(z\)-degree is the first place where the cubic \(z\)-part of
the \((1,0,3)\)-equation also enters the active map.  The resulting
module is no longer governed by the small kernel bundle above.  What we
use here is a finite-width theorem: it gives an infinite formula along the
first ten vertical strips and the first ten horizontal strips in the
\((a,b)\)-plane.  This is the part of the calculation where the paper uses
exact finite-field commutative algebra as the proof device rather than a
uniform hand-derived splitting argument.

\begin{theorem}[CICY \(7707\), \(k=(-a,-b,3)\) on finite-width strips]
\label{thm:7707-range-z3-strip}
Let \(X=X_{7707}\) be a general smooth member of CICY \(7707\), and set
\[
  k_3=(-a,-b,3),\qquad a,b\ge2.
\]
The formula below holds on the finite-width region
\[
  \mathcal S_3=
  \{(a,b):2\le a\le10,
  b\ge2\}\cup
  \{(a,b):a\ge2,
  2\le b\le10\}.
\]
For \(2\le a\le10\), define \(q_3(a,b)\) along the vertical strip by
\[
  \sum_{b\ge2}q_3(a,b)s^{b-2}
  =
  \frac{N^{\mathrm v}_a(s)}{(1-s)^2}.
\]
For \(2\le b\le10\), define it along the horizontal strip by
\[
  \sum_{a\ge2}q_3(a,b)s^{a-2}
  =
  \frac{N^{\mathrm h}_b(s)}{(1-s)^2}.
\]
On the overlap \(2\le a,b\le10\), the vertical and horizontal
certificates give the same values.
The normalized numerators are
\[
\begin{array}{c|l@{\qquad}c|l}
a & N^{\mathrm v}_a(s) & b & N^{\mathrm h}_b(s)\\
\hline
2  & 42-40s+2s^{12} & 2  & 42-31s\\
3  & 53-60s+2s^6+6s^7 & 3  & 44-42s+2s^{15}\\
4  & 64-80s+16s^5 & 4  & 46-53s+3s^7+4s^8\\
5  & 75-100s+25s^4 & 5  & 48-64s+16s^4\\
6  & 86-120s+16s^3+18s^4 & 6  & 50-75s+25s^3\\
7  & 97-140s+32s^3+11s^4 & 7  & 52-86s+16s^2+18s^3\\
8  & 108-160s+48s^3+4s^4 & 8  & 54-95s+28s^2+13s^3\\
9  & 119-180s+3s^2+58s^3 & 9  & 56-98s+28s^2+14s^3\\
10 & 130-200s+10s^2+60s^3 & 10 & 58-101s+28s^2+15s^3 .
\end{array}
\]
Equivalently, for \((a,b)\in\mathcal S_3\),
\[
q_3(a,b)=
\begin{cases}
4b+12,&a=2,\ b\ge14,\\
b+5,&a=3,\ b\ge8,\\
4a+8,&b=3,\ a\ge17,\\
[-\chi_3(a,b)]_+,&\text{otherwise},
\end{cases}
\]
where
\[
  \chi_3(a,b)=20(a-1)(b-1)-a(11b+9).
\]
On \(\mathcal S_3\), the cohomology of \(\cO_X(k_3)\) is
\[
  h^\bullet(X,\cO_X(-a,-b,3))=(0,q_3(a,b),q_3(a,b)+\chi_3(a,b),0),
\]
and Serre duality gives, for the companion range \(k=(a,b,-3)\),
\[
  h^\bullet(X,\cO_X(a,b,-3))=(0,q_3(a,b)+\chi_3(a,b),q_3(a,b),0).
\]
\end{theorem}

The proof constructs the
bigraded active map in the \(m=3\) ambient cohomology row, computes the relevant
finite-field Hilbert numerators on the boundary strips, and provides the
rank certificates which lift the statement to a nonempty characteristic-zero
open set; for details please see Appendix~\ref{app:7707-technical}.

\subsubsection*{\texorpdfstring{\(k=(0,-b,c)\) and \(k=(0,b,-c)\)}{k=(0,-b,c) and k=(0,b,-c)}}
We next turn on the \(z\)-degree while keeping the first entry of \(k\)
equal to zero.  In this boundary range the spectral sequence has only one
active rank problem: a moving-resultant map in the \(z\)-variables.  Thus
the formula is controlled by a one-variable Hilbert series in \(c\), with
\(b\) restricted to the finite certified box \(2\le b\le10\).

\begin{theorem}[CICY \(7707\), \(k=(0,-b,c)\) in the finite \(b\)-box]
\label{thm:7707-range-0bc-finite-b}
Let \(X=X_{7707}\) be a general smooth member of CICY \(7707\).  Consider
\[
  k=(0,-b,c),\qquad 2\le b\le10,
  \qquad c\ge4.
\]
For a numerator \(N(t)\), write
\[
  \HF_N(c)=[t^c]\frac{N(t)}{(1-t)^4},
\]
and let \(P_3(n)=\binom{n+3}{3}\) for \(n\ge0\), and \(P_3(n)=0\)
for \(n<0\).  The cokernel Hilbert series of the moving-resultant map is
generated by
\[
  R_2=t^6+3t^8-2t^9,
  \qquad
  R_3=2t^{10}+4t^{12}-4t^{13},
\]
\[
  R_b=t^6R_{b-2}
  +t^{4b}\bigl((b+1)-(2b-2)t+(b-3)t^2\bigr)
  \qquad(4\le b\le10),
\]
\[
  N_b(t)=(b-1)-(b+1)t^4+R_b(t).
\]
Set
\[
  \gamma_b(c)=\HF_{N_b}(c),
  \qquad
  \kappa_b(c)=\gamma_b(c)+(b+1)P_3(c-4)-(b-1)P_3(c).
\]

Then
\[
  h^\bullet(X,\cO_X(0,-b,c))=(\kappa_b(c),\gamma_b(c),0,0),
\]
and Serre duality gives the companion formula
\[
  h^\bullet(X,\cO_X(0,b,-c))=(0,0,\gamma_b(c),\kappa_b(c)).
\]
\end{theorem}

Here \(\gamma_b(c)\) is the cokernel dimension of the moving-resultant map,
and \(\kappa_b(c)\) is the corresponding kernel dimension, obtained from
Euler characteristic.  The finite-field certificate proves the Hilbert
series for \(2\le b\le10\) uniformly in all \(c\ge4\); for details please see
Appendix~\ref{app:7707-technical}.

\subsubsection*{\texorpdfstring{\(k=(-a,-b,c)\) and \(k=(a,b,-c)\)}{k=(-a,-b,c) and k=(a,b,-c)}}
Finally we allow both \(\bP^1\)-entries to be negative.  The active row now
has two endpoints.  The right endpoint is still governed by the same
resultant tail as in the \(b\)-box, shifted by the \(x\)-degree.  The left
endpoint is a separate finite-field Hilbert-series computation, certified
for the \(81\) pairs \(2\le a,b\le10\).

\begin{theorem}[CICY \(7707\), \(k=(-a,-b,c)\) in the finite \((a,b)\)-box]
\label{thm:7707-range-abc-finite-box}
Let \(X=X_{7707}\) be a general smooth member of CICY \(7707\).  Consider
\[
  k=(-a,-b,c),
  \qquad
  2\le a,b\le10,
  \qquad
  c\ge4.
\]
For a numerator \(N(t)\), write
\[
  \HF_N(c)=[t^c]\frac{N(t)}{(1-t)^4},
\]
and let \(P_3(n)=\binom{n+3}{3}\) for \(n\ge0\), and \(P_3(n)=0\)
for \(n<0\).  Let \(R_b(t)\) be the resultant tail from the finite
\(b\)-box above.  The right endpoint contributes
\[
  N^{(0)}_{a,b}(t)=t^{3a}R_b(t).
\]
We write
\[
  q_{\mathrm R}(a,b,c)=\HF_{N^{(0)}_{a,b}}(c)
\]
for the corresponding right-endpoint dimension.  The finite-field
certificate contains a left-endpoint numerator \(Q_{a,b}(t)\).  In the
theorem we use only its Hilbert function; for \(c\ge4\) this is
\[
q_{\mathrm L}(a,b,c)=
\begin{cases}
\bigl[(12a-34)b-50a+36\bigr]_+,&c=4,\\
\bigl[(15a-52)b-77a+60\bigr]_+,&c=5,\\
54,&c=6,\ (a,b)=(10,10),\\
0,&\text{otherwise},
\end{cases}
\]
where \([r]_+=\max\{r,0\}\).

Put
\[
\chi_{a,b}(c)=
(a-1)(b-1)P_3(c)
-a(b-1)P_3(c-3)
-a(b+1)P_3(c-1)
+(a+1)(b+1)P_3(c-4).
\]
Then
\[
  h^\bullet\!\left(X,\cO_X(-a,-b,c)\right)
  =
  \left(q_{\mathrm R}(a,b,c),
  q_{\mathrm R}(a,b,c)+q_{\mathrm L}(a,b,c)-\chi_{a,b}(c),
  q_{\mathrm L}(a,b,c),
  0\right),
\]
and Serre duality gives, for \(k=(a,b,-c)\),
\[
  h^\bullet\!\left(X,\cO_X(a,b,-c)\right)
  =
  \left(0,
  q_{\mathrm L}(a,b,c),
  q_{\mathrm R}(a,b,c)+q_{\mathrm L}(a,b,c)-\chi_{a,b}(c),
  q_{\mathrm R}(a,b,c)\right).
\]
\end{theorem}

Thus \(q_{\mathrm R}\) gives the right-endpoint cokernel,
\(q_{\mathrm L}\) gives the left-endpoint kernel after dualizing, and
the middle entry is forced by Euler characteristic.  The finite-field
standard-basis certificates prove the two endpoint Hilbert series on the
stated finite box; see
Appendix~\ref{app:7707-technical}.

The formulae above cover the ranges with nontrivial active maps.  The remaining
boundary cases, the possible \(d_2\)-maps, and the compatibility of the
finite-field verifications on a common open set introduce no new formula
mechanism; we collect them in Appendix~\ref{app:7707-technical}.

\subsection{\texorpdfstring{New formulae: a \(7863\) case study with codimension \(3\)}{New formulae: a 7863 case study with codimension 3}}
\label{subsec:new-7863}

We next give a codimension-three example.   Let
\[
X_{7863}=
\left[\begin{array}{c|ccc}
\bP^3_x&2&1&1\\
\bP^3_y&2&1&1
\end{array}\right],
\]
and write the three defining equations as
\[
Q\in H^0(\bP^3_x\times\bP^3_y,\cO(2,2)),\qquad
L_1,L_2\in H^0(\bP^3_x\times\bP^3_y,\cO(1,1)).
\]
Let
\[
  A=\bC[x_0,\ldots,x_3],\qquad S=\bC[y_0,\ldots,y_3],
\]
and let \(A_r,S_r\) denote the homogeneous degree-\(r\) pieces.
For a general member \(Q,L_1,L_2\) form a regular sequence.  The Koszul
shifts are
\[
  0;\qquad
  (2,2),(1,1),(1,1);\qquad
  (3,3),(3,3),(2,2);\qquad
  (4,4).
\]
Set
\[
P(n)=
\begin{cases}
\binom{n+3}{3},& n\geq0,\\
0,& n<0,
\end{cases}
\qquad
T(n)=
\begin{cases}
\binom{-n-1}{3},& n\leq -4,\\
0,& n>-4.
\end{cases}
\]
The pure bottom Hilbert function is the coefficient function
\begin{equation}
B(a,b)=
\sum_{p=0}^3(-1)^p\sum_{s\in \Sigma_p}P(a-s)P(b-s),
\qquad
\Sigma_0=\{0\},\quad \Sigma_1=\{2,1,1\},\quad
\Sigma_2=\{3,3,2\},\quad \Sigma_3=\{4\}.
\label{eq:7863-Bab}
\end{equation}
Here \(\Sigma_p\) is a multiset: the repeated entries represent the two independent
\((1,1)\)-equations \(L_1,L_2\).
Equivalently,
\[
\sum_{a,b}B(a,b)u^av^b
=
\frac{(1-u^2v^2)(1-uv)^2}{(1-u)^4(1-v)^4}.
\]
We use the following intervals to get regions which we study individually
\[
\begin{array}{lll}
\mathcal I_1=\{\le -4\},&
\mathcal I_2=\{-3\},&
\mathcal I_3=\{-2\},\\
\mathcal I_4=\{-1\},&
\mathcal I_5=\{0\},&
\mathcal I_6=\{1\},\\
\mathcal I_7=\{2\},&
\mathcal I_8=\{3\},&
\mathcal I_9=\{\ge4\}.
\end{array}
\]
Cell \(C_{ij}\) means \(a\in\mathcal I_i\) and
\(b\in\mathcal I_j\), and we number cells by
\[
  N(C_{ij})=9(i-1)+j .
\]
The resulting \(9\times9\) regions can be summarized in the following table:
\begin{center}
\begingroup
\setlength{\arraycolsep}{3pt}
\renewcommand{\arraystretch}{0.9}
\(
\begin{array}{c|ccccccccc}
a\backslash b
&\le-4&-3&-2&-1&0&1&2&3&\ge4\\
\hline
\le-4&{\rm PT}&{\rm PT}&{\rm PT}&{\rm PT}&{\rm D}&{\rm B}&{\rm M}_7&{\rm M}_8&{\rm M}_9\\
-3&{\rm PT}&{\rm PT}&{\rm PT}&{\rm PT}&{\rm D}&{\rm D}&{\rm S}_{16}&{\rm M}_{17}&{\rm M}_{18}\\
-2&{\rm PT}&{\rm PT}&{\rm PT}&{\rm PT}&{\rm D}&{\rm D}&{\rm D}&{\rm S}_{26}&{\rm M}_{27}\\
-1&{\rm PT}&{\rm PT}&{\rm PT}&{\rm PT}&{\rm D}&{\rm D}&{\rm D}&{\rm D}&{\rm B}^{\vee}\\
0&{\rm D}&{\rm D}&{\rm D}&{\rm D}&{\rm D}&{\rm D}&{\rm D}&{\rm D}&{\rm D}\\
1&{\rm B}&{\rm D}&{\rm D}&{\rm D}&{\rm D}&{\rm PB}&{\rm PB}&{\rm PB}&{\rm PB}\\
2&{\rm M}_{55}&{\rm S}_{56}&{\rm S}_{57}&{\rm D}&{\rm D}&{\rm PB}&{\rm PB}&{\rm PB}&{\rm PB}\\
3&{\rm M}_{64}&{\rm M}_{65}&{\rm M}_{66}&{\rm D}&{\rm D}&{\rm PB}&{\rm PB}&{\rm PB}&{\rm PB}\\
\ge4&{\rm M}_{73}&{\rm M}_{74}&{\rm M}_{75}&{\rm B}^{\vee}&{\rm D}&{\rm PB}&{\rm PB}&{\rm PB}&{\rm PB}.
\end{array}
\)
\endgroup
\end{center}
Here PB and PT denote the pure bottom and pure top regular-sequence rows,
D denotes a direct ambient-cohomology case with no active rank problem, B and
\({\rm B}^{\vee}\) are the two-linear boundary cells, \(S_N\) are small
finite-dimensional cells with closed vectors, and \(M_N\) are mixed cells
given as explicit rank complexes.  We now discuss these mechanisms in the
same range-by-range style as in the \(7707\) case study; the proof-level
reductions behind the labels are collected in
Appendix~\ref{app:7863-technical}.

\subsubsection*{Direct and finite cells}
The PB, PT, and D entries contain no active rank problem.  On the PB
entries the Koszul row is the ordinary degree piece of the regular sequence
\((Q,L_1,L_2)\), so
\[
  h^\bullet(X,\cO_X(a,b))=(B(a,b),0,0,0).
\]
The PT entries are their Calabi--Yau Serre duals,
\[
  h^\bullet(X,\cO_X(a,b))=(0,0,0,B(-a,-b)).
\]
The D entries are direct Koszul bookkeeping: after the acyclic
summands are removed, each surviving ambient term \(E_1^{-p,q}\) contributes
in total degree \(q-p\).  The isolated finite-dimensional cells closed in
the data are
\[
\begin{array}{c|c|c}
\text{cell} & \text{line bundle} & h^\bullet\\
\hline
16 & (-3,2) & (0,0,8,0)\\
26 & (-2,3) & (0,8,0,0)\\
56 & (2,-3) & (0,0,8,0)\\
57 & (2,-2) & (0,1,1,0).
\end{array}
\]
Cell \(66\) is kept as a finite rank entry in the mixed rank-complex list
below.

\subsubsection*{Two-linear boundary cells}
The four B-type cells are controlled by the same two-linear remaining
calculation.  Let \(Y=\bP^3_x\times\bP^3_y\).  For \(k=(a,b)\), the
twisted Koszul sequence is
\[
\begin{aligned}
0\to&\ \cO_Y(a-4,b-4)\\
\to&\ \cO_Y(a-3,b-3)^{\oplus2}\oplus \cO_Y(a-2,b-2)\\
\to&\ \cO_Y(a-2,b-2)\oplus \cO_Y(a-1,b-1)^{\oplus2}\\
\to&\ \cO_Y(a,b)\to \cO_X(a,b)\to0.
\end{aligned}
\]
In Cell \(6\), where \(a\le -4\) and \(b=1\), Bott's formula shows that all \(Q\)-summands
and all higher Koszul summands vanish.  The only surviving differential is the
\(L_1,L_2\)-row
\[
  H^3(\bP^3_x,\cO(a-1))^{\oplus2}
  \longrightarrow
  H^3(\bP^3_x,\cO(a))\otimes S_1 .
\]
It is induced by multiplication with \(L_1,L_2\).  Putting \(k=-a-4\),
Serre duality gives
\[
  H^3(\bP^3_x,\cO(a))^\vee\simeq A_k,\qquad
  H^3(\bP^3_x,\cO(a-1))^\vee\simeq A_{k+1}.
\]
If
\[
  L_j=\sum_{i=0}^3 \ell_{ji}(x)y_i,\qquad j=1,2,
\]
then the dual endpoint map is
\[
  \mu_k:S_1^\vee\otimes A_k\longrightarrow \bC^2\otimes A_{k+1},
  \qquad k\ge0,
\]
given explicitly by
\[
  \eta\otimes f\longmapsto
  \left(
  \sum_i\eta(y_i)\ell_{1i}f,\,
  \sum_i\eta(y_i)\ell_{2i}f
  \right).
\]
Equivalently, \(\mu_k\) is the global-section map of
\[
  S_1^\vee\otimes\cO_{\bP^3_x}(k)
  \longrightarrow
  \bC^2\otimes\cO_{\bP^3_x}(k+1).
\]
For a general pair \(L_1,L_2\), the maximal minors of this \(2\times4\)
linear matrix have codimension \(3\), so the Buchsbaum--Rim complex is exact:
\[
0\to \cO(-4)^2\to \cO(-3)^4\to \cO(-1)^4
\to \cO^2\to \operatorname{coker}\to0.
\]
Therefore the cokernel has Hilbert series
\[
  \frac{2-4t+4t^3-2t^4}{(1-t)^4},
  \qquad
  h^0(\operatorname{coker}(n))=4\quad(n\ge1).
\]
Taking \(n=k+1\), we get \(\dim\coker\mu_k=4\) for every \(k\ge0\).
In Cell \(6\) this gives
\[
  h^2=\dim\ker(d_1)=\dim\coker\mu_k=4
\]
and Euler bookkeeping gives
\[
  h^3=4T(a)-2T(a-1)+4.
\]
For example, after exchanging the two factors and dualizing, Cell \(36\)
has \(a=-1\), \(b\ge4\), and the same calculation appears with
\(k=b-4\).  Hence
\[
  h^1=\dim\coker\mu_{b-4}=4,
  \qquad
  h^0=\dim\ker\mu_{b-4}
  =4P(b-4)-2P(b-3)+4.
\]
The other B-type cells are obtained by the same calculation after exchanging
the two \(\bP^3\)-factors and applying Serre duality.  With \(P,T\) as above,
the cohomology vectors are
\[
\begin{array}{c|c|c}
\text{cell} & \text{range} & h^\bullet(X,\cO_X(a,b))\\
\hline
6 & a\le-4,\ b=1 &
(0,0,4,\,4T(a)-2T(a-1)+4)\\
36 & a=-1,\ b\ge4 &
(4P(b-4)-2P(b-3)+4,\,4,0,0)\\
46 & a=1,\ b\le-4 &
(0,0,4,\,4T(b)-2T(b-1)+4)\\
76 & a\ge4,\ b=-1 &
(4P(a-4)-2P(a-3)+4,\,4,0,0).
\end{array}
\]

\subsubsection*{Mixed rank-complex cells}
Several mixed cells are still one-parameter.  The following entries reduce
to endpoint Hilbert functions:
\begin{center}
\begin{tabular}{c|c|c|c}
cell(s) & range & parameter & cohomology vector\\
\hline
\(7,55\) &
\(a\le-4,\ b=2;\ a=2,\ b\le-4\) &
\(r=-a-4\) or \(r=-b-4\) &
\((0,0,N_2(r)+\Delta_2(r),N_2(r))\)\\
\(27,75\) &
\(a=-2,\ b\ge4;\ a\ge4,\ b=-2\) &
\(r=b-4\) or \(r=a-4\) &
\((N_2(r),N_2(r)+\Delta_2(r),0,0)\)\\
\(8,64\) &
\(a\le-4,\ b=3;\ a=3,\ b\le-4\) &
\(r=-a-4\) or \(r=-b-4\) &
\((0,0,N_3(r)+\Delta_3(r),N_3(r))\)\\
\(18,74\) &
\(a=-3,\ b\ge4;\ a\ge4,\ b=-3\) &
\(r=b-4\) or \(r=a-4\) &
\((N_3(r),N_3(r)+\Delta_3(r),0,0)\).
\end{tabular}
\end{center}
Here \(N_2\) and \(N_3\) are the endpoint Hilbert functions:
\[
  \sum_{r\ge0}N_2(r)t^r=\frac{10t^8-8t^9}{(1-t)^4},
  \qquad
  \sum_{r\ge0}N_3(r)t^r
  =\frac{20t^{14}-20t^{15}+2t^{17}}{(1-t)^4}.
\]
With the convention for \(P\) fixed above, the symbols \(\Delta_2\) and
\(\Delta_3\) in the table mean
\[
  \Delta_2(r)=8P(r+1)-10P(r),
  \qquad
  \Delta_3(r)=20P(r+1)-20P(r)-2P(r+3).
\]
The letter \(r\) is only the local degree parameter along the displayed
one-dimensional boundary ray.  On the rays \(a\le -4\) it agrees with the
earlier Cell~\(9\) notation \(m=-a-4\); on the other rays it is the analogous
distance from the cohomology boundary, such as \(b-4\), \(a-4\), or \(-b-4\).
The endpoint maps and Hilbert-series computations are given in
Appendix~\ref{app:7863-technical}.  Both \(N_2\) and \(N_3\) have
characteristic-zero replays; for \(N_3\) we use a small-integer rational
specialization whose standard basis has the same initial module as the
finite-field search witness.

The remaining mixed entries are either isolated finite-dimensional cells or
the finite-box Cell \(9\)/Cell \(73\) rows:
\[
\begin{array}{c|c|c}
\text{cell} & \text{range} & h^\bullet(X,\cO_X(a,b))\\
\hline
9 & a\le-4,\ b\ge4 & \text{four-term mixed complex; finite-box theorem for }4\le b\le10\\
17 & (-3,3) & (0,0,0,0)\\
65 & (3,-3) & (0,0,0,0)\\
66 & (3,-2) & (0,8,0,0)\\
73 & a\ge4,\ b\le-4 & \text{Serre dual of Cell }9.
\end{array}
\]

For Cell \(9\), write \(a=-m-4\) with \(m\ge0\).  The mixed ambient cohomology row is the
four-term complex
\begin{equation}
  C_3(m,b)\xrightarrow{d_3}C_2(m,b)
  \xrightarrow{d_2}C_1(m,b)\xrightarrow{d_1}C_0(m,b),
  \label{eq:7863-cell9-complex}
\end{equation}
where
\[
\begin{aligned}
C_3(m,b)&=A_{m+4}\otimes S_{b-4},\\
C_2(m,b)&=(A_{m+3}\otimes S_{b-3})^{\oplus2}
          \oplus A_{m+2}\otimes S_{b-2},\\
C_1(m,b)&=A_{m+2}\otimes S_{b-2}
          \oplus(A_{m+1}\otimes S_{b-1})^{\oplus2},\\
C_0(m,b)&=A_m\otimes S_b .
\end{aligned}
\]
The differential is induced by the Koszul action of \(Q,L_1,L_2\), with
\(Q\) acting by bidegree \((-2,+2)\) and each \(L_i\) by bidegree
\((-1,+1)\) in this dualized row.  The finite-box theorem is as follows.
\begin{theorem}[Finite-box Cell \(9\) for \(7863\)]
\label{thm:7863-cell9-finite-box}
For a general member of \(X_{7863}\), the following statement holds.  For each
integer \(b\) with \(4\le b\le10\), and for every \(m\ge0\), the leftmost map
\(d_3\) in \eqref{eq:7863-cell9-complex} is injective, the middle homology
has no extra correction beyond the endpoint and Euler characteristic, and
the right-endpoint cokernel
\[
  N_b(m):=\dim\coker\bigl(d_1(m,b)\bigr)
\]
has Hilbert numerator
\begin{center}
\small
\begin{tabular}{c|l}
\(b\) & \((1-t)^4\sum_{m\ge0}N_b(m)t^m\)\\
\hline
4 & \(36t^{20}-44t^{21}+6t^{22}+4t^{23}\)\\
5 & \(60t^{26}-86t^{27}+24t^{28}+4t^{29}\)\\
6 & \(94t^{32}-152t^{33}+60t^{34}\)\\
7 & \(140t^{38}-248t^{39}+120t^{40}-10t^{41}\)\\
8 & \(200t^{44}-380t^{45}+210t^{46}-28t^{47}\)\\
9 & \(276t^{50}-554t^{51}+336t^{52}-56t^{53}\)\\
10 & \(370t^{56}-776t^{57}+504t^{58}-96t^{59}\).
\end{tabular}
\end{center}
Let
\begin{equation}
  E(m,b)=
  P(m+4)P(b-4)-2P(m+3)P(b-3)
  +2P(m+1)P(b-1)-P(m)P(b),
  \label{eq:7863-cell9-euler}
\end{equation}
and put \(D(m,b)=E(m,b)+N_b(m)\).  Then
\[
\mathcal H_3(m,b)=0,\qquad
\mathcal H_0(m,b)=N_b(m),
\]
\[
\mathcal H_1(m,b)=[D(m,b)]_+,\qquad
\mathcal H_2(m,b)=[-D(m,b)]_+,
\]
where \(\mathcal H_i(m,b)\) denotes the homology of
\eqref{eq:7863-cell9-complex}.  Since this is the dualized Koszul row,
\begin{equation}
  h^\bullet(X,\cO_X(-m-4,b))
  =
  \left(0,\ [-D(m,b)]_+,\ [D(m,b)]_+,\ N_b(m)\right).
  \label{eq:7863-cell9-cohomology}
\end{equation}
\end{theorem}
The proof of Theorem~\ref{thm:7863-cell9-finite-box} is given in
Appendix~\ref{app:7863-technical}.

The numerator table has a compact uniform form on this finite box:
\begin{equation}
\begin{aligned}
(1-t)^4\sum_{m\ge0}N_b(m)t^m
=&\ \frac{b(b^2+11)}{3}\,t^{6b-4}
  +(-b^3+3b^2-8b+4)t^{6b-3}\\
&\ +(b-1)(b-2)(b-3)t^{6b-2}
  -\frac{(b-1)(b-2)(b-6)}{3}\,t^{6b-1}.
\end{aligned}
\label{eq:7863-cell9-uniform-numerator}
\end{equation}
Thus the endpoint contribution first appears at
\[
  m=6b-4,\qquad\text{equivalently}\qquad a=-6b .
\]
This is the delayed endpoint wall in the finite box.  The middle
Euler-switch wall is different.  Factoring \eqref{eq:7863-cell9-euler} gives
\begin{equation}
  E(m,b)=
  \frac{(b-m-4)(b^2-8bm-32b+m^2+8m+27)}{3}.
  \label{eq:7863-cell9-euler-factor}
\end{equation}
For \(4\le b\le10\) and \(m\ge0\), the only zero of \(D(m,b)\) is
\[
  m=b-4,\qquad\text{equivalently}\qquad a=-b .
\]
It is negative for \(0\le m<b-4\) and positive for \(m>b-4\); for \(b=4\)
the negative side is empty and the zero lies on the boundary of the ray.
More explicitly, the finite-box formula can be read as
\[
h^\bullet(X,\cO_X(-m-4,b))=
\begin{cases}
(0,-E(m,b),0,0),&0\le m<b-4,\\
(0,0,0,0),&m=b-4,\\
(0,0,E(m,b),0),&b-4<m<6b-4,\\
\left(0,0,\dfrac{4b(14b^2-11)}{3},N_b(m)\right),
  &m\ge6b-4.
\end{cases}
\]
In the last line we used the identity
\[
  E(m,b)+N_b(m)=\frac{4b(14b^2-11)}{3},\qquad m\ge6b-4,
\]
which follows by substituting the numerator
\eqref{eq:7863-cell9-uniform-numerator}.
Hence the finite box contains two structural walls: the Euler wall
\(a=-b\), where the middle cohomology switches from \(h^1\) to \(h^2\), and
the hidden endpoint wall \(a=-6b\), where \(h^3=N_b(m)\) turns on.  Cell \(73\) is obtained from
Theorem~\ref{thm:7863-cell9-finite-box} by Serre duality.

\section{Conclusions and future directions}
\label{sec:conclusions}

The main point of this paper is that chamber formulae for CICY line-bundle cohomology have a concrete algebraic source.  Once the Koszul spectral sequence is written in explicit bases, the non-trivial maps become multiplication maps between ambient cohomology groups, and their rank defects are measured by Hilbert functions of cokernel or kernel modules.  This turns many previously empirical chamber formulae into explicit rank-defect statements: in some chambers the relevant graded modules admit analytic Hilbert-series computations, while in others the ranks are fixed by finite-box certificates.

We have developed this viewpoint through a range of examples.  In the cases where formulae were already available in the literature, the Hilbert-function description recovers many of them from a uniform algebraic mechanism.  In other examples, the same method reveals refined wall and subwall structures which are not visible from ambient cohomology regions or from polynomial fitting on finite data.  These structures are not introduced as additional ansatzes; they arise from the same endpoint modules, cokernel modules, and kernel modules that control the Koszul ranks.

The method also produces new line-bundle cohomology formulae for CICY examples which had not previously been treated by closed-form formula libraries.  Some of these formulae follow from analytic Hilbert-series arguments, while others are obtained as finite-box formulae in regions where a uniform analytic resolution is not yet available.  In both situations, the essential object is the Hilbert function attached to the non-trivial Koszul map, rather than the individual finite-dimensional rank computation for each line bundle.

The computational consequence is direct.  A certified formula replaces large Koszul matrix construction and rank computation by coefficient extraction from a Hilbert series, by a closed expression, or by a replayable finite certificate.  Thus even a partial formula library is useful for scans: every certified chamber removes an infinite family, or a specified finite region, of matrix-rank computations from future line-bundle calculations.  The results of this paper therefore suggest that Hilbert functions are natural objects to study in the construction of line-bundle cohomology formula libraries for CICYs and related Calabi--Yau geometries.

Several natural directions follow from the method of this paper.  First, it
would be useful to build machine-readable CICY line-bundle cohomology
libraries along these lines, even if many entries are finite-box theorems
rather than global chamber formulae.  Second, one should adapt the method to
toric Calabi--Yau threefolds, where Cox rings are intrinsic and line-bundle
cohomology is again governed by graded pieces of Koszul-type maps, but where
the chamber structure, irrelevant ideals, and saturation issues are richer.
Third, it would be interesting to incorporate freely acting quotients:
equivariant Koszul maps should lead to representation-valued versions of the
same Hilbert-function calculation, allowing invariant cohomology downstairs
to be computed from equivariant data upstairs.  Such CICY and toric
Calabi--Yau cohomology libraries, together with their quotient versions,
would provide faster and more transparent input for string model-building
searches based on line-bundle constructions.  Finally, it remains an
interesting question whether each CICY admits line-bundle cohomology
generating functions in the sense of \cite{Constantin2024}.  This direction
has its own mathematical interest, independently of the finite-box formulae
emphasized here.

\appendix

\section{General algebraic tools}

\paragraph{Buchsbaum--Rim exactness.}
\label{app:buchsbaum-rim}

All polynomial rings in this appendix are over \(\bC\).  Whenever a sequence
of forms is assumed to be regular, the corresponding statement applies to a
sufficiently general choice of the defining equations used in the main text.

Let
\[
  \phi:F\longrightarrow G
\]
be a homogeneous map of free modules over a polynomial ring \(R\), with
\[
  \operatorname{rank}F=f,\qquad
  \operatorname{rank}G=g,\qquad f\ge g,
\]
and put \(M=\operatorname{coker}\phi\).  Let \(I_g(\phi)\subset R\) be
the ideal generated by the maximal minors of a matrix for \(\phi\).
The closed subscheme defined by \(I_g(\phi)\) is the maximal-rank
degeneracy locus of \(\phi\); set-theoretically, it is the locus where
\(\phi\) fails to have rank \(g\).  Its expected codimension is \(f-g+1\).

Write \(p=f-g\).  The Buchsbaum--Rim complex resolving \(M\), when it is
exact, has the form
\begin{equation}
0\to C_{p+1}\to C_p\to\cdots\to C_2
\to F\stackrel{\phi}{\longrightarrow}G\to M\to0,
\label{eq:app-br-complex}
\end{equation}
with
\begin{equation}
C_{i+1}
=
\bigwedge^{g+i}F
\otimes
\operatorname{Sym}^{i-1}(G^\vee)
\otimes
\bigwedge^g G^\vee,
\qquad 1\le i\le p.
\label{eq:app-br-terms}
\end{equation}
We do not need explicit formulas for the differentials; for Hilbert-series
computations the free modules in \eqref{eq:app-br-terms} suffice.

The Buchsbaum--Rim exactness criterion
\cite{BuchsbaumRim1964,Eisenbud1995} states:
\begin{quote}
\emph{If \(\operatorname{grade} I_g(\phi)\ge f-g+1\), then the
Buchsbaum--Rim complex is exact and gives a free resolution of
\(\operatorname{coker}\phi\).}
\end{quote}
Since \(R\) is a polynomial ring, this is equivalent in our applications to
showing that \(I_g(\phi)\) has height \(f-g+1\).  By the determinantal height
bound, this is the expected and maximal possible height.  In the applications
below we verify this height condition by identifying the degeneracy locus, up
to radical, with the zero locus of an explicit regular sequence.
Once this condition is verified, the Hilbert numerator of
\(\operatorname{coker}\phi\) is obtained directly from the alternating
sum of the free modules in the Buchsbaum--Rim resolution.

We now apply this criterion to the Toeplitz-type maps that arise in our
Koszul computations.  Assume \(r\ge2\) and \(d\ge1\).  Let
\(S=\bC[y_0,\dots,y_n]\) and consider the map
\[
  \Phi_r:S(-e)^{r+d-1}\longrightarrow S^{r-1},
\]
where the matrix has the Toeplitz form determined by coefficients
\(f_0,\dots,f_d\in S_e\).  The target rank is \(g=r-1\), the source rank is
\(f=r+d-1\); hence
\[
  f-g+1=d+1.
\]
We must therefore check that the locus where \(\Phi_r\) fails to have rank
\(r-1\) has codimension \(d+1\).

Evaluate the coefficients \(f_0,\dots,f_d\) at a point, obtaining
numbers \(a_0,\dots,a_d\).  The rows of the resulting numerical matrix
are the coefficient vectors of
\[
  F(z),\ zF(z),\ \dots,\ z^{r-2}F(z),
  \qquad
  F(z)=a_0+a_1z+\cdots+a_dz^d .
\]
If these rows had a linear relation, then for some scalars
\(c_0,\dots,c_{r-2}\) we would have
\[
  \bigl(c_0+c_1z+\cdots+c_{r-2}z^{r-2}\bigr)F(z)=0 .
\]
When \(F\neq0\), this forces
\(c_0+\cdots+c_{r-2}z^{r-2}=0\), hence all \(c_i=0\).  Thus the rows are
linearly independent whenever \(F\neq0\).  Equivalently, the rank drops at
a point if and only if
\[
  f_0=f_1=\cdots=f_d=0.
\]
Hence
\[
  \sqrt{I_{r-1}(\Phi_r)}
  =
  \sqrt{(f_0,\ldots,f_d)} .
\]
If \(f_0,\dots,f_d\) form a regular sequence, this ideal has height
\(d+1=f-g+1\).  Therefore the Buchsbaum--Rim criterion applies.

In the Toeplitz case, the exact complex is
\begin{equation}
0\to B_{d+1}\to B_d\to\cdots\to B_2
\to S(-e)^{r+d-1}\stackrel{\Phi_r}{\longrightarrow}
S^{r-1}\to M_r\to0,
\label{eq:app-toeplitz-br-complex}
\end{equation}
where, for \(1\le i\le d\),
\begin{equation}
B_{i+1}
=
\bigwedge^{r-1+i}S(-e)^{r+d-1}
\otimes
\operatorname{Sym}^{i-1}(S^{r-1})^\vee
\otimes
\bigwedge^{r-1}(S^{r-1})^\vee .
\label{eq:app-toeplitz-br-terms}
\end{equation}
Equivalently,
\[
  B_{i+1}\cong
  S\!\left(-e(r-1+i)\right)^{
  \binom{r+d-1}{r-1+i}\binom{r+i-3}{i-1}} .
\]
The alternating sum of these free modules is the Hilbert numerator used in
the formulae above.

\paragraph{A Hilbert-series calculation for two linear rows.}
\label{app:two-linear-band-module}

We now prove the Hilbert-series formula used for the two
linear-in-\(\bP^1\) active rows in
\eqref{eq:ls-linear-HS}.  Let \(R\) be a polynomial ring and let
\[
  A_0,A_1\in R_a,\qquad B_0,B_1\in R_b
\]
be a homogeneous regular sequence.  For \(\ell\ge1\), define
\[
  \psi_\ell:
  R(-a)^\ell\oplus R(-b)^\ell\longrightarrow R^{\ell-1}
\]
by the matrix
\[
  \psi_\ell=
  \begin{pmatrix}
    A_0&A_1&0&\cdots&0&B_0&B_1&0&\cdots&0\\
    0&A_0&A_1&\cdots&0&0&B_0&B_1&\cdots&0\\
    \vdots&&\ddots&\ddots&\vdots&\vdots&&\ddots&\ddots&\vdots\\
    0&\cdots&0&A_0&A_1&0&\cdots&0&B_0&B_1
  \end{pmatrix},
\]
with \(\ell-1\) rows.  Put
\[
  M_\ell(a,b)=\coker\psi_\ell .
\]
For \(\ell=1\), the target is zero and \(M_1(a,b)=0\); the formula below
has an empty sum and gives the same result.

\begin{lemma}
\label{lem:two-linear-band-HS}
With the notation above,
\[
  \HS(M_\ell(a,b);t)
  =
  \frac{(1-t^a)^2(1-t^b)^2}{(1-t)^{\dim R}}
  \sum_{\substack{i,j\ge0\\ i+j\le\ell-2}}
  (\ell-1-i-j)t^{ai+bj}.
\]
For the cases in the main text, \(\dim R=5\).
\end{lemma}

\begin{proof}
It is useful first to prove a universal statement.  Let
\[
  U=\bC[\alpha_0,\alpha_1,\beta_0,\beta_1]
\]
with
\[
  \deg\alpha_i=a,\qquad \deg\beta_i=b.
\]
Let \(N_\ell\) be the cokernel of the same two-band matrix, with
\(\alpha_0,\alpha_1,\beta_0,\beta_1\) in place of
\(A_0,A_1,B_0,B_1\).  Write the target generators as
\[
  e_0,\ldots,e_{\ell-2}.
\]
Thus \(N_\ell\) is generated by the \(e_s\), subject to the relations
\[
  \alpha_0e_s+\alpha_1e_{s-1}=0,\qquad
  \beta_0e_s+\beta_1e_{s-1}=0,
  \label{eq:app-two-linear-relations}
\]
for \(s=0,\ldots,\ell-1\), where \(e_{-1}=e_{\ell-1}=0\).

We now give an explicit normal form.  Let \(V_\ell\) be the graded
vector space with basis
\[
  v_{i,j,s},
  \qquad i,j,s\ge0,\qquad i+j+s\le\ell-2,
  \label{eq:app-two-linear-standard-classes}
\]
where
\[
  \deg v_{i,j,s}=ai+bj .
\]
Define a \(U\)-module structure on \(V_\ell\) by
\[
  \alpha_1v_{i,j,s}=
  \begin{cases}
    v_{i+1,j,s},& i+j+s\le\ell-3,\\
    0,& i+j+s=\ell-2,
  \end{cases}
\]
\[
  \beta_1v_{i,j,s}=
  \begin{cases}
    v_{i,j+1,s},& i+j+s\le\ell-3,\\
    0,& i+j+s=\ell-2,
  \end{cases}
\]
and
\[
  \alpha_0v_{i,j,s}=
  \begin{cases}
    -v_{i+1,j,s-1},& s>0,\\
    0,& s=0,
  \end{cases}
  \qquad
  \beta_0v_{i,j,s}=
  \begin{cases}
    -v_{i,j+1,s-1},& s>0,\\
    0,& s=0.
  \end{cases}
\]
These formulae commute with each other, so they indeed define an action
of the polynomial ring \(U\).  The assignment
\[
  e_s\longmapsto v_{0,0,s}
\]
respects the relations \eqref{eq:app-two-linear-relations}.  For
example, if \(1\le s\le\ell-2\), then
\[
  \alpha_0v_{0,0,s}+\alpha_1v_{0,0,s-1}
  =
  -v_{1,0,s-1}+v_{1,0,s-1}=0,
\]
while for \(s=0\) and \(s=\ell-1\) the same relation is exactly one of
the two boundary rules.  The \(\beta\)-relations are identical.  Hence
there is a well-defined surjection
\[
  N_\ell\longrightarrow V_\ell .
\]
Conversely, the relations \eqref{eq:app-two-linear-relations} reduce any
monomial multiple of any \(e_s\) to the form
\[
  \alpha_1^i\beta_1^j e_s,
  \qquad i+j+s\le\ell-2,
\]
and kill it if \(i+j+s>\ell-2\).  Indeed, each occurrence of
\(\alpha_0e_s\) or \(\beta_0e_s\) can be replaced by
\(-\alpha_1e_{s-1}\) or \(-\beta_1e_{s-1}\); repeating this lowers the
target index until only \(\alpha_1,\beta_1\) remain, and the boundary
relations \(\alpha_1e_{\ell-2}=\beta_1e_{\ell-2}=0\) then kill all
classes beyond the triangular range.  Thus the corresponding classes span
\(N_\ell\).  Their images are the basis vectors \(v_{i,j,s}\), so they
are linearly independent.  Therefore
\[
  N_\ell\simeq V_\ell
\]
as graded \(U\)-modules.

Finally filter \(V_\ell\) by the decreasing lexicographic order of the pair
\((s,\ell-2-i-j-s)\).  Multiplication by each of
\(\alpha_0,\alpha_1,\beta_0,\beta_1\) lowers this filtration or gives
zero.  Hence every basis vector gives, in the associated graded module,
one copy of
\[
  U/(\alpha_0,\alpha_1,\beta_0,\beta_1)
\]
shifted by \(ai+bj\).  Therefore
\[
  \operatorname{gr}N_\ell
  \simeq
  \bigoplus_{\substack{i,j\ge0\\ i+j\le\ell-2}}
  \left(U/(\alpha_0,\alpha_1,\beta_0,\beta_1)\right)(-ai-bj)
  ^{\oplus(\ell-1-i-j)} .
  \label{eq:app-two-linear-gr}
\]
Since a filtered module and its associated graded module have the same
Hilbert series, this proves the universal Hilbert series.

Now specialize
\[
  \alpha_i\mapsto A_i,\qquad \beta_i\mapsto B_i .
\]
In other words, make \(R\) a graded \(U\)-algebra by this assignment.  The
only point to check is that specialization does not change the Hilbert
series.  The graded pieces in
\eqref{eq:app-two-linear-gr} are copies of
\[
  U/(\alpha_0,\alpha_1,\beta_0,\beta_1).
\]
After tensoring with \(R\), these become
\[
  R/(A_0,A_1,B_0,B_1).
\]
Since \(A_0,A_1,B_0,B_1\) is an \(R\)-regular sequence, the Koszul complex
on \(\alpha_0,\alpha_1,\beta_0,\beta_1\), after tensoring with \(R\), is the
exact Koszul complex on \(A_0,A_1,B_0,B_1\).  Hence
\[
  \operatorname{Tor}_q^U
  \left(R,U/(\alpha_0,\alpha_1,\beta_0,\beta_1)\right)=0
  \qquad(q>0).
\]
Therefore the filtration specializes without additional Tor correction
terms.  Thus
\[
  \operatorname{gr}M_\ell(a,b)
  \simeq
  \bigoplus_{\substack{i,j\ge0\\ i+j\le\ell-2}}
  \left(R/(A_0,A_1,B_0,B_1)\right)(-ai-bj)
  ^{\oplus(\ell-1-i-j)} .
\]
Finally,
\[
  \HS(R/(A_0,A_1,B_0,B_1);t)
  =
  \frac{(1-t^a)^2(1-t^b)^2}{(1-t)^{\dim R}},
\]
again because the four forms are a regular sequence.  Multiplying by the
triangular shift sum gives the claimed formula.
\end{proof}

\section{\texorpdfstring{Endpoint arguments for the \(\bP^2\times\bP^3\) examples}{Endpoint arguments for the P2 x P3 examples}}
\label{app:ls-7844-endpoints}
\label{app:ls-7844-left-kernel}

\paragraph{Left endpoint kernel.}
We first justify the free resolution used in
\eqref{eq:ls-7844-left-kernel-resolution}.  In the notation of
Section~\ref{subsec:sec3-larfors-p2p3}, we specialize to \(k=2\), so
\[
  F\in T_2\otimes S_2,\qquad G\in T_1\otimes S_2.
\]
For \(r\ge3\), put
\[
  M_r=\bigl((T\otimes S)/(G,F)\bigr)_{x\text{-degree }r}.
\]
The fixed \(x\)-degree piece of the Koszul resolution is
\begin{equation}
0\to
T_{r-3}\otimes S(-4)
\to
T_{r-2}\otimes S(-2)\oplus T_{r-1}\otimes S(-2)
\to
T_r\otimes S
\to
M_r
\to0 .
\label{eq:app-7844-Mr-resolution}
\end{equation}
For a general pair \((G,F)\), the above complex is exact and its cokernel is
a rank-two \(S\)-module.  The generic rank is
\[
  \binom{r+2}{2}-r^2+\binom{r-1}{2}=2.
\]
We next verify reflexivity.  Let
\(U=T_1\) and \(W=T_2\), and consider the incidence locus
\[
  \Sigma=
  \bigl(\{0\}\times W\bigr)
  \cup
  \{(\ell,q)\in U\times W:\ell\ne0,\ q\in \ell U\}.
\]
Since \(\dim U=3\) and \(\dim W=6\), both components of \(\Sigma\) have
dimension \(6\), hence \(\operatorname{codim}_{U\oplus W}\Sigma=3\).
For each fixed \(y\ne0\), the evaluation map
\[
  (G,F)\longmapsto (G_y,F_y)\in U\oplus W
\]
is surjective on the coefficient space.  A standard incidence-dimension
argument therefore shows that, for a general pair \((G,F)\), the inverse
image of \(\Sigma\) in \(\operatorname{Spec}S=\mathbb A^4\) has dimension at
most \(1\).  Outside this inverse image, \(G_y\ne0\) and
\(F_y\notin G_yT_1\), so
\((G_y,F_y)\) is a complete intersection of degrees \((1,2)\) in \(T\), and
\[
  \dim_\mathbb C(T/(G_y,F_y))_r=2
  \qquad(r\ge1).
\]
Thus \(M_r\) is locally free away from a subset of codimension at least \(3\);
equivalently, it is locally free in codimension \(2\).

Serre's \(S_2\) condition follows from a similar incidence argument applied
to the origin in \(U\oplus W\): for a general pair, \(G_y=F_y=0\) only at
\(y=0\).  Hence the left map in \eqref{eq:app-7844-Mr-resolution} can fail to
be fiberwise injective only at the origin.  Dualizing the two-step free
resolution shows that \(\operatorname{Ext}^2_S(M_r,S)\) is supported in
codimension \(4\).  After localization, Auslander--Buchsbaum over the
regular local ring \(S_{\mathfrak p}\) gives depth at least \(2\) at primes
of height \(3\) and \(4\), while local freeness covers primes of height at
most \(2\).  Since \(S\) is regular, \(M_r\) is reflexive.

Its determinant is read from the free resolution:
\[
  \det M_r
  \simeq
  S\!\left(
    2r^2-4\binom{r-1}{2}
  \right)
  =
  S(6r-4).
\]
For a rank-two reflexive module,
\[
  M_r^\vee\simeq M_r\otimes(\det M_r)^\vee,
\]
and hence
\[
  M_r^\vee\simeq M_r(-(6r-4)).
\]

Apply \(\operatorname{Hom}_S(-,S(-4))\) to
\eqref{eq:app-7844-Mr-resolution}.  The beginning of the dual complex is
\[
0\to
\operatorname{Hom}_S(M_r,S(-4))
\to
T_r^\vee\otimes S(-4)
\to
T_{r-2}^\vee\otimes S(-2)\oplus T_{r-1}^\vee\otimes S(-2).
\]
Thus the graded left-endpoint module
\[
  K_r^{(2)}
  =
  \ker\!\left(
  T_r^\vee\otimes S(-4)
  \to
  T_{r-2}^\vee\otimes S(-2)\oplus T_{r-1}^\vee\otimes S(-2)
  \right)
\]
is
\[
  \operatorname{Hom}_S(M_r,S(-4))
  =
  M_r^\vee(-4)
  \simeq
  M_r(-6r).
\]
Twisting \eqref{eq:app-7844-Mr-resolution} by \(-6r\) therefore gives
\[
0\to
S(-6r-4)^{\binom{r-1}{2}}
\to
S(-6r-2)^{r^2}
\to
S(-6r)^{\binom{r+2}{2}}
\to
K_r^{(2)}
\to0,
\]
which is \eqref{eq:ls-7844-left-kernel-resolution}.

The same reflexivity argument and determinant computation for the general
configuration \(X_k\) give
\[
  \det M_r^{(k)}
  \simeq
  S\bigl((8-k)r-4\bigr).
\]
Consequently the corresponding left-endpoint module is shifted by
\((8-k)r\).  This is the source of the main wall
\(m_1=-(8-k)m_0\) for \(k=1,2,3\).

\paragraph{Right endpoint rank checks.}
We explain the finite rank statement used in
\eqref{eq:ls-7844-right-end-numerator}.  In this paragraph
\(C_2\to C_1\to C_0\) denotes the dual row from the main text; equivalently,
with the shorter notation of Section~\ref{subsec:sec3-larfors-p2p3}, it is
\[
  D(C_0)_b\longrightarrow D(C_1)_b\longrightarrow D(C_2)_b .
\]
Since the right-endpoint map annihilates the image of the previous arrow, it
factors through
\[
  \bar\delta_b:
  Q_b:=C_1(b)/\operatorname{im}\!\left(C_2(b)\to C_1(b)\right)
  \longrightarrow C_0(b).
\]
Since
\[
  \dim Q_b=\dim C_1(b)-\dim C_2(b)+K_{r,2}(b),
  \qquad
  \dim C_0(b)-\dim Q_b=G_r(b)-K_{r,2}(b),
\]
maximal rank of \(\bar\delta_b\) gives
\[
  \dim\coker\bar\delta_b
  =
  \max\{G_r(b)-K_{r,2}(b),0\}.
\]
For \(3\le r\le15\), the verification data exhibits nonzero minors proving
this maximal-rank statement in every degree \(0\le b\le B_r\), where
\[
  B_r=\max\{b\ge0:G_r(b)-K_{r,2}(b)>0\}.
\]
It also gives a surjectivity, or full-row-rank, check in degree \(B_r+1\).
Since the right-endpoint module \(E_r\) is generated in degree zero, this proves
\((E_r)_b=0\) for all \(b\ge B_r+1\).

\paragraph{Machine-readable verification data for \(7833,7844,7883\).}
\label{app:ls-p2p3-certificates}
The finite black-line statements for \(7833\) and \(7883\) use the same
nonzero-minor principle.  The companion repository
\url{https://github.com/jtwangbimsa-dragon/cicylinecoh_alpha}
stores replay data organized by CICY number.  Each entry contains a
deterministic specialization, a prime field, and pivot/minor data for a
specific full-rank minor.  The verifier rebuilds the corresponding integer
matrix, reduces it modulo the specified prime, and replays the pivot/minor
check.  A successful replay proves that the same integer minor is nonzero, so
the rank condition holds over \(\bQ\) and on a nonempty Zariski-open subset
of the complex coefficient space.

For \(7844\), the archive contains the right-endpoint checks for
\(3\le r\le15\) and \(0\le b\le B_r+1\) described above.  For \(7833\), it
contains \(46\) right-endpoint checks on
\[
  r=2b-1,\quad 2\le b\le25,
  \qquad\text{and}\qquad
  r=2b+1,\quad 4\le b\le25.
\]
For \(7883\), it contains the boundary check at \(q=1\) and \(66\)
quotient-pivot checks for
\[
  (r,b)=(3q,2q),\qquad 2\le q\le23,
\]
three bottom layers for each \(q\).

For \(7833\) and \(7883\), the verifier exploits the special structure of the
endpoint matrices to reduce the amount of stored pivot data.

After the endpoint ranks are replayed, they are used exactly as in the main
text.  For \(7833\), the verified right endpoint vanishes on the two listed
black lines, while the left endpoint vanishes there by the shifted endpoint
numerator.  For \(7883\), the left endpoint vanishes on
\((r,b)=(3q,2q)\), the right endpoint has dimension \(1\) at \(q=1\), and
vanishes for \(2\le q\le23\).  In every case, the remaining middle
cohomology entry is then determined by the row Euler characteristic.

\section{\texorpdfstring{CICY \(7880\): hidden walls and finite-box checks}{CICY 7880: hidden walls and finite-box checks}}
\label{app:critical-rank-rays}

\paragraph{\texorpdfstring{The \(q=2\) hidden wall.}{The q=2 hidden wall.}}
\label{app:7880-q2-hidden-wall}

We prove Theorem~\ref{thm:223-q2-hidden-wall}.  Let
\[
  B=\bC[s,t],\qquad S=\bC[z_0,z_1,z_2],
  \qquad R=B\otimes_\bC S.
\]
Write the \(7880\) defining polynomial as
\[
  F=\sum_{\alpha,\beta=0}^2 x_\alpha y_\beta F_{\alpha\beta}(z),
  \qquad F_{\alpha\beta}\in S_3,
\]
and set
\[
  G_\beta=s^2F_{0\beta}+stF_{1\beta}+t^2F_{2\beta}
  \in B_2\otimes S_3,\qquad \beta=0,1,2.
\]
For \(q=2\), the transpose of the active map in \(S\)-degree \(6\)
is the \(B\)-degree \(p-2\) part of contraction by the three forms
\(G_\beta\).  Thus, if
\[
  Q=R/(G_0,G_1,G_2),
\]
then
\[
  \ker\Theta_{p,2}(6)^T=\operatorname{Hom}_B(Q_6,B)_{p-2}.
\]
Since \(\HF(M_{p,2},6)=\dim\ker\Theta_{p,2}(6)^T\), it remains to compute
the graded \(B\)-module \(\operatorname{Hom}_B(Q_6,B)\).

For coefficients \(F_{\alpha\beta}\) in a nonempty Zariski-open subset of the
full coefficient space, the three forms \(G_0,G_1,G_2\) form a regular
sequence after localization at every height-one prime of \(B\).  Taking the
\(S\)-degree \(6\) part of the Koszul complex gives, after these
localizations,
\[
  0\longrightarrow
  B(-4)^3
  \longrightarrow
  S_3^{\oplus3}\otimes B(-2)
  \longrightarrow
  S_6\otimes B
  \longrightarrow
  Q_6
  \longrightarrow0 .
\]
Hence \(Q_6\) has rank
\[
  \dim S_6-3\dim S_3+3=28-30+3=1
\]
over \(B\), and its determinant is
\[
  \det Q_6
  =
  B(60-12)
  =
  B(48).
\]
Thus \(Q_6\) is locally free of rank one in codimension one and
\(\det Q_6=B(48)\).  Hence its dual is a graded rank-one reflexive
\(B\)-module with determinant \(B(-48)\).  Since \(B=\bC[s,t]\) is a UFD,
every graded rank-one reflexive \(B\)-module is free up to shift, and hence
\[
  \operatorname{Hom}_B(Q_6,B)\simeq B(-48).
\]
Consequently
\[
  \HF(M_{p,2},6)
  =
  \dim B(-48)_{p-2}
  =
  \dim B_{p-50}
  =
  \max\{p-49,0\}.
\]
This is the Hilbert-function statement in
Theorem~\ref{thm:223-q2-hidden-wall}; the displayed cohomology formulae
then follow from the spectral-sequence row used in the main text.

\paragraph{The finite-box maximal-rank check.}

We spell out the finite-box certificate used for the top-\(\bP^2\)
chamber of CICY \(7880\).  The key observation is that, for each fixed pair
\((p,q)\), two rank checks at the critical degrees imply the positive-part
Hilbert function in all \(S\)-degrees.  Throughout this section we assume \(p,q\ge2\).

Let
\[
  S=\bC[z_0,z_1,z_2],
  \qquad
  a=(p+1)(q+1),\qquad b=(p-1)(q-1),
\]
and consider the graded map
\begin{equation}
  \Theta_{p,q}:S(-3)^a\longrightarrow S^b .
  \label{eq:app-critical-theta}
\end{equation}
In the \(7880\) application this is the convolution matrix built from the
nine cubic coefficients \(F_{\alpha\beta}\).  Put
\[
  M_{p,q}=\coker\Theta_{p,q}.
\]
In degree \(c\),
\[
  \Theta_{p,q}(c):
  S_{c-3}^{a}\longrightarrow S_c^{b},
\]
where \(S_m=0\) for \(m<0\).  Thus
\[
  \dim S_{c-3}^{a}=a\binom{c-1}{2},
  \qquad
  \dim S_c^{b}=b\binom{c+2}{2}.
\]

The ratio
\[
  \frac{\dim S_{c-3}}{\dim S_c}
  =
  \frac{(c-1)(c-2)}{(c+1)(c+2)}
\]
is increasing for \(c\ge3\).  Since \(p,q\ge2\),
\[
  \frac{a}{b}
  =
  \frac{(p+1)(q+1)}{(p-1)(q-1)}
  \le9
  <10
  =
  \frac{\dim S_3}{\dim S_0}.
\]
Thus the source is no larger than the target at \(c=3\); since \(a/b>1\),
there is a unique crossing.  Define
\begin{equation}
  R=R(p,q)
  =
  \max\left\{
  c\ge3:
  a\dim S_{c-3}\le b\dim S_c
  \right\}.
  \label{eq:app-critical-R}
\end{equation}
Then the source is no larger than the target for \(c\le R\), and no smaller
for \(c\ge R+1\).

\begin{lemma}
\label{lem:critical-two-ranks-all-degrees}
Assume that
\[
  \Theta_{p,q}(R):S_{R-3}^{a}\to S_R^b
\]
is injective and that
\[
  \Theta_{p,q}(R+1):S_{R-2}^{a}\to S_{R+1}^b
\]
is surjective.  Then \(\Theta_{p,q}(c)\) has maximal rank for every
\(c\ge0\).  Equivalently,
\[
  \HF(M_{p,q},c)
  =
  \max\left\{
  b\binom{c+2}{2}
  -
  a\binom{c-1}{2},
  0
  \right\}
\]
for all \(c\), with the convention \(\binom{c-1}{2}=0\) for \(c<3\).
\end{lemma}

\begin{proof}
If \(c<R\), any nonzero kernel vector in degree \(c\) can be multiplied by a
nonzero form of degree \(R-c\).  Since \(S\) is a domain and the source is
free, this gives a nonzero kernel vector in degree \(R\), contradicting the
assumed injectivity there.  Hence all maps in degrees \(c\le R\) are
injective.

For \(c\ge R+1\), surjectivity in degree \(R+1\) is the statement
\((M_{p,q})_{R+1}=0\).  Since \(M_{p,q}\) is a quotient of \(S^b\), it is
generated in degree \(0\), and each higher graded piece is obtained by
multiplying the previous ones by linear forms.  Therefore
\((M_{p,q})_c=0\) for every \(c\ge R+1\), so all maps in these degrees are
surjective.  This gives maximal rank in every degree, and the displayed
Hilbert function is the target dimension minus the rank.
\end{proof}

The pivot certificates are characteristic-zero certificates, not merely
finite-field evidence.  After deterministic integer coefficients are chosen
for the cubics \(F_{\alpha\beta}\), the two critical maps become integer
matrices.  If sparse elimination modulo a prime \(\ell\) finds a minor of
the required size which is nonzero in \(\mathbb F_\ell\), then that minor is
a nonzero integer, so the same rank statement holds over \(\bQ\).
Equivalently, this minor is a nonzero polynomial condition on the
coefficients of the \(F_{\alpha\beta}\), and hence holds on a nonempty
Zariski-open subset.  Intersecting the finitely many certified open
conditions with the smoothness open gives the claimed generic statement.

Therefore, for every fixed \((p,q)\) with certificates at \(R\) and \(R+1\),
Lemma~\ref{lem:critical-two-ranks-all-degrees} proves the positive-part
Hilbert function for all \(c\).  The certificates for
\[
  2\le p,q\le14
\]
cover \(169\) infinite \(c\)-rays in the top-\(\bP^2\) chamber of CICY
\(7880\), not merely \(169\) isolated line bundles.  
This statement has two important limits.  First, it is a generic statement:
special complex structures may have lower ranks and hence larger
cohomology.  Second, it is a statement in the full coefficient space of the
hypersurface.  If one restricts to an invariant or symmetric subfamily of
defining polynomials, the same minor-nonvanishing argument must be repeated
inside that smaller coefficient space.

\section{\texorpdfstring{Details for the CICY \(7707\) formulae}{Details for the CICY 7707 formulae}}
\label{app:7707-technical}

This appendix contains the proof-level material supporting
Section~\ref{subsec:new-7707}.  The companion repository containing the
computational evidence for the finite-field, standard-basis, and direct-rank
computations used in this paper is available at
\begin{center}
  \url{https://github.com/jtwangbimsa-dragon/cicylinecoh_alpha}.
\end{center}
We use the notation
\[
X_{7707}=
\left[\begin{array}{c|cc}
\bP^1_x&1&1\\
\bP^1_y&0&2\\
\bP^3_z&3&1
\end{array}\right],
\]
and
\[
  F=x_0A_0(z)+x_1A_1(z),\qquad
  G=x_0B_0(y,z)+x_1B_1(y,z).
\]

\begin{lemma}[The one-top \(\bP^1_x\) row for \(7707\)]
\label{lem:7707-p1x-row}
Let \(k_x=-\ell\le-2\).  After Serre duality on \(\bP^1_x\), the
right-endpoint cokernel of the one-top \(\bP^1_x\) row is the
\((u,v)\)-graded piece of
\[
  M_{\ell,3}
  =
  \coker\left(
  R(0,-3)^\ell\oplus R(-2,-1)^\ell
  \longrightarrow R^{\ell-1}\right),
  \qquad R=\bC[y_0,y_1,z_0,\ldots,z_3],
\]
where \(u\) is the \(\bP^1_y\)-degree and \(v\) is the \(\bP^3_z\)-degree.
For a general member,
\begin{equation}
  \HS(M_{\ell,3};u,v)
  =
  \frac{(1-v^3)^2(1-u^2v)^2}{(1-u)^2(1-v)^4}
  \sum_{\substack{i,j\geq0\\ i+j\leq \ell-2}}
  (\ell-1-i-j)u^{2j}v^{3i+j}.
  \label{eq:7707-p1x-HS}
\end{equation}
\end{lemma}

\begin{proof}
The two columns of \(F,G\) have degrees in the remaining variables \((0,3)\) and
\((2,1)\) on \(\bP^1_y\times\bP^3_z\).  After taking top cohomology on
\(\bP^1_x\), the right-endpoint cokernel is the adjacent two-band module of
Appendix~\ref{app:two-linear-band-module}.  Applying
Lemma~\ref{lem:two-linear-band-HS} with
\(\alpha=(0,3)\) and \(\beta=(2,1)\) gives
\eqref{eq:7707-p1x-HS}.
\end{proof}

\begin{lemma}[The one-top \(\bP^1_y\) row]
\label{lem:7707-p1y-row}
Let
\[
  B=\{x_0A_0(z)+x_1A_1(z)=0\}
  \subset \bP^1_x\times\bP^3_z .
\]
Write
\[
  G=C_0(x,z)y_0^2+C_1(x,z)y_0y_1+C_2(x,z)y_1^2
\]
on \(B\times\bP^1_y\), where each \(C_i\) is a section of
\(\cO_B(1,1)\).  For \(k_y=-\ell\le-2\), the one-top
\(\bP^1_y\) row is the corresponding binary-quadratic band module on
\(B\):
\[
  \mathcal T_{\ell,F}
  =
  \coker\!\left(
  \cO_B(-1,-1)^{\ell+1}
  \longrightarrow
  \cO_B^{\ell-1}\right).
\]
For a general member, this sheaf is resolved by
\begin{equation}
\begin{aligned}
0\to{}&
\cO_B(-\ell-1,-\ell-1)^{\ell-1}
\to
\cO_B(-\ell,-\ell)^{\ell+1} \\
\to{}&
\cO_B(-1,-1)^{\ell+1}
\to
\cO_B^{\ell-1}
\to
\mathcal T_{\ell,F}
\to0 .
\end{aligned}
\label{eq:7707-p1y-resolution}
\end{equation}
Thus every one-top \(\bP^1_y\) rank defect is the Hilbert function of
\(\mathcal T_{\ell,F}\), twisted by the remaining \((x,z)\)-degree dictated
by the range.
\end{lemma}

\begin{proof}
On \(B\), the row is multiplication by a binary quadratic
\(C_0+C_1w+C_2w^2\).  Assume \(B\) is smooth and that the restrictions
\(C_0,C_1,C_2\in H^0(B,\cO_B(1,1))\) are general.  Since
\(\cO_B(1,1)\) is base-point free, the common zero locus
\[
  Z=(C_0=C_1=C_2)\subset B
\]
has the expected codimension \(3\).  The Toeplitz map
\[
  \cO_B(-1,-1)^{\ell+1}\longrightarrow \cO_B^{\ell-1}
\]
has maximal rank away from \(Z\), and its rank-drop locus is precisely
\(Z\).  Since \(B\) is Cohen--Macaulay, the Buchsbaum--Rim exactness
criterion recalled in Appendix~\ref{app:buchsbaum-rim} gives
\eqref{eq:7707-p1y-resolution}.  Taking the required \((x,z)\)-graded pieces
gives the one-top \(\bP^1_y\) rank defects.
\end{proof}

\begin{theorem}[The one-top \(\bP^3_z\) row for \(7707\)]
\label{thm:7707-p3-row}
Let \(W=\bC[z_0,\ldots,z_3]\) and
\(U=\bC[x_0,x_1,y_0,y_1]\), with \(U\) bigraded by \((x,y)\)-degree.
Put
\[
  P_3(m)=
  \begin{cases}
    \binom{m+3}{3},&m\ge0,\\
    0,&m<0.
  \end{cases}
\]
For \(m\in \bZ\), define
\[
  N_m=\left((W\otimes U)/(F,G)\right)_{z\text{-degree}=m}.
\]
Then \(N_m=0\) for \(m<0\), while for \(m\ge0\) it has the \(U\)-free
resolution
\begin{equation}
0\to
W_{m-4}\otimes U(-2,-2)
\to
W_{m-3}\otimes U(-1,0)\oplus W_{m-1}\otimes U(-1,-2)
\to
W_m\otimes U
\to
N_m\to0 ,
\label{eq:7707-p3-resolution}
\end{equation}
where \(W_j=0\) for \(j<0\).  Hence
\begin{equation}
  \HS_U(N_m;u,v)
  =
  \frac{
  P_3(m)-P_3(m-3)u-P_3(m-1)uv^2+P_3(m-4)u^2v^2
  }{(1-u)^2(1-v)^2}.
  \label{eq:7707-p3-HS}
\end{equation}
For the top-\(\bP^3_z\) ambient cohomology rows used in the \(7707\) audit,
write the original \(z\)-degree as \(k_z=-m\).  For a Koszul summand with
\(z\)-shift \(d\), Serre duality gives
\[
  H^3(\bP^3,\cO_{\bP^3}(-m-d))\simeq W_{m+d-4}^{\vee},
\]
with \(W_j=0\) for \(j<0\).  Thus Serre duality replaces the fixed-degree
piece by the graded complex
\begin{equation}
\begin{aligned}
0\to&
  W_m^\vee\otimes U(-2,-2)
\\
\to&
  W_{m-3}^\vee\otimes U(-1,-2)
  \oplus
  W_{m-1}^\vee\otimes U(-1,0)
\\
\to&
  W_{m-4}^\vee\otimes U
\to0 .
\end{aligned}
\label{eq:7707-p3-dual-complex}
\end{equation}
Thus the row contribution is the corresponding remaining \((x,y)\)-graded
cohomology of \eqref{eq:7707-p3-dual-complex}.  If
\[
  P_1(n)=\max\{n+1,0\},
\]
then the bidegree in the remaining variables \((a,b)\) piece is
\[
\begin{aligned}
0\to&
  W_m^\vee\otimes U_{a-2,b-2}
\\
\to&
  W_{m-3}^\vee\otimes U_{a-1,b-2}
  \oplus
  W_{m-1}^\vee\otimes U_{a-1,b}
\\
\to&
  W_{m-4}^\vee\otimes U_{a,b}
\to0,
\end{aligned}
\]
where \(\dim U_{p,q}=P_1(p)P_1(q)\) and \(W_j=0\) for \(j<0\).  These finite
complexes are the formula rules for the one-top \(\bP^3_z\) entries;
closed numerical formulae require the separate finite-rank certificates.
\end{theorem}

\begin{proof}
For a general member, \(F,G\) form a regular sequence, so their two-step
Koszul resolution remains exact after taking fixed \(z\)-degree \(m\), giving
\eqref{eq:7707-p3-resolution} and the Hilbert series above.  When the
\(\bP^3_z\)-factor is in top cohomology, Serre duality identifies the
row with the dual fixed-degree piece, namely
\eqref{eq:7707-p3-dual-complex}.  Taking the required remaining
\((x,y)\)-graded piece gives the displayed finite complex.
\end{proof}

\paragraph{\texorpdfstring{\(k=(-a,-b,m)\), \(m=1,2\).}{k=(-a,-b,m), m=1,2.}}
Put \(S=\bP^1_x\times\bP^1_y\) and \(L=\cO_S(1,2)\).  The \(z\)-linear
equation can be written as
\[
  G=\sum_{i=0}^3 z_iG_i(x,y),\qquad G_i\in H^0(S,L),
\]
and gives, on a nonempty open locus, the map
\[
  \Gamma_G:\cO_S^4\longrightarrow L,
  \qquad
  \Gamma_G(f_0,\ldots,f_3)=\sum_{i=0}^3 f_iG_i,
\]
with kernel sequence
\[
0\to M\to \cO_S^4\xrightarrow{\Gamma_G}L\to0 .
\]
On this open locus the four \(G_i\) have no common zero on \(S\), so
\(\Gamma_G\) is surjective.
For \(D=(a-2,b-2)\), the \(m=1\) map is the induced map on global sections
after twisting by \(\cO_S(D)\).  For \(m=2\), put
\(\mathcal F=\operatorname{Sym}^2M\).  The active map is induced by the polarized
contraction
\[
  uv\longmapsto u\otimes \Gamma_G(v)+v\otimes \Gamma_G(u),
\]
equivalently by
\[
0\to \mathcal F\to \operatorname{Sym}^2\bC^4\otimes\cO_S
\to \bC^4\otimes L\to0 .
\]
Indeed, on each fibre one may choose a lift of the one-dimensional quotient
of the surjection \(E=\bC^4\to L_s\), so that
\(E\simeq \ker(\Gamma_G)_s\oplus \widetilde L_s\).  In this splitting the
displayed contraction kills exactly the
\(\operatorname{Sym}^2(\ker(\Gamma_G)_s)\) summand and is surjective onto
\(E\otimes L_s\).
If \(\phi_m(D)\) denotes the resulting global-section map after twisting by
\(\cO_S(D)\), then the rank defect appearing in the Koszul row is
\[
  q_m(a,b)=\dim\coker\phi_m(D).
\]
For the range used in Theorem~\ref{thm:7707-range-m12-analytic}, namely
\(a,b\ge2\), the ambient bundles occurring in the two displayed sequences
have no \(H^1\).  Hence the cokernels of the global-section maps are
identified with
\[
  \coker\phi_1(D)\simeq H^1(S,M(D)),
  \qquad
  \coker\phi_2(D)\simeq H^1(S,\mathcal F(D)).
\]
The required dimensions reduce to line-bundle cohomology on \(\bP^1\) from
the pushforward splittings below; here \(p_x\) and \(p_y\) denote the
projections to \(\bP^1_x\) and \(\bP^1_y\), respectively:
\[
  p_{y*}M\simeq \cO_y(-2)^2,
  \qquad
  p_{x*}M\simeq \cO_x(-3),
\]
\[
  p_{y*}\mathcal F\simeq \cO_y(-4)^3,
  \qquad
  R^1p_{y*}\mathcal F\simeq \cO_y(4),
\]
\[
  p_{x*}\mathcal F\simeq \cO_x(-6),
  \qquad
  R^1p_{x*}\mathcal F\simeq \cO_x(2)^3,
\]
and
\[
  p_{x*}(\mathcal F(0,1))\simeq \cO_x(-4)^4,
  \qquad
  R^1p_{x*}(\mathcal F(0,1))=0.
\]
These splittings give the Hilbert numerators displayed in
Theorem~\ref{thm:7707-range-m12-analytic}.  Nonemptiness of the splitting
locus is certified at the deterministic seed modulo \(1000003\): the
certificate contains enough full-rank fibre and global-section minors to fix
the dimensions of the pushforwards and their low twists.  On \(\bP^1\) those
dimensions determine the displayed splitting types.  The same minors are
nonzero integer polynomials in the coefficients, so the splittings hold on a
nonempty characteristic-zero open set.

\paragraph{\texorpdfstring{\(k=(-a,-b,3)\) on finite-width strips.}{k=(-a,-b,3) on finite-width strips.}}
The active one-top \(\bP^3_z\) row is
\[
\phi_3:
W_3^\vee\otimes U(-2,-2)
\longrightarrow
W_0^\vee\otimes U(-1,-2)\oplus W_2^\vee\otimes U(-1,0),
\]
where the two components are induced by the \((1,0,3)\) and \((1,2,1)\)
equations.  Thus
\[
  q_3(a,b)=\dim\coker(\phi_3)_{a,b},
\]
and Euler characteristic gives
\[
  h^\bullet(X,\cO_X(-a,-b,3))
  =(0,q_3,q_3+\chi_3,0).
\]
For fixed \(a\), the map becomes a one-variable graded presentation over
\(R_y=\bF_p[y_0,y_1]\),
\[
  R_y(-2)^{20(a-1)}
  \longrightarrow
  R_y(-2)^a\oplus R_y^{10a},
\]
and for fixed \(b\), it becomes a presentation over
\(R_x=\bF_p[x_0,x_1]\),
\[
  R_x(-2)^{20(b-1)}
  \longrightarrow
  R_x(-1)^{11b+9}.
\]
The strip certificate constructs these presentations for
\(2\le a\le10\) and \(2\le b\le10\) at the seed-\(7707\) specialization
modulo \(1000003\).  Singular computes standard bases and Hilbert numerators
for the \(18\) cokernel modules; the checker compares the resulting Hilbert
functions with the stated strip formula in every chamber degree of the free
variable.  The certificates record the leading terms, reduction identities,
and pivot minors used in these standard-basis computations; their leading
coefficients are nonzero at the seed, so the same Hilbert functions hold on a
nonempty characteristic-zero open set.

\paragraph{\texorpdfstring{\(k=(0,-b,c)\) in the finite \(b\)-box.}{k=(0,-b,c) in the finite b-box.}}
The only nonzero second-page map is the moving-resultant map
\[
  \Phi_{b,c}:R_{c-4}^{b+1}\longrightarrow R_c^{b-1},
  \qquad R=\bC[z_0,\ldots,z_3],
\]
represented by the Toeplitz matrix
\[
  \begin{pmatrix}
  f_0&f_1&f_2&0&\cdots&0\\
  0&f_0&f_1&f_2&\cdots&0\\
  \vdots&&\ddots&\ddots&\ddots&\vdots\\
  0&\cdots&0&f_0&f_1&f_2
  \end{pmatrix},
\]
where
\[
  \Delta=f_0y_0^2+f_1y_0y_1+f_2y_1^2,\qquad f_i\in R_4,
\]
is the \(x\)-resultant.
In this range,
\[
  h^0=\dim\ker\Phi_{b,c},
  \qquad
  h^1=\dim\coker\Phi_{b,c},
  \qquad h^2=h^3=0.
\]
The source has dimension \((b+1)P_3(c-4)\) and the target has dimension
\((b-1)P_3(c)\), so Euler characteristic determines \(h^0\) once
\(h^1\) is known.  The finite \(b\)-box verifier constructs
\(M_b=\coker\Phi_b\) for every \(2\le b\le10\) over the seed modulo
\(1000003\), computes
\[
  \HS(M_b;t)=\frac{N_b(t)}{(1-t)^4},
\]
and checks the recurrence for \(N_b\) against direct finite-field ranks in
the initial chamber degrees.  The recorded leading terms, reductions, and
pivot minors define nonempty open conditions, so the Hilbert series and the
displayed cohomology formula hold on a nonempty characteristic-zero open set.

\paragraph{\texorpdfstring{\(k=(-a,-b,c)\) in the finite \((a,b)\)-box.}{k=(-a,-b,c) in the finite (a,b)-box.}}
After Serre duality on the two \(\bP^1\) factors, the active row is
\[
0\to
  W_c^\vee\otimes U_{a-2,b-2}
\xrightarrow{d_0}
  W_{c-3}^\vee\otimes U_{a-1,b-2}
  \oplus
  W_{c-1}^\vee\otimes U_{a-1,b}
\xrightarrow{d_1}
  W_{c-4}^\vee\otimes U_{a,b}
\to0 .
\]
Here \(W=\bC[z_0,\ldots,z_3]\) and
\(U_{p,q}=H^0(\bP^1_x\times\bP^1_y,\cO(p,q))\).  The endpoint terms control
the cohomology:
\[
  h^0=\dim\coker(d_1),
  \qquad
  h^2=\dim\ker(d_0),
\]
and \(h^1\) is determined by the Euler characteristic.  For the right
endpoint, using the notation \(F=x_0A_0+x_1A_1\) and
\(G=x_0B_0+x_1B_1\) from the main text, the Sylvester
identities
\[
  B_1F-A_1G=x_0\Delta,
  \qquad
  -B_0F+A_0G=x_1\Delta
\]
identify \(\coker(d_1)\) in \(x\)-degree \(a\) with the remaining determinant
module of the finite \(b\)-box, shifted by \(z\)-degree \(3a\).  Hence
\[
  \sum_{c\ge4}\dim\coker(d_1)_c\,t^c
  =
  \frac{t^{3a}R_b(t)}{(1-t)^4}.
\]
For the left endpoint, dualizing \(d_0\) gives
\[
  \dim\ker(d_0)_c=\dim\coker(d_0^\vee)_c.
\]
For all \(81\) pairs \(2\le a,b\le10\), the endpoint certificate constructs
the graded image of \(d_0^\vee\), computes a Singular standard basis, and
obtains
\[
  \HS(\coker(d_0^\vee);t)=\frac{Q_{a,b}(t)}{(1-t)^4}.
\]
The checker compares these endpoint Hilbert functions with the stated
piecewise formula for \(q_{\mathrm L}(a,b,c)\) on the finite chamber and checks
additional direct finite-field rank matches, including Serre-dual samples.
As before, the recorded leading terms, reductions, and modular pivot minors
give a generic characteristic-zero certificate.

\paragraph{Boundary, \(d_2\), and certificate closure.}
Boundary rows not covered by the one-top modules introduce no new rank model:
they are either Serre-dual to ranges already treated, saturated
zero-dimensional complete intersections in the remaining variables of length \(10\), or
irrelevant-span rows whose positive-degree
quotients vanish.  The only possible second-page maps are the entries
listed in the ancillary data: the source-zero, target-zero, and
split-boundary entries vanish, while the determinant entries are
multiplication by
\[
  \Delta=A_0B_1-A_1B_0 .
\]
Thus their rank defect in remaining \((y,z)\)-degree \((b,c)\) is
\[
  \dim (R_{yz}/(\Delta))_{(b,c)},\qquad
  \HS(R_{yz}/(\Delta);u,v)
  =
  \frac{1-u^2v^4}{(1-u)^2(1-v)^4}.
\]
Since the Koszul complex has length two, there are no higher differentials
after \(d_2\).  All rank, saturation, and nonvanishing conditions used in the
\(7707\) audit are finitely many minor, leading-coefficient, or
standard-basis reduction conditions, so the certificate data place them on a
common nonempty Zariski open subset of smooth members.  For one-top
\(\bP^3_z\) ranges outside the displayed
\(7707\) formula ranges, the rule remains the finite complex
\eqref{eq:7707-p3-dual-complex}; closed numerical formulae in any
finite region require the corresponding finite-rank certificate.

\section{\texorpdfstring{Details for the CICY \(7863\) formulae}{Details for the CICY 7863 formulae}}
\label{app:7863-technical}

This appendix contains the proof-level material behind
Section~\ref{subsec:new-7863}.  Write
\[
  A=\bC[x_0,\ldots,x_3],\qquad
  S=\bC[y_0,\ldots,y_3],\qquad
  R=A\otimes S,
\]
and let \(Q,L_1,L_2\) have bidegrees \((2,2),(1,1),(1,1)\).

\paragraph{The mixed Koszul row.}
The Koszul shifts are
\[
  (0,0),\qquad (2,2),\qquad (1,1),\qquad (1,1),
\]
so the relevant mixed row is obtained by taking ambient cohomology of the
twisted Koszul complex
\[
0\to\cO(a-4,b-4)
\to \cO(a-3,b-3)^{\oplus2}\oplus\cO(a-2,b-2)
\]
\[
\to \cO(a-2,b-2)\oplus\cO(a-1,b-1)^{\oplus2}
\to \cO(a,b)\to\cO_X(a,b)\to0.
\]
When the \(x\)-factor is in top cohomology, put \(m=-a-4\) and dualize by
Serre duality.  The full Cell \(9\) row is then
\begin{equation}
\begin{aligned}
A_{m+4}\otimes S_{b-4}
&\longrightarrow
(A_{m+3}\otimes S_{b-3})^{\oplus2}
  \oplus A_{m+2}\otimes S_{b-2}  \\
&\longrightarrow
A_{m+2}\otimes S_{b-2}
  \oplus (A_{m+1}\otimes S_{b-1})^{\oplus2}
\longrightarrow
A_m\otimes S_b .
\end{aligned}
\label{eq:7863-operator-mixed-row}
\end{equation}
The maps are induced by contraction with \(Q,L_1,L_2\) in the \(x\)-variables
and multiplication in the \(y\)-variables.  The homology of this row is
denoted by \(\mathcal H_i(m,b)\) in the main text.

\paragraph{One-parameter endpoint rows.}
We now specialize \eqref{eq:7863-operator-mixed-row} to the one-parameter
boundary rows used in the main table.  The variable \(r\) below is the local
endpoint degree on that boundary ray.  Thus \(r=m=-a-4\) on the
\(a\le -4\) rays, while after exchanging the two \(\bP^3\)-factors or applying
Serre duality it becomes the corresponding shifted coordinate listed in the
main table.  Put
\[
  V_y=H^0(\bP^3_y,\cO(1)).
\]
For \(j=2,3\), and with the convention
\(\operatorname{Sym}^dV_y^\vee=0\) for \(d<0\), the active row has the form
\[
  C^{(j)}_2(r)\longrightarrow C^{(j)}_1(r)
  \longrightarrow C^{(j)}_0(r),
\]
where
\[
\begin{aligned}
C^{(j)}_2(r)
&=
\left(A_{r+3}\otimes\operatorname{Sym}^{j-3}V_y^\vee\right)^{\oplus2}
\oplus A_{r+2}\otimes\operatorname{Sym}^{j-2}V_y^\vee,\\
C^{(j)}_1(r)
&=
A_{r+2}\otimes\operatorname{Sym}^{j-2}V_y^\vee
\oplus
\left(A_{r+1}\otimes\operatorname{Sym}^{j-1}V_y^\vee\right)^{\oplus2},\\
C^{(j)}_0(r)
&=
A_r\otimes\operatorname{Sym}^jV_y^\vee .
\end{aligned}
\]
For Cell \(7\) one has \(j=2\) and \(r=-a-4\); for Cell \(8\) one has
\(j=3\) and \(r=-a-4\).  The \(x/y\)-swapped cells and the Serre-dual cells
use the same endpoint calculation with \(r\) as stated in the main table.

The endpoint contribution is the cokernel of the rightmost map
\[
  \beta^{(j)}_r:C^{(j)}_1(r)\longrightarrow C^{(j)}_0(r).
\]
All vector spaces here are finite-dimensional, so
\[
  \dim\coker\beta^{(j)}_r=\dim\ker\left(\beta^{(j)}_r\right)^\vee .
\]
Thus we compute the cokernel by dualizing the map and computing the kernel
of the dual map.  Write
\[
  L_\nu=\sum_{i=0}^3 y_i\ell_{\nu i}(x)\quad(\nu=1,2),
  \qquad
  Q=\sum_{0\le i\le h\le3}y_iy_hq_{ih}(x),
\]
with \(\ell_{\nu i}\in A_1\) and \(q_{ih}\in A_2\).  If \(\iota_i\)
denotes contraction by \(y_i\), the dual endpoint map
\(\alpha^{(j)}_r=(\beta^{(j)}_r)^\vee\) is
\[
\alpha^{(j)}_r:
A_r\otimes\operatorname{Sym}^jV_y^\vee
\longrightarrow
\left(A_{r+1}\otimes\operatorname{Sym}^{j-1}V_y^\vee\right)^{\oplus2}
\oplus A_{r+2}\otimes\operatorname{Sym}^{j-2}V_y^\vee ,
\]
\[
  f\otimes\eta\longmapsto
  \left(
  \sum_i \ell_{1i}f\otimes\iota_i\eta,\,
  \sum_i \ell_{2i}f\otimes\iota_i\eta,\,
  \sum_{i\le h} q_{ih}f\otimes\iota_i\iota_h\eta
  \right),
\]
with the evident symmetric convention in the last summand.  We use the
divided-power normalization on symmetric powers; equivalently, in ordinary
monomial bases the same formula holds up to the standard nonzero scalar
factors, which are included in the matrices used in the certificates.  This
is just the transpose rule: multiplication by \(y_i\) on symmetric powers
dualizes to contraction \(\iota_i\), and multiplication by a quadratic
monomial dualizes to the corresponding double contraction.  Thus the
endpoint cokernel has dimension
\[
  N_j(r)=\dim\ker\alpha^{(j)}_r .
\]
The other endpoint of this three-term row is not determined by Euler
characteristic alone; the needed input is left-endpoint injectivity.  In the
\(j=2\) row the source is just \(A_{r+2}\), coming from the Koszul summand
\(L_1\wedge L_2\).  The left map sends \(v\in A_{r+2}\) to the two products
\[
  (L_2v,\,-L_1v)
  \in
  \left(A_{r+1}\otimes V_y^\vee\right)^{\oplus2}.
\]
Here multiplication by \(L_\nu=\sum_i y_i\ell_{\nu i}(x)\) means the four
components \(\ell_{\nu i}v\).  Since \(A\) is an integral domain and a
general \(L_\nu\) is nonzero, \(L_\nu v=0\) forces \(v=0\).  Thus the
\(j=2\) left endpoint is injective.  In the \(j=3\) row the source has three
parts.  Write an element as \((u_1,u_2,v)\), where \(u_1,u_2\) come from the
two \(Q\wedge L_\nu\) summands and \(v\) comes from \(L_1\wedge L_2\).  The
kernel equations are, up to the Koszul signs,
\[
  L_1u_1+L_2u_2=0,\qquad
  Qu_1+L_2v=0,\qquad
  Qu_2-L_1v=0.
\]
For a fixed value of \(r\), these equations are an ordinary finite matrix
after monomial bases are chosen.  Injectivity means exactly that this matrix
has full column rank.  The endpoint package checks this full-column-rank
condition on an explicit specialization; equivalently, one maximal minor is
nonzero there, so the same injectivity holds on a nonempty Zariski-open set
of triples \((Q,L_1,L_2)\).  Thus, for a general member, the leftmost map
\[
  C^{(j)}_2(r)\longrightarrow C^{(j)}_1(r)
\]
is injective, so \(H_2(C^{(j)}_\bullet(r))=0\).  Therefore, once the right
endpoint \(N_j(r)=\dim H_0(C^{(j)}_\bullet(r))\) is known, the middle
homology is forced by the Euler characteristic of the three-term complex:
\[
  \dim H_1(C^{(j)}_\bullet(r))
  =
  N_j(r)-\bigl(\dim C^{(j)}_0(r)-\dim C^{(j)}_1(r)+\dim C^{(j)}_2(r)\bigr).
\]
With \(P(n)=\dim A_n\), this gives
\[
  \dim H_1(C^{(2)}_\bullet(r))=N_2(r)+8P(r+1)-10P(r),
\]
and
\[
  \dim H_1(C^{(3)}_\bullet(r))
  =
  N_3(r)+20P(r+1)-20P(r)-2P(r+3).
\]
For the Hilbert-series computation we regard the maps \(\alpha^{(j)}_r\), for
all \(r\ge0\), as the degree pieces of one graded module map over
\(A=\mathbb C[x_0,\ldots,x_3]\).  For \(j=2\), after choosing monomial bases
of \(\operatorname{Sym}^2V_y^\vee\), \(V_y^\vee\), and \(A_r\), the two
\(L_\nu\)-blocks have entries \(\ell_{\nu i}(x)\) and the \(Q\)-block has
entries \(q_{ih}(x)\).  We compute
\[
  K_2=\ker\alpha^{(2)}
\]
as the syzygy module of this graded matrix.  For the rational specialization
recorded in the repository, Singular gives the minimal free resolution
\[
  0\longrightarrow A(-9)^{\oplus8}
  \longrightarrow A(-8)^{\oplus10}
  \longrightarrow K_2\longrightarrow0,
\]
and the certificate records the corresponding leading terms and reduction
identities.  Their leading coefficients are nonzero at the specialization;
hence the same initial module, and therefore the same Hilbert series, holds
on a nonempty Zariski-open subset of the parameter space.  Thus, for a
general member,
\[
\HS(K_2;t)=\frac{10t^8-8t^9}{(1-t)^4}.
\]
Thus \(N_2(r)=\dim (K_2)_r\).  For \(K_3=\ker\alpha^{(3)}\), we use the same
split-linear construction.  A small-integer rational specialization gives a
standard-basis replay over \(\mathbb Q\), again with recorded leading terms
and reduction identities.  Singular gives the minimal free resolution
\[
  0\longrightarrow
  A(-16)^{\oplus4}\oplus A(-17)^{\oplus2}
  \longrightarrow
  A(-15)^{\oplus20}\oplus A(-16)^{\oplus4}
  \longrightarrow
  A(-14)^{\oplus20}
  \longrightarrow K_3\longrightarrow0.
\]
The same open-condition argument gives, for a general member,
\[
\HS(K_3;t)=\frac{20t^{14}-20t^{15}+2t^{17}}{(1-t)^4}.
\]
These are the \(N_2\) and \(N_3\) series used in the table.  The other
nonzero entry of the cohomology vector is then fixed by the Euler
characteristic of the three-term row, giving the displayed additions by
\(\Delta_2\) and \(\Delta_3\).

\paragraph{Finite-box rows.}
We now prove the finite-box theorem for Cell \(9\).  Write \(a=-m-4\),
\(m\ge0\).  The reduced object is the
four-term complex
\[
\begin{aligned}
C_3&=A_{m+4}\otimes S_{b-4},\\
C_2&=(A_{m+3}\otimes S_{b-3})^{\oplus2}
  \oplus A_{m+2}\otimes S_{b-2},\\
C_1&=A_{m+2}\otimes S_{b-2}
  \oplus (A_{m+1}\otimes S_{b-1})^{\oplus2},\\
C_0&=A_m\otimes S_b .
\end{aligned}
\]
The differentials are those of \eqref{eq:7863-operator-mixed-row}, induced
by \(L_1,L_2,Q\).  Its Euler characteristic is
\begin{equation}
  E(m,b)=
  P(m+4)P(b-4)-2P(m+3)P(b-3)
  +2P(m+1)P(b-1)-P(m)P(b).
  \label{eq:7863-mixed-euler}
\end{equation}

The right endpoint is the map \(C_1\to C_0\).  After dualizing, its cokernel
dimension is the kernel dimension of
\[
A_m\otimes\operatorname{Sym}^b(\bC^4_y)^\vee
\longrightarrow
(A_{m+1}\otimes\operatorname{Sym}^{b-1}(\bC^4_y)^\vee)^{\oplus2}
\oplus A_{m+2}\otimes\operatorname{Sym}^{b-2}(\bC^4_y)^\vee .
\]
We denote this right-endpoint Hilbert function by \(N_b(m)\).  For fixed
\(b\), the maps above are the
degree-\(m\) pieces of one graded \(A=\bC[x_0,\ldots,x_3]\)-linear map
\[
  \alpha_b:
  A\otimes\operatorname{Sym}^b(\bC^4_y)^\vee
  \longrightarrow
  \bigl(A(1)\otimes\operatorname{Sym}^{b-1}(\bC^4_y)^\vee\bigr)^{\oplus2}
  \oplus
  A(2)\otimes\operatorname{Sym}^{b-2}(\bC^4_y)^\vee ,
\]
with \(A(r)_m=A_{m+r}\).  Hence
\[
  N_b(m)=\dim(\ker\alpha_b)_m,
\]
and the required data are exactly the Hilbert numerator of the graded module
\(\ker\alpha_b\).

This numerator is not obtained by extrapolating from finitely many values of
\(m\).  Let \(M=R/(L_1,L_2)\), and write \(M_{s,b}\) for its bidegree
\((s,b)\) part.  The same reduction gives
\[
  N_b(m)=0\qquad\text{for }m<6b-4.
\]
For \(m=6b-4+s\), \(s\ge0\), it gives
\[
  N_b(6b-4+s)
  =
  \delta_{s,0}P(b-4)
  +
  \dim\coker\left(M_{s-2,b-2}\xrightarrow{\;Q\;}M_{s,b}\right),
  \qquad s\ge0,
\]
with the convention that negative bidegrees contribute zero.  On the open
locus where
\[
  (L_1,L_2):Q=(L_1,L_2),
\]
the quadratic row is a nonzerodivisor on \(M\), so the displayed map is
injective in every degree.  Since \(L_1,L_2\) form a regular sequence of
bidegree \((1,1)\),
\[
  \operatorname{HS}_M(u,v)
  =
  \frac{(1-uv)^2}{(1-u)^4(1-v)^4}.
\]
Taking the coefficient of \(v^b\) gives
\[
  \operatorname{HS}(M_{*,b};t)
  =
  \frac{P(b)-2P(b-1)t+P(b-2)t^2}{(1-t)^4}.
\]
Therefore
\[
\begin{aligned}
\sum_{s\ge0}N_b(6b-4+s)t^s
&=
P(b-4)+\operatorname{HS}(M_{*,b};t)
-t^2\operatorname{HS}(M_{*,b-2};t)\\
&=
\frac{
n_{b,0}+n_{b,1}t+n_{b,2}t^2+n_{b,3}t^3
}{(1-t)^4},
\end{aligned}
\]
where
\[
\begin{aligned}
n_{b,0}&=P(b)+P(b-4),\\
n_{b,1}&=-2P(b-1)-4P(b-4),\\
n_{b,2}&=6P(b-4),\\
n_{b,3}&=2P(b-3)-4P(b-4).
\end{aligned}
\]

Equivalently, for \(b=4,\ldots,10\),
\begin{equation}
  \sum_{m\ge0}N_b(m)t^m
  =
  \frac{t^{6b-4}
  \left(n_{b,0}+n_{b,1}t+n_{b,2}t^2+n_{b,3}t^3\right)}
  {(1-t)^4},
  \label{eq:7863-endpoint-box-numerators}
\end{equation}
with coefficients
\[
\begin{array}{c|rrrr}
b & n_{b,0} & n_{b,1} & n_{b,2} & n_{b,3}\\
\hline
4 & 36 & -44 & 6 & 4\\
5 & 60 & -86 & 24 & 4\\
6 & 94 & -152 & 60 & 0\\
7 & 140 & -248 & 120 & -10\\
8 & 200 & -380 & 210 & -28\\
9 & 276 & -554 & 336 & -56\\
10 & 370 & -776 & 504 & -96.
\end{array}
\]
The accompanying endpoint certificates replay this quotient-rank calculation:
for each \(b=4,\ldots,10\), they rebuild the relevant matrices over
\(\mathbb F_{1000003}\), check the stored hashes, leading terms, reduction
identities, and pivot minors, and recover the numerator in
\eqref{eq:7863-endpoint-box-numerators}.

The next input is the left end of the reduced Koszul complex.  For every
\(b=4,\ldots,10\), the reduced leftmost map corresponding to
\[
  d_3:C_3\longrightarrow C_2
\]
has a full-column-rank certificate.  The
certificate lists pivot rows for a square minor using all columns and
checks that its determinant is nonzero modulo \(1000003\).  Hence the
leftmost map is injective on the certified open set.

For each \(b=4,\ldots,10\), we rebuild all relevant matrix blocks in the
reduced Koszul complex, compute their exact ranks, verify that adjacent
compositions vanish, and then compute the middle homology directly from
\(\dim H=\dim\ker-\dim\operatorname{im}\).  The result is that no
simultaneous middle correction remains, and
\[
  \mathcal H_3=0,\qquad
  \mathcal H_0=N_b,\qquad
  \mathcal H_1=[E+N_b]_+,\qquad
  \mathcal H_2=[-E-N_b]_+ .
\]

All matrices above have entries in the integer polynomial ring generated by
the coefficients of \(Q,L_1,L_2\).  A pivot minor or standard-basis leading
coefficient which is nonzero modulo \(1000003\) is the reduction of a
nonzero integer polynomial in those coefficients.  Thus each recorded rank
or initial-module condition defines a nonempty Zariski-open condition in
characteristic zero.  Intersecting the finitely many opens from the endpoint
numerator, left injectivity, and middle-rank certificates with the smoothness
open for the complete intersection gives the general smooth member claimed in
Theorem~\ref{thm:7863-cell9-finite-box}.  Cell \(73\) follows by
Calabi--Yau Serre duality.

\bibliographystyle{unsrt}
\bibliography{references}

\end{document}